\documentclass[a4paper,11pt]{article}
\pdfoutput=1
\usepackage{jheppub} % for details on the use of the package, please
\usepackage{caption}
\usepackage{tkz-tab}

\usepackage{amsfonts,esint}

\usepackage{mathrsfs}
\usepackage{comment}
\usepackage{graphicx,subcaption}

\usepackage{epstopdf}

\usepackage{amsmath,amssymb,bbm,amsthm}

\usepackage{array} 
\usepackage{graphicx}

\usepackage{hyperref}

\newtheorem{theorema}{Theorem}

\newtheorem{theorem}{Theorem}[section]
\newtheorem{corollary}{Corollary}[theorem]

\newtheorem{proposition}[theorem]{Proposition}

\DeclareMathOperator{\Tr}{Tr}

\def\be {\begin{equation}}
\def\ee {\end{equation}}
\def\bea {\begin{eqnarray}}
\def\eea {\end{eqnarray}}

\allowdisplaybreaks

\begin{document}
\begin{titlepage}

\begin{center}

 {\LARGE\bfseries  Weyl metrics and Wiener-Hopf factorization 
 }
\\[10mm]

\textbf{P. Aniceto, M.C. C\^amara, G.L.~Cardoso and M.~Rossell\'o}

\vskip 6mm
{\em  Center for Mathematical Analysis, Geometry and Dynamical Systems,\\
  Department of Mathematics, 
  Instituto Superior T\'ecnico,\\ Universidade de Lisboa,
  Av. Rovisco Pais, 1049-001 Lisboa, Portugal}\\[4mm]
  
 {\tt pedro.aniceto@tecnico.ulisboa.pt}\, ,
{\tt 
cristina.camara@tecnico.ulisboa.pt} \, ,
{\tt gabriel.lopes.cardoso@tecnico.ulisboa.pt}\, ,
{\tt martirossello@tecnico.ulisboa.pt }\\[10mm]

\end{center}

\vskip .2in
%%%%%%%%%%%%%%%%%%%%%%%%%%%%%%%%%%%%%%%%%%%%%%%%%%%%%
\begin{center} {\bf ABSTRACT } \end{center}
\begin{quotation}\noindent 

\noindent
We consider the Riemann-Hilbert factorization approach to the construction of Weyl metrics in four space-time dimensions.
We present, for the first time, a rigorous proof of the
remarkable fact that the
canonical Wiener-Hopf factorization
of a matrix obtained from a general (possibly unbounded) monodromy matrix, with respect to an appropriately chosen contour,
 yields a solution to the non-linear gravitational field equations. This holds regardless of whether the 
 dimensionally reduced metric in two dimensions has 
 Minkowski or Euclidean signature.
 We show moreover that, by taking advantage of a certain degree of freedom in the choice of the contour, the same monodromy 
matrix generally yields
various distinct solutions to the field equations. Our proof, which fills various gaps in the existing literature,
is based on the solution of a second Riemann-Hilbert problem and highlights
the deep role of the spectral curve, the normalization condition in the factorization 
and the choice of the contour. This approach allows us to construct explicit solutions,
including new ones, 
 to the non-linear gravitational field equations,
using simple complex analytic results.

%%%%%%%%%%%
\end{quotation}
\vfill
%%%%%%%%%%
\today
%%%%%%%%%%%%%%%%%%
\end{titlepage}

\tableofcontents

%%%%%%%

%%%%%%%
\section{Introduction}
%%%%%%%%%%%%%%%%

The field equations of gravitational theories in $D$ space-time dimensions are
a system of non-linear 
PDE's for the space-time metric which are,
in general, very difficult to solve. Exact solutions can however be found
under simplifying assumptions, for instance 
spherical symmetry. One such exact solution,
of significant physical and mathematical interest,
is the well-known Schwarzschild
solution which, in $D=4$ dimensions
in coordinates $(t,r,\theta, \phi)$, takes the form
\bea
ds^2_4 = - \left( 1 - \frac{2m}{r} \right) dt^2 +
 \left( 1 - \frac{2m}{r} \right)^{-1} dr^2 + r^2 
\left( d \theta^2 + \sin^2 \theta d \phi^2 \right) \;,
\label{schwarzm}
\eea
where $m \in \mathbb{R}^+, \, r>0, \, 0 < \theta < \pi$ and 
$0 \leq \phi  < 2 \pi$.

By restricting to the subspace of solutions that only depend on two of the $D$ space-time coordinates, various approaches 
to solving the field equations become available
(see for instance \cite{Alekseev:2010mx} for a 
recent survey 
thereof).
In this paper, we focus on the 
Riemann-Hilbert approach, 
which can be
applied to gravitational theories that satisfy certain requirements, see for instance \cite{Breitenlohner:1986um,Nicolai:1991tt,Katsimpouri:2012ky,
Camara:2017hez,Cardoso:2017cgi}
and references therein.
This approach is remarkable in that it allows us to obtain explicit solutions to the reduced field equations through the appropriate
factorization of a matrix function which depends on a complex variable $\tau$ that does not appear in the original problem. In
this factorization, the
space-time coordinates are taken as parameters, thus playing the role of constants. However, to this date, a clear and rigorous proof
that, under very general assumptions, this type of matrix factorization {\it always} yields a solution to the field equations has not been given.
One of the goals of this paper is to provide such a proof. We illustrate the power of this approach by 
showing that all type $A$ space-time metrics \cite{Ehlers:1962zz,Griffiths:2009dfa}, which includes the Schwarzschild solution, can be obtained from a single class of diagonal matrices, 
and by presenting new solutions to the gravitational field equations that we believe would be difficult to obtain through other methods.

Let us briefly summarise the novel aspects discussed in this paper: we give 
a clear and rigorous proof that, under very general and natural assumptions, 
the canonical Wiener-Hopf factorization of a so-called monodromy matrix always 
yields a solution to the gravitational field equations; 
we show that by choosing different factorization contours 
one such monodromy matrix gives rise to several distinct solutions to the field equations; we discuss 
the intricacy of reconstructing the interior region of the Schwarzschild solution from factorization data, which requires gluing together solutions by means of affine transformations; we perform a case study of the relation between meromorphic (as opposed to canonical) factorizations of a monodromy matrix and solutions to the gravitational field equations;
we obtain explicit solutions that, to the best of our knowledge, are new.

Let us now give a few more details.

When reducing to two dimensions, there are two distinct cases to consider, which we will 
distinguish through a parameter $\sigma$: either all the directions over which one 
reduces are 
space-like ($\sigma =-1$), or one of them is time-like and the others are space-like ($\sigma =1$).

In this paper, for the first time, a unified Riemann-Hilbert approach to both cases is presented. Namely, we give a rigorous proof that the 
canonical Wiener-Hopf factorization (defined in Section \ref{sec:sum}) of a general
monodromy matrix ${\cal M}$, with respect to an appropriately chosen contour $\Gamma$ in the complex plane, always yields exact solutions to the gravitational field equations, 
for $\sigma
= \pm 1$. Note that  a Wiener-Hopf factorization must be defined with respect to a given contour \cite{WH,CG, Speck}.

This proof, which is based on the formulation and solution of another associated
Riemann-Hilbert problem, highlights the power of the Riemann-Hilbert approach and clarifies the roles played by the so-called spectral curve and the properties of the Wiener-Hopf factorization. In this way we revisit and generalize the results obtained in
\cite{Breitenlohner:1986um,Nicolai:1991tt,Katsimpouri:2012ky,Camara:2017hez}, also filling a few gaps in the presentation given there.

As a result one finds, in particular, that from rational $n \times n$ monodromy matrices
${\cal M}$, whose canonical Wiener-Hopf factorization can be constructed explicitly and in a computationally simple manner, 
one obtains explicit exact solutions that would be very difficult to obtain through other approaches. This is the case of the
novel solutions presented in
\cite{Cardoso:2017cgi}, whose construction, based on the Riemann-Hilbert approach of \cite{Camara:2017hez}, is hereby rigorously
justified. Using our improved understanding of the role of the factorization contour, we return to 
one of the
new solutions obtained in \cite{Cardoso:2017cgi},
which was restricted to a certain domain in space-time.
Here we complete its analysis, discuss its meaning and properties, and we also consider other ranges of parameters.

We also show in this paper that  by taking advantage
of the possible choices of factorization contours $\Gamma$ and appropriate changes of coordinates, each monodromy matrix gives rise
not  to one solution, but to a whole class of exact solutions. 
We illustrate this
surprising result by showing that from
the canonical Wiener-Hopf factorization of a rational diagonal matrix of a very simple kind that is easily factorizable, one obtains 
a wide class of metrics that includes 
all  type $A$ space-time metrics,  
a cosmological Kasner solution and the Rindler metric, as well as solutions whose metric tensor is continuous but not smooth due to
the presence of null hypersurfaces \cite{Ehlers:1962zz,Barrabes:1991ng,barrab,Griffiths:2009dfa}. 
This wide class includes the solutions that describe the exterior and the interior
regions of the Schwarzschild black hole.
In doing so, we moreover put in evidence the essential difference between the case $\sigma = 1$ and $\sigma = -1$. 
Previous recent work that addresses this difference includes \cite{Jones:2004pz,Jones:2005hj}.

This paper is organized as follows. We chose to keep the introduction as brief as possible, supplementing it by a summary of the main
results in Section \ref{sec:sum}. In Section \ref{sec:prem_res} various preliminary results are presented that will be subsequently used.
In view of the intricacies mentioned above, we chose to write a self-contained paper to facilitate the reading. Therefore, 
we revisit the Breitenlohner-Maison linear system and the integrability of the field equations
in Section \ref{sec:bm}, and in
Section \ref{sec:mon} we
introduce and discuss monodromy matrices.
In Section \ref{sec:mtbm} we prove the main theorem of the paper, which states that
the canonical Wiener-Hopf factorization of a monodromy matrix always yields exact solutions, in both cases $\sigma = \pm 1$. In Section \ref{sec:kasn}
we discuss the possibility to obtain such solutions from other types of matrix factorizations, allowing for poles in the factors.
In Section \ref{sec:schw} we present a thorough study of the exact solutions that can be obtained from the
canonical factorization, for $\sigma = \pm 1$, of a 
very simple rational
diagonal 
monodromy matrix, given by \eqref{monschwarzeps} with $\epsilon =1$.
We call this monodromy matrix the Schwarzschild monodromy matrix. We show that 
by using the different allowed choices of the factorization contour, one can construct a wide class of metrics, including 
for instance the solution that describes 
the
interior region of the Schwarzschild black hole. 
In Section \ref{sec:schw0} we discuss the canonical factorization of the monodromy matrix  \eqref{monschwarzeps} with $\epsilon =0$.
The resulting solutions include the Rindler metric and a cosmological Kasner solution. 
In Section \ref{sec:nsol} we return to one of the new solutions 
given in \cite{Cardoso:2017cgi}, complete its analysis and study new cases. Further details about these solutions are given in Appendix \ref{nmetQP}.
In Appendix \ref{sec:glue} we discuss 
jumps of the transverse extrinsic curvature that arise in the 
presence of null hypersurfaces in space-time across which the metric tensor is only continuous.
In Appendix \ref{sec:classA} we summarize the class of $A$-metrics, which here is obtained
from the canonical factorization of the monodromy matrix \eqref{monschwarzeps}.

%%%%%%%%%%%%%%%%%%%

%%%%%%%%%%
\section{Summary of the main results \label{sec:sum}}
%%%%%%%%%%%%%%%%
To be able to use a Wiener-Hopf factorization to obtain solutions of the field equations of gravitational theories in $D$ space-time
dimensions, we consider gravitational theories that satisfy certain requirements, as follows.  Firstly,
we restrict to gravitational theories in the absence of a cosmological constant, that is,
gravitational theories that have Minkowski space-time as a vacuum solution.
Secondly, we focus on the subspace of solutions possessing sufficiently many commuting  isometries
so that the theory can be dimensionally reduced, first to three dimensions, and subsequently to
two dimensions. Thirdly, we take 
the dimensionally reduced theory in three dimensions to be described by
a scalar sigma-model coupled to three-dimensional gravity, such that the target space of the 
sigma-model is a symmetric space $G/H$.  This target space is endowed with an involution,
associated with an involutive Lie algebra automorphism
of the Lie algebra of $G$.
This involution, also called `generalized transposition',
leaves the coset representative $M \in G/H$ invariant.
For instance, for  $H \subset O(n)$, this involution acts by matrix transposition, while for 
$H \subset U(n)$ it acts by Hermitian conjugation. The involution acts anti-homomorphically
on $g \in G$
(see \cite{Camara:2017hez} for a review).

The problem of solving the gravitational field equations in $D$ dimensions is then reduced to solving a system
of non-linear second order PDE's depending on two coordinates, which we denote by $(\rho,v) \in {\mathbb R}^2$,
with $\rho >0$ and $v \in \mathbb{R}$.   These coordinates are called Weyl coordinates, 
and accordingly we will denote the $(\rho,v)$ upper-half plane by Weyl upper-half plane.
The system
of non-linear second order PDE's to be solved is given by
 \cite{Schwarz:1995af,Lu:2007jc}
\bea 
d ( \rho \star A) &=& 0 \,, \label{eq_motion_A} 
\eea
where $A$ denotes the matrix one-form
\bea
A = M^{-1} dM \;.
\label{AMdM}
\eea
Here, $\star$ denotes the Hodge star operator in two dimensions, which maps a one-form to a one-form, and 
satisfies
\be
( \star )^2 = - \sigma \, \mathrm{id} \,, \qquad  \star \, d  \rho = - \sigma dv \,, \qquad  \star \,  dv  = d \rho  \,.
 \label{hodgedualcoordinates}
\ee
We recall that $\sigma = \pm 1$ distinguishes between the case where both $\rho$ and $v$ are
space-like coordinates ($\sigma = 1$), and the case where one of them is a time-like and the other
one is space-like ($\sigma =-1$).

Given a solution of \eqref{eq_motion_A}, one then obtains a solution to the gravitational field equations.
Note that, if $M$ is a diagonal matrix, $M^{-1}$ also provides a solution to the field equations.
As an illustration,  let us consider solutions to the field equations of General Relativity in four dimensions
that possess two commuting isometries. 
 Accordingly, we take the space-time metric to have the 
Weyl-Lewis-Papapetrou form
\bea
ds_4^2 = - \sigma \Delta (dy  + B  d \phi  )^2 + \Delta^{-1} 
\left(  e^\psi \, ds_2^2 + \rho^2 d\phi^2 \right) \;,
\label{4dWLP}
\eea
with $\Delta > 0$, where  $ds_2^2$ denotes a flat, two-dimensional line element, which we take to be either
\begin{equation}
ds^2_2 = \sigma \, d \rho^2 + dv^2 \;.
\label{weylframe}
\end{equation}
or 
\begin{equation}
ds^2_2 = d \rho^2 + \sigma \, dv^2 \;.
\label{weylframe2}
\end{equation}
$\Delta, B$ and $\psi$ are
functions of the Weyl coordinates  $(\rho, v)$ only. 
By dimensionally reducing over $(y, \phi)$ to two dimensions, the functions $\Delta$ and $B$ become encoded in the matrix $M$,
which in this case is a two-by-two matrix.  Thus, given a solution $M$ of \eqref{eq_motion_A}, one immediately reads off the expressions
for $\Delta$ and $B$. The function $\psi$ in \eqref{4dWLP} is determined by integration \cite{Schwarz:1995af,Lu:2007jc},
\be
\partial_\rho \psi = \tfrac{1}{4} \, \rho \,  \Tr \left(A_\rho^2 - \sigma A_v^2\right) \,,\qquad \partial_v \psi = \tfrac{1}{2} \, \rho \, 
 \Tr \left(A_\rho A_v \right) \,. \label{eq_diff_eq_psi}
\ee

Therefore, the central question is: how do we determine a solution $M(\rho,v)$ of \eqref{eq_motion_A}?\\

To do so, we proceed in three steps.\\

{\it Step 1:} We use the fact that the non-linear PDE's \eqref{eq_motion_A} form an integrable system, i.e. 
they are the solvability conditions \cite{Its} for a certain Lax pair,
the so-called Breitenlohner-Maison (BM) linear system \cite{Breitenlohner:1986um}.
This is an auxiliary linear system in two dimensions given by
\begin{equation} \label{Laxpairi}
	\varphi(\rho, v)\, [ \,dX(\rho, v) + A(\rho, v) X(\rho, v)\,]= \star \, dX(\rho, v) \;,
\end{equation}
specified in terms of the matrix one-form $A = M^{-1} dM$
and a function $\varphi$ of the form
\bea
 \varphi (\rho,v) = 
\frac{-\sigma(\omega-v) \pm \sqrt{(\omega-v)^2+\sigma\rho^2}}{\rho} \;,\;
\text{with $ \omega \in \mathbb{C}\backslash  \left\{ v_0 \pm \sqrt{-\sigma} \, \rho_0 \right\}$ }
\;,
\label{setcalT}
\eea
where $(\rho_0, v_0)$ is a point in the neighbourhood of which the linear system is to be solved.
We denote by ${\cal T}$ the set of all functions of this form.
If the BM linear system has a non-trivial solution satisfying certain invertibility and differentiability conditions, then
 $A = M^{-1} dM$ satisfies \eqref{eq_motion_A}.
 
 Note that  $\pm \sqrt{-\sigma}$ are the fixed points of 
 the involution 
$\iota_\sigma$  in $\mathbb{C} \backslash \{0\}$,
\begin{equation}
	\iota_\sigma(\tau)=  -\frac{\sigma}{\tau} \;,
	\label{invsig}
\end{equation}
which will play a fundamental role.\\

 {\it Step 2:}  
We will assume that we can define an involutive map $\natural$ acting on matrix functions $N(\tau, \rho,v)$
which coincides with the `generalized transposition'  in $G/H$ whenever $N \in G/H$.
Examples thereof are $N^{\natural} (\tau, \rho, v)  = N^T (\tau, \rho, v) $ and $N^{\natural} (\tau, \rho, v)  = \eta \, N^{\dagger}
(\overline{\tau}, \rho, v) 
\, \eta^{-1}$, 
where in the second case 
$\eta$ denotes a constant invertible matrix and $\tau$ takes values in the unit disc.  We also assume that
$\mathcal{T}$ is closed under $\natural$.

 If we take $\varphi, \, \chi = - \sigma/\varphi^{\natural}  \in {\cal T}$, and if $X$ and $\widetilde X$ denote solutions to the corresponding
 BM linear system, then we can prove that the matrix 
 \be
 {\cal M}(\rho, v)= {\widetilde X}^{\natural} (\rho,v) \, M(\rho,v) \, X(\rho,v) 
 \label{mxmx}
 \ee
 satisfies 
 \be
 d {\cal M} (\rho, v) =0 \;.
 \ee
Thus, the matrix ${\cal M}$ is independent of the Weyl coordinates, even though all the individual factors
depend on $(\rho,v)$. 
The expression \eqref{mxmx}
shows that ${\cal M}$ can be constructed by a product involving solutions to the BM linear system and a matrix $M$
that solves the field equations \eqref{eq_motion_A}.\\

 {\it Step 3:} Now we turn to the reverse question: given a matrix ${\cal M}$ that is
 independent of the Weyl coordinates, can we construct a factorization of the form \eqref{mxmx},
 such that the middle factor $M(\rho,v)$ solves the field equations, and the factor $X$
 provides a solution to the BM linear system?  
 If this can be done, 
 we say that $\mathcal{M}$ is a {\it monodromy matrix} for $M(\rho, v)$.

This is where Wiener-Hopf factorization comes into play.  We will have to make 
certain (though very general) assumptions, that we summarize
as follows (these assumptions may seem a bit technical, but in fact they are natural conditions in order for the 
class of involutions $\natural$ to be as general as possible).
\\

Let $\mathcal{M} (u)$ be an invertible matrix function of the complex variable $u$, and consider
the matrix that is obtained by composition with 
\be
u = v + \sigma \, \frac{\rho}{2}\, \frac{\sigma - \tau^2}{\tau} \;\;\;,\;\;\; \tau \in \mathbb{C}\backslash\{ 0 \}\;,
\label{utauu}
\ee
i.e.
\be
 \mathcal{M} \left(v + \sigma \, \frac{\rho}{2} \frac{\sigma - \tau^2}{\tau} \right) \;\;\;,\;\;\; \tau \in \mathbb{C}\backslash \{0\} \;.
 \label{calMuu}
 \ee
For each $(\rho,v)$, the relation \eqref{utauu} is an algebraic curve in the complex variables $u$ and $\tau$, called the 
{\it spectral curve}. It plays a fundamental role in this study. Note also that if we replace $\tau$ by any $\varphi \in {\cal T}$,
 the relation \eqref{utauu} is satisfied for $u = \omega$.
 
 In  \eqref{calMuu}, we consider $(\rho,v)$ as an arbitrary pair of parameters in a neighbourhood of $(\rho_0, v_0)$, and take $\tau$ as 
 the (complex) independent variable. To emphasize this aspect, we will denote the matrix ${\cal M}$ in  \eqref{calMuu} by ${\cal M}_{(\rho, v)} (\tau)$.
\\

\noindent
{\it Assumption 1: 
There exists an open set $S$ in the Weyl upper-half plane such that, for every 
 $(\rho_0, v_0)\in S$, one can find a
simple closed curve $\Gamma$ in the $\tau$-plane, which is 
$\iota_\sigma$-invariant and 
encircles the origin, such that:
\\
%a)
 For all 
$(\rho,v)$ in a neighbourhood of $(\rho_0, v_0)$, $ {\mathcal M}^{\pm 1}_{(\rho, v)}(\tau) $
is analytic in an $\iota_\sigma$-invariant open set $O$ 
in the $\tau$-plane, 
containing $\Gamma$ 
(see  Figure \ref{ribbon}), and such that  $\mathcal{M}^{\natural}_{(\rho, v)}(\tau) = \mathcal{M}_{(\rho, v)}(\tau) $ on $O$.
}\\ 

Let $D^+$ denote the simply connected interior region of $\Gamma$, and $D^- = \mathbb{C} \backslash (D^+ \cup \Gamma )$.
\\

\noindent
{\it Assumption 2:  If $N(\tau, \rho, v)$ is analytic in $D^+ \cup O$, then $N^{\natural} (\tau, \rho, v)$ is also analytic in $D^+ \cup O$.}
\\

\noindent
{\it Assumption 3:
$ {\mathcal M}_{(\rho, v)}(\tau) $ admits a canonical
Wiener-Hopf factorization with respect to $\Gamma$,
\be
 \mathcal M_{(\rho, v)}(\tau) =  M^-_{(\rho, v)}(\tau) \, M^+_{(\rho, v)}(\tau) \;\;\; {\rm on } \;\; \Gamma \;,
 \label{Mfactori}
 \ee
where $M^+_{(\rho, v)}(\tau)$ (respectively  $M^-_{(\rho, v)}(\tau)$) and its inverse are analytic and bounded in 
the open set  $D^+ \cup O$ (respectively  $D^- \cup O$).
$M^+_{(\rho, v)}(\tau)$ satisfies the 
normalization condition
$M^+_{(\rho, v)}(0) = \mathbb{I}$. 
We set $X(\tau, \rho, v)  = M^+_{(\rho, v)}(\tau)$, so $X(0, \rho, v) = \mathbb{I}$.
}\\

Consider an $n \times n$ matrix function ${\cal M}(\tau), \, \tau \in \Gamma$, such that both ${\cal M}$ and ${\cal M}^{-1}$
are continuous on $\Gamma$.
A representation of 
${\cal M}(\tau), \, \tau \in \Gamma$, as a product 
\bea
{\cal M}(\tau) = M_-(\tau) \, d (\tau)  \, M_+ (\tau) \;,
\eea
where $d$ is a diagonal matrix of the form $d(\tau) = \textrm{diag} ( \tau^{{k}_j})_{j=1, 2, \dots, n}$ with $k_j \in \mathbb{Z}$, and 
where $M_+^{\pm 1}$ (respectively $M_-^{\pm 1}$) admit bounded analytic extensions to $D^+$ (respectively $D^-$), is called
a (bounded) Wiener-Hopf factorization. When $d = \mathbb{I}$, this factorization is called canonical.
The latter, if it exists, is unique, up to a constant matrix factor which can be fixed by imposing
a normalization condition, such as $M_+(0) =  \mathbb{I}$ \cite{CG,LS,MP}.
Then, as shown in \cite{Camara:2017hez}, \eqref{Mfactori} can be written as
\be
 \mathcal M_{(\rho, v)}(\tau) =  X^{\natural} (- \frac{\sigma}{\tau}, \rho, v) \, M(\rho, v) \, X(\tau, \rho, v) \;\;\;,\;\;\; \tau \in \Gamma \;,
 \label{MMXi}
 \ee
where 
\be 
M(\rho, v) = \lim_{\tau \rightarrow \infty}
M^-_{(\rho, v)}(\tau) = M^-_{(\rho, v)}(\infty) \;.
\ee
\\

\noindent
{\it Assumption 4: $M(\rho,v)$ is of class $C^2$, and for each $\tau \in D^+ \cup O$ the matrix function $X$ 
is of class $C^2$ as a function of $(\rho, v)$, 
and $\partial X/ \partial \rho$, $\partial X/ \partial v$, $\partial X^{\natural}/ \partial \rho$ and $\partial X^{\natural}/ \partial v$ 
are analytic as functions of $\tau$ in the domain $D^+ \cup O$.
}
\\

Under these assumptions, we have the following main theorem:

\begin{theorema}  Let Assumptions 1- 4 be satisfied.
Then
$M(\rho,v)$ is a solution of the field equations \eqref{eq_motion_A} on $S$.
\end{theorema}

The proof of this theorem gives, moreover, an affirmative answer to the question raised in {\it Step 3}.

By using this result, we study the solutions to the field equations of General Relativity in four dimensions that arise from the 
canonical Wiener-Hopf
factorization of a particular type of matrices. These monodromy matrices are chosen to be of the {\it simplest possible, non-constant}
kind, namely, they are diagonal, rational, with $\det \mathcal{M} =1$, possessing only one zero and one pole in the extended $u$-plane,
as follows,
\be
\mathcal{M} (u) = \begin{bmatrix}
	\sigma \,  \frac{u -\epsilon \, m}{\epsilon \, u +m}  & 0 \\
	0 & \sigma \,  \frac{\epsilon \, u + m}{u - \epsilon \, m}
\end{bmatrix} \;\;\;,\;\;\; m \in \mathbb{R}^+ \;\;\;,\;\; \epsilon = 0, 1 \;.
\label{monschwarzeps}
\ee
Note that the matrix $ \mathcal M_{(\rho, v)}(\tau)$, obtained from \eqref{monschwarzeps} by composition with \eqref{utauu}, always
admits a canonical Wiener-Hopf factorization \cite{Cardoso:2017cgi}. {From}  that factorization we obtain, by choosing different
possible factorization contours, 
a wide class of solutions that includes all  type $A$ space-time metrics,  
a cosmological Kasner solution and the Rindler metric, as well as solutions whose metric tensor is continuous but not smooth due to
the presence of null hypersurfaces
\cite{Ehlers:1962zz,Barrabes:1991ng,barrab,Griffiths:2009dfa}.

Although the matrices  $\mathcal M_{(\rho, v)}(\tau) $ in the cases $\sigma = 1$ and $\sigma = -1$ are very similar, the behaviour
of the corresponding factors $M(\rho, v)$ is remarkably different, due to the presence of square roots that may vanish along certain
lines in the $(\rho,v)$ upper-half plane when $\sigma = -1$. Such lines are absent when $\sigma = 1$.
Obtaining the $A$-metrics for $\sigma =-1$ also involves extending the real-valued solutions by affine transformations in
 the $(\rho,v)$ upper-half plane, followed by appropriate changes of coordinates.

So far, our discussion focussed on canonical factorizations of monodromy matrices.  What about other types of factorizations? Do other
types of matrix factorizations, where different analyticity and normalization conditions are imposed on the factors
(allowing them to be meromorphic, for instance), also yield solutions to the 
field equations, at least in certain cases? To address this question,
we consider
the example of a cosmological 
Kasner solution which, when expressed in terms of Weyl coordinates, belongs to the class $\sigma = -1$.
By solving the BM linear system for this solution, we construct an infinite set of monodromy matrices. We then pick one of them, and we show
that this particular monodromy matrix possesses a  meromorphic factorization \cite{CLS}  that gives back the Kasner solution, whereas
its canonical factorization gives rise to a different solution to the field equations. Both solutions to the field equations are, however, related by a certain
transformation, which we give in Section \ref{sec:kasn}. To our knowledge, this is the first time that it is shown that a meromorphic factorization of a monodromy
matrix can give rise to a solution of the gravitational field equations.

In Section \ref{sec:nsol} we return to one of the monodromy matrices introduced in \cite{Cardoso:2017cgi}. Its canonical factorization 
can give rise to various solutions due to different choices
of the factorization contour.   The contour considered in \cite{Cardoso:2017cgi} was the unit circle on the complex $\tau$-plane.
There is, however, no need to pick this particular contour: as mentioned in Assumption 1, we are free to choose other 
simple closed contours that enclose
the origin and pass through the fixed points of the involution \eqref{invsig}. By replacing the unit circle with another such contour,
we complete
the analysis of the solution obtained in \cite{Cardoso:2017cgi} and find that is possesses two Killing horizons.  
We give the exact expression for the space-time metric, and we discuss its interpretation. The solution discussed in 
\cite{Cardoso:2017cgi} carries one electric charge $Q> 0$ and one magnetic charge $P>0$. We also discuss what happens when $Q P <0$, 
for the same choice of the factorization contour. We find that the case $Q P <0$
is markedly different from the case $Q P >0$ discussed in \cite{Cardoso:2017cgi}.

%%%%%%%%
\section{Preliminary results \label{sec:prem_res}}
%%%%%%%%%%

We begin by introducing the involution 
$\iota_\sigma$  in $\mathbb{C} \backslash \{0\}$,
\begin{equation}
	\iota_\sigma(\tau)=  -\frac{\sigma}{\tau} \;.
	\label{invtt}
\end{equation}
It has two fixed  points, which we denote by $\pm p_F$: $\pm i$ if $\sigma = 1$, and $\pm 1$ if $\sigma = -1$.
We will denote the set of fixed points by
$FP_{\sigma}$.

Next, we introduce the set
$W_{\rho,v} = \left\{ v \pm \sqrt{-\sigma} \, \rho \right\}$ for any $(\rho,v)$ in the Weyl upper-half plane.
 We can write this set as
$W_{\rho,v} = \left\{ v \pm p_F \, \rho \right\}$.

\noindent
{\it Remark:}  
Let
\be
\omega =  
v + \sigma \, \frac{\rho}{2} \, \frac{\sigma - \tau^2}{\tau} \;\;\;,\;\;\; \tau \in \mathbb{C} \backslash \{0\} \;.
\label{omvrt}
\ee
Then $\tau = \pm p_F$ iff $\omega  \in W_{\rho,v}$, because $p^2_F = - \sigma$.

\subsection{Properties of $\varphi \in \mathcal{ T}$}

\begin{proposition} Let $(\rho_0,v_0)$ be a point in the Weyl upper half-plane.
For all $\omega \in \mathbb{C} \backslash W_{\rho_0,v_0} $, there exists a branch of the square root such that $\varphi$ and 
$\widetilde{\varphi} = - \sigma/ \varphi$, given by
\be
\varphi (\rho, v)= \frac{-\sigma(\omega-v) + \sqrt{(\omega-v)^2+\sigma\rho^2}}{\rho} \;,
\label{vfi}
\ee
are of class $C^{\infty}$ in a neighbourhood of $(\rho_0,v_0)$.
\end{proposition} 

\begin{proof} 
Let $(\rho_0,v_0)$ be any point in the Weyl upper half-plane. We start by noting that
$(\omega-v_0)^2+\sigma\rho_0^2$ vanishes precisely for $\omega \in W_{\rho_0,v_0}$.
So let $\omega \in \mathbb{C} \backslash W_{\rho_0,v_0} $. Decomposing $\omega = \omega_R + i \omega_I$
with $\omega_R, \omega_I \in \mathbb{R}$, we have
\be
(\omega-v_0)^2+\sigma\rho_0^2 = (\omega_R-v_0)^2 - \omega_I^2 +\sigma\rho_0^2 + 2 i \omega_I (\omega_R -v_0) \;.
\label{omdecp}
\ee
Let us first consider the case $\sigma =1$. We choose the principal branch of $\sqrt{z}$ (i.e. $\arg z \in ] - \pi, \pi]$), with branch cut $\mathbb{R}^-$.
We have $(\omega-v_0)^2+\sigma\rho_0^2 \in \mathbb{R}$ iff $\omega_I =0$ or $\omega_R =v_0$. If $\omega_I =0$, the real part of 
\eqref{omdecp} is positive. If  $\omega_R =v_0$,  the real part of 
\eqref{omdecp} is also positive unless $|\omega_I| > \rho_0$ (we exclude $v_0 \pm i \rho_0 \in  W_{\rho_0,v_0} $ from our considerations). 
Therefore, with the chosen branch of $\sqrt{z}$,
$\sqrt{(\omega-v)^2+\rho^2}$ is $C^{\infty}$ in a neighbourhood of
$(\rho_0,v_0)$ unless $\omega_R = v_0$ and $|\omega_I| > \rho_0$.  Since in the latter case the real part of \eqref{omdecp} is negative, we choose
the branch of $\sqrt{z}$ with $\arg z \in [0, 2 \pi[$ for this case. 

When $\sigma = -1$, a similar reasoning leads to the choice of the latter branch for $\sqrt{z}$, unless $\omega_I =0$ and $| \omega_R -v_0| > \rho_0$,
in which case we choose the principal branch.

\end{proof}

Next, for any given $(\rho_0,v_0)$, 
we define the set $\mathcal{T}_{\rho_0, v_0}$ of all functions $\varphi$ on the Weyl upper half-plane, of the form 
\be
\varphi (\rho, v)= \frac{-\sigma(\omega-v) \pm \sqrt{(\omega-v)^2+\sigma\rho^2}}{\rho} 
\label{vfi2}
\ee
for some  $\omega \in \mathbb{C} \backslash W_{\rho_0,v_0}$. 
Note that if $\varphi \in \mathcal{T}_{\rho_0, v_0}$, then also $ \widetilde{\varphi} \in \mathcal{T}_{\rho_0, v_0}$, where
\be
\widetilde{\varphi} = - \sigma/\varphi \;,
\label{vphitildevphi}
\ee
i.e. $\mathcal{T}_{\rho_0, v_0}$ is $\iota_\sigma$-invariant.
Moreover, for $\varphi \in \mathcal{T}_{\rho_0, v_0}$, we have that $\varphi (\rho_0, v_0) \neq 0$ and $\varphi^2 (\rho_0, v_0) + \sigma \neq 0$.
Therefore, $\varphi$  and $\varphi^2 + \sigma$ do not vanish in a 
neighbourhood of $(\rho_0, v_0)$.
Also
note that for
any $\varphi \in \mathcal{T}_{\rho_0, v_0}$, that is $\varphi$ of the form  \eqref{vfi2}, we have
 \begin{equation} \label{spectralcurve}
	v + \sigma \,\frac{\rho}{2}\, \left( \frac{\sigma-\varphi^2(\rho,v)}{\varphi(\rho,v)}  \right)= \omega \;.
\end{equation}
This relation is invariant under the replacement of $\varphi$ by $ \widetilde{\varphi}$.

Now we discuss various properties of $d \varphi$ which are valid in a neighbourhood of 
$(\rho_0, v_0)$ where $\varphi \in \mathcal{T}_{\rho_0, v_0}$ is of class $C^{\infty}$.
For ease of notation, we will from now on denote $ \mathcal{T}_{\rho_0, v_0}$ simply by $ \mathcal{T}$.

\begin{proposition} \label{lemmadifferentials} 
Let $\varphi\in\mathcal{T}$.
Then we have
\begin{equation} \label{lemma4.1}
	d\varphi = \frac{\varphi}{\rho} \left[ \frac{\sigma - \varphi^2}{\sigma + \varphi^2}d\rho + \frac{2\sigma \varphi}{\sigma + \varphi^2}dv \right],
\end{equation}

\begin{equation} \label{lemma4.2}
	\star d\varphi = \sigma \, \frac{\varphi}{\rho} \left[ -  \frac{\sigma - \varphi^2}{\sigma + \varphi^2} dv + \frac{2 \varphi}{\sigma + \varphi^2}d\rho \right],
\end{equation}

\begin{equation} \label{lemma5}
	\frac{2\varphi\sigma}{\varphi^2+\sigma}d\varphi + \frac{\varphi^2-\sigma}{\varphi^2+\sigma} \star d\varphi = \sigma\frac{\varphi}{\rho}dv.
\end{equation}
\end{proposition}
\begin{proof} 
Recall that 
$\varphi \neq 0 $ and $\varphi^2 + \sigma \neq 0$ for $\varphi\in\mathcal{T}$.
Equation \eqref{lemma4.1} 
is easily verified, and equation
\eqref{lemma4.2} follows by using \eqref{hodgedualcoordinates}. Then, 
\begin{eqnarray}
	\frac{2\varphi\sigma}{\varphi^2 + \sigma}d\varphi + \frac{\varphi^2-\sigma}{\varphi^2+\sigma} \star d\varphi &=&  \frac{1}{(\varphi^2+\sigma)^2}\frac{\varphi}{\rho}\left[ 2\sigma\varphi(\sigma - \varphi^2)d\rho + 4\varphi^2\sigma^2 dv + \sigma(\varphi^2-\sigma)^2 dv \right. \nonumber\\
	  && \left.  \qquad \qquad \qquad  + (\varphi^2-\sigma)2\sigma \varphi d\rho \right] \nonumber\\
&=& \frac{1}{(\varphi^2+\sigma)^2}\frac{\varphi}{\rho}\sigma(\varphi^2+\sigma)^2dv 
	= \sigma\frac{\varphi}{\rho} dv \;,
\end{eqnarray}
which yields  \eqref{lemma5}. 
\end{proof}

\begin{proposition} \label{lemmatangent}
Let $\varphi\in\mathcal{T}$. Then
\begin{equation}
	d\left( \rho \star  \frac{d\varphi}{\varphi}\right) = 0.
\end{equation}
\end{proposition}

\begin{proof}
Recall that 
$\varphi \neq 0 $ and $\varphi^2 + \sigma \neq 0$ for $\varphi\in\mathcal{T}$.
We set
\begin{equation}
S = \frac{ 2 \varphi}{\sigma + \varphi^2} \;\;\;,\;\;\; C = \frac{\sigma - \varphi^2}{\sigma + \varphi^2} \;,
\end{equation}
and compute
\begin{equation}
d S = 2  C \, \frac{d \varphi}{\sigma + \varphi^2} \;\;\;,\;\;\; 
d C =  - \sigma \, 2 S \,   \frac{d \varphi}{\sigma + \varphi^2} \;.
\end{equation}
Using \eqref{lemma4.1} in the form
\begin{equation}
d \varphi = \frac{\varphi}{\rho} \left( C d\rho + \sigma S dv \right) \;,
\end{equation}
we obtain
\begin{equation}
d\left( \rho \star  \frac{d\varphi}{\varphi}\right) = \sigma \frac{2}{\sigma + \varphi^2} \left( \sigma S d\varphi \wedge d v + C d \varphi \wedge d \rho
\right) =  \sigma \frac{2 \rho}{\varphi(\sigma + \varphi^2)}  d\varphi \wedge d \varphi = 0 \;.
\end{equation}

\end{proof}

%%%%%%%%
\subsection{Contour properties}
%%%%%%%%%%%

The following result will be needed in subsequent sections.

\begin{proposition} \label{lemmainsideoutside} \textit{Let $\Gamma$ be a simple closed curve
 in the complex $\tau$-plane, such that it encircles the origin of the $\tau$-plane, and such that it is invariant  under the involution $\iota_\sigma$.
 Let  $D^+$ denote the simply connected interior region of $\Gamma$ (i.e. $0\in D^+$), and let
 $D^- = \mathbb{C} \backslash \left(D^+ \cup \Gamma\right)$, so that $D^+ \cup \Gamma \cup D^- = \mathbb{C}$.
Let $p \in \mathbb{C} \backslash\{0\} $. If
  $p\in D^+$ then $\iota_\sigma(p) \in D^-$, and if $p\in D^-$ then $\iota_\sigma(p) \in D^+$.}
  \end{proposition}

\begin{proof} 
We take $\Gamma$ to be positively oriented.
Consider a 
point $w_0 \in \mathbb{C} \backslash \left( \{ 0 \} \cup \Gamma \right)$.
Then, the winding number of $\Gamma$ around $w_0$ is either zero or one. 
Now, for any such $w_0$,
\bea
\frac{1}{2\pi i} \ointctrclockwise_\Gamma \frac{dz}{z- w_0} &=&
\frac{1}{2\pi i} \frac{\sigma}{w_0} \ointctrclockwise_\Gamma \frac{dw}{(w + \frac{\sigma}{w_0} ) w} 
= \frac{1}{2\pi i} \ointctrclockwise_\Gamma \frac{dw}{w} - 
\frac{1}{2\pi i}  \ointctrclockwise_\Gamma \frac{dw}{w + \frac{\sigma}{w_0} } \nonumber\\
&=& 1 - \frac{1}{2\pi i}  \ointctrclockwise_\Gamma \frac{dw}{w + \frac{\sigma}{w_0} } \:, 
\eea
where we used  the invariance of $\Gamma$ under $\iota_\sigma$. Hence, if the winding number around $w_0$ is zero,
the winding number around  $\iota_\sigma (w_0)$ is one, and vice-versa.

\end{proof}

\begin{corollary} \label{lemmafixedpoints} \textit{
Let $\Gamma$ be a simple closed curve
 in the complex $\tau$-plane, such that it encircles the origin of the $\tau$-plane, and such that it is invariant  under the involution $\iota_\sigma$.
Then, $\Gamma$ passes through the fixed points of $\iota_\sigma$.}

\end{corollary}

\begin{proof} 
This is a simple consequence of Proposition \ref{lemmainsideoutside}. 
Let $w_0$ denote one of the fixed points. If $w_0 \notin \Gamma$, then either $w_0 \in D^+$
or $w_0 \in D^-$, in which case $w_0$ and $\iota_\sigma (w_0)$ have different winding numbers by Proposition \ref{lemmainsideoutside}, which
contradicts the assumption that $w_0$ is a fixed point of $\iota_\sigma$ (i.e. $w_0 = \iota_\sigma (w_0)$).

 \end{proof}

\subsection{Affine transformations}

We close this section with a discussion of 
changes of Weyl coordinates $(\rho, v) \mapsto ({\tilde \rho}, \tilde v)$
that preserve the form of the 
 two-dimensional line element  $ds_2^2$, given either by \eqref{weylframe} or by  \eqref{weylframe2}.
Namely, we consider
affine transformations
\begin{eqnarray}
\begin{pmatrix}
{\tilde \rho} \\
{\tilde v} 
\end{pmatrix}
=
\begin{pmatrix}
a & b \\
c & d 
\end{pmatrix}
\begin{pmatrix}
{\rho} \\
{v} 
\end{pmatrix}
+ 
\begin{pmatrix}
{\alpha} \\
{\beta} 
\end{pmatrix} \;\;\;,\;\;\; a, b, c, d, \alpha, \beta \in \mathbb{R} \;.
\label{affinetransf}
\end{eqnarray}
Invertibility of the linear part of the transformation requires $ad - bc \neq 0$. 
In addition, we 
demand that 
$ \sigma d \rho^2 + dv^2 =  k \,  (\sigma  d {\tilde \rho}^2 + d {\tilde v}^2 ) $, 
where $k \in \mathbb{R} \backslash \{0\}$. This
requires imposing 
\begin{eqnarray}
\sigma a^2 + b^2 = \sigma d^2 + c^2   \;\;\;,\;\;\; a c = - \sigma b d \;,
\label{condcoef}
\end{eqnarray}
or equivalently,
\begin{eqnarray}
\sigma a^2 + c^2 = \sigma d^2 + b^2   \;\;\;,\;\;\; a b = - \sigma c d \;.
\end{eqnarray}
Thus, the linear part of the transformation is a linear conformal isometry when $k>0$, and a linear conformal anti-isometry when $k<0$.

Under the change of Weyl coordinates $({\tilde \rho}, {\tilde v}) \mapsto ({ \rho}, {v})$, 
a solution ${\tilde M}({\tilde \rho}, {\tilde v})$
to \eqref{eq_motion_A} gets mapped to $ { \tilde M}  ( {\tilde \rho} (\rho,v), {\tilde v}(\rho, v)) \equiv {M} (\rho, v)$, which
satisfies the field equations
\bea 
d \left(  {\tilde \rho} (\rho, v)   {M}^{-1} (\rho, v) \star d {M} (\rho, v) \right) &=& 0 \,.
\eea

%%%%%%%%%%%%%
\section{The Breitenlohner-Maison linear system \label{sec:bm}}
%%%%%%%%%%%%%%%

The non-linear equations \eqref{eq_motion_A}   are an integrable system, i.e., they are the solvability conditions \cite{Its} for a certain Lax pair,
the so-called Breitenlohner-Maison (BM) linear system \cite{Breitenlohner:1986um}.
This is an auxiliary linear system in the Weyl upper-half plane given by
\begin{equation} \label{Laxpair}
	\varphi(\rho, v) \Big( dX(\rho, v) +  A(\rho, v) X(\rho, v) \Big)= \star \, dX(\rho, v) \;.
\end{equation}
It is defined in terms of a function $\varphi$ and a matrix one-form $A$. The latter 
takes the form given in \eqref{AMdM}. 
We assume that the components of $A$ are one-forms, whose coefficient functions are continuously differentiable functions in an open set in 
the Weyl upper-half plane.
We will refer to this requirement by saying that $A$ is of class $C^1$.
The function $\varphi$ is taken from the set $\mathcal{T} $ defined in \eqref{vfi2}.

Given $A$ and $\varphi$,  we seek solutions $X$ with the following properties: the matrix $X$ is invertible; $X$ is  twice continuously differentiable (i.e. of class $C^2$) and
$X^{-1}$ is continuously differentiable (i.e. of class $C^1$) with respect to $(\rho,v)$.

Next, let us discuss the solvability of the BM linear system \eqref{Laxpair}, following \cite{Lu:2007jc,Camara:2017hez}.
We will make use of the relations
\begin{equation} \label{lemma2}
	d(dX X^{-1}) = dX X^{-1}\wedge dX X^{-1}
\end{equation}
and
\begin{equation} \label{lemma3}
	d(\star \, A) = \frac{1}{\rho}d(\rho \star A) - \frac{1}{\rho}d\rho \wedge \star A,
\end{equation}
where we recall that $\rho >0$.

\begin{proposition} \label{propositioniff}
Let $\varphi \in \mathcal{T}$. Then
\begin{equation*}
	\varphi \Big(dX + A \, X \Big) = \star \, dX
\end{equation*}
\begin{equation*}
	\Longleftrightarrow
\end{equation*}
\begin{equation}
	(\varphi^2 +\sigma)dX X^{-1}= -\varphi^2 \, A -\varphi \star A.
	\label{dxxinsig}
\end{equation}
\end{proposition}

\begin{proof} 
$\Rightarrow$: We multiply \eqref{Laxpair} by $X^{-1}$ on the right, to obtain
\begin{equation} \label{19B}
	\varphi \, dX X^{-1} +\varphi A = \star \, dX X^{-1}.
\end{equation}
Using \eqref{hodgedualcoordinates}, we obtain for 
its Hodge dual,
\begin{equation}
	\varphi \star dX X^{-1} +\varphi  \star A = -\sigma dX X^{-1}.
\end{equation}
Multiplying by $\varphi$ gives 
\begin{equation}
	\left\{\begin{array}{rcl}
		\varphi^2 dX X^{-1} +\varphi^2 A - \varphi \star dX X^{-1} & =& 0 \;, \\
		\varphi  \star dX X^{-1} +\varphi  \star A +\sigma dX X^{-1}& =& 0 \;.
	\end{array}\right.
\end{equation}
Their sum simplifies to
\begin{equation}
	(\varphi^2 +\sigma) dX X^{-1} = -\varphi^2 A - \varphi \star A. 
\end{equation}

$\Leftarrow$: Recall that for $\varphi \in \mathcal{T}$ we have $\varphi^2 + \sigma \neq 0$.
Starting with
\begin{equation}
	(\varphi^2 + \sigma)dX X^{-1} = -\varphi^2 A - \varphi \star A,
\end{equation}
and multiplying by $X$ on the right gives,
\begin{equation}
	(\varphi^2 + \sigma)dX = -\varphi^2 AX - \varphi \star AX.
\end{equation}
Applying the Hodge star operator to it, and multiplying it by $\varphi$, respectively, we obtain the following equations,
\begin{equation}
	\left\{\begin{array}{l}
		\varphi^2 \star  dX + \sigma \star dX = - \varphi^2 \star  AX +\varphi \sigma A X \;,\\
		\varphi^3 dX + \varphi \sigma dX = - \varphi^3 AX - \varphi^2 \star AX.
	\end{array}\right.
\end{equation}
Subtracting the second equation from the first, we get
\begin{equation}
	(\varphi^2 + \sigma) \star dX - (\varphi^2 + \sigma)\varphi dX = (\varphi^2 + \sigma)\varphi A X
\end{equation}
which, when divided by $(\varphi^2 + \sigma)$, gives the BM linear system
\begin{equation}
	\varphi \left( d X + A X \right) = \star \, dX.
	\label{laxBM}
\end{equation}

\end{proof}

Now we show that the solvability of the linear system \eqref{Laxpair} implies the field equations \eqref{eq_motion_A}.

\begin{theorem} \label{theoremlaxtoeom} 
Let $\varphi\in\mathcal{T}$.
If the equation for the BM linear system, 
\begin{equation} %\label{Laxpair}
	\varphi  \Big( dX + A \, X \Big) = \star \, dX
\end{equation}
is satisfied, then $A$ is a solution to the field equations
\begin{equation}
	d(\rho \star A) = 0.
\end{equation}
\end{theorem}

\begin{proof} 

Using Proposition \ref{propositioniff}, 
and taking the differential of \eqref{dxxinsig},
we get
\begin{equation}
	2\varphi d\varphi \wedge dX X^{-1} + (\varphi^2+\sigma)d \Big(dX X^{-1}\Big) = -2\varphi d\varphi\wedge A - \varphi^2 dA- d\varphi\wedge \star A - \varphi d(\star A).
\end{equation}
By \eqref{lemma2} and \eqref{lemma3}, the previous equation becomes
\begin{equation}
	2\varphi d\varphi \wedge dX X^{-1} + (\varphi^2+\sigma)dX X^{-1}\wedge dX X^{-1} = -2\varphi d\varphi\wedge A - \varphi^2 dA- d\varphi\wedge \star A -  \frac{\varphi}{\rho} d(\rho \star A) - \frac{\varphi}{\rho}d\varphi\wedge \star A.
\end{equation}
Substituting \eqref{dxxinsig} into it, 
 we obtain
\begin{align}
\label{eqdvarphiA}
	& - \frac{\varphi}{\rho}d\varphi\wedge \star A = 2\varphi d\varphi \wedge\left(-\frac{\varphi^2}{\varphi^2+\sigma}A - \frac{\varphi}{\varphi^2+\sigma}\star A \right) \nonumber\\
	& +(\varphi^2+\sigma) \left[ \frac{\varphi^4}{(\varphi^2+\sigma)^2} A\wedge A + \frac{\varphi^3}{(\varphi^2+\sigma)^2}A\wedge \star A + \frac{\varphi^3}{(\varphi^2+\sigma)^2}\star A \wedge A + \frac{\varphi^2}{(\varphi^2+\sigma)^2} \star A \wedge \star A \right]  \nonumber\\
	&  + 2\varphi d\varphi \wedge A  + \varphi^2 dA + d\varphi \wedge \star A - \frac{\varphi}{\rho} d\varphi \wedge \star A.
\end{align}
Now we use the following relations for one-forms $B,C$, valid in two dimensions, 
\begin{equation}
	\left\{\begin{array}{rcl}
		\star B \wedge \star C& =& \sigma B\wedge C \;, \\
		\star B \wedge C &=& - B\wedge \star C  \;,
	\end{array}\right.
\end{equation}
as well as the relation $dA + A \wedge A = 0$ satisfied by $A = M^{-1} d M$.
These relations, together with \eqref{lemma4.1}, lead to a simplification of \eqref{eqdvarphiA},
\begin{equation}
	- \frac{\varphi}{\rho}d(\rho \star A) = -\frac{2\varphi^3}{\varphi^2+\sigma} d\varphi \wedge A - \frac{2\varphi^2}{\varphi^2+\sigma} d\varphi \wedge \star A + 2\varphi d\varphi \wedge A + d\varphi \wedge A - \frac{\varphi}{\rho} d\rho \wedge \star A.
\end{equation}
This can be written as
\begin{equation}
	- \frac{\varphi}{\rho}d(\rho \star A)  = \frac{2\varphi \sigma}{\varphi^2 + \sigma} d\varphi \wedge A + \left(\frac{2\varphi^2}{\varphi^2+\sigma}-1 \right)\star d\varphi \wedge A + \frac{\varphi}{\rho} \star  d\rho \wedge A.
\end{equation}
Then, using \eqref{hodgedualcoordinates}, this results in
\begin{equation}
	- \frac{\varphi}{\rho}d(\rho \star A)  = \left[\frac{2\varphi \sigma}{\varphi^2 + \sigma} d\varphi  + \left(\frac{\varphi^2-\sigma}{\varphi^2+\sigma}\right) \star d\varphi - \sigma\frac{\varphi}{\rho} dv \right]\wedge A.
\end{equation}
Using \eqref{lemma5} and $ \varphi\neq 0$, we obtain 
\begin{equation}
	d(\rho \star A)  = 0. 
\end{equation}

\end{proof}

%%%%%%%%%%%%%%%%%%%
\section{Monodromy matrix \label{sec:mon}}
%%%%%%%%%%%%%%%%%%%%%%%%%

The BM linear system  \eqref{Laxpair} uses, as an input, a matrix one-form $A = M^{-1} dM$  and a function 
$\varphi \in \mathcal{T}$. These quantities are defined on the 
Weyl upper half-plane.
Given a solution $X(\rho, v)$ to the BM linear system,
Breitenlohner and Maison constructed \cite{Breitenlohner:1986um} a 
 matrix ${\cal M}$ that is independent of the Weyl coordinates $(\rho, v)$.
Here we revisit and generalize 
their construction.
In doing so, we refrain from explicitly indicating
the dependence of $\varphi \in \mathcal{T}$ on the Weyl coordinates $(\rho,v)$, for ease of notation.

We will assume that we can define an involutive map $\natural$ acting on matrix functions $N(\rho,v)$ of a given order
such that
\bea
(N_1 \, N_2)^{\natural} = N_2^{\natural} \, N_1^{\natural} \;\;\;,\;\;\; (N_1 +  N_2)^{\natural} = N_1^{\natural} + N_2^{\natural} \;\;\;,\,\;\;
(a \, N)^{\natural} = a \, N^{\natural} \;\;\; \forall a \in \mathbb{R} \;,
\label{NNa}
\eea
and whenever $N \in G/H$, it coincides with the `generalized transposition' in $G/H$. We will assume that this involution
commutes with differentiation,
\begin{eqnarray}
\partial_{\rho} ( N^{\natural} ) = ( \partial_{\rho} N)^{\natural} \;\;\;,\;\;\; \partial_{v} ( N^{\natural} ) = ( \partial_v N)^{\natural} \;,
\label{comnatd}
\end{eqnarray}
and that $\mathcal{T}$ is closed under $\natural$. Examples are $N^{\natural} = N^T, \, N^{\natural} = \eta \, N^{\dagger}
\, \eta^{-1}$, where $\eta$ denotes a constant invertible matrix.

We follow the notation used in \cite{Camara:2017hez}.
The coset representative $M \in G/H$ takes the form $M = V^{\natural} V$. 
Since we take
$A$ to be of class $C^1$ in an open set, $V$ is of class $C^2$ in this set. Since $M$ is locally invertible, so is $V$.
We decompose $dV V^{-1}$ into one-forms that under the involution $\natural$ are either invariant or anti-invariant,
\begin{equation}
	dV V^{-1} = P+Q \;,
\end{equation}
where 
\begin{equation} \label{pqinvolution}
	P^{\natural} = P, \hspace{5mm}  Q^{\natural}=-Q.
\end{equation}
Using \eqref{comnatd}, we note the  relation
\begin{equation}
	2 P = VAV^{-1} \;.
	\label{PVAV}
\end{equation}
Given a solution $X$ to the BM linear system with input $(A, \varphi)$, we define \cite{Breitenlohner:1986um}
\begin{equation}
	\mathcal{P} = V\, X \;.
	\label{PVX}
\end{equation}
Now, instead of $\varphi \in  \mathcal{T}$,
 consider picking ${\chi} \in \mathcal{T}$ with  $\chi = - \sigma/ {\varphi}^{\natural}$.
 Given any solution $ {\widetilde X}$ to
  BM linear system with input $(A, \chi)$, 
let $ \widetilde{ \mathcal{P} } $ denote
 \begin{equation}
 \widetilde{ \mathcal{P} } = V\, {\widetilde X}  \;.
\end{equation}
We then use ${\cal P}$
 and $\tilde{\cal P}$ to 
define the matrix ${\cal M}$ by
\begin{equation}
	\mathcal{M} = \widetilde{\mathcal{P}}^{\natural} \, \mathcal{P}  = \widetilde{X}^{\natural} \, M\, X \;.
	\label{calMXMX}
\end{equation}
Note that $X$ is not necessarily a priori related to ${\widetilde X }$, except for the relation between $\varphi$ and 
 $\chi$. Hence, \eqref{calMXMX}  generalizes the definition given in 
 \cite{Breitenlohner:1986um}.

\begin{theorem} \label{theoremmonodromy} Let $\varphi, \chi \in \mathcal{T}$, 
with $\chi = - \sigma/ {\varphi}^{\natural}$.
Let 
$X$ and $\widetilde{X}$ denote solutions to the corresponding 
BM linear system based on 
$A = M^{-1}dM$. Then, the matrix ${\cal M} (\rho, v)$ defined in \eqref{calMXMX} satisfies
\begin{equation}
	d\mathcal{M}(\rho,v) = 0.
\end{equation}
\end{theorem}

\begin{proof} 
We consider the BM linear system with input $(A,   \varphi)$ and solution $X$.
Then, by Proposition \ref{propositioniff}, $X, A,\varphi$ satisfy,
\begin{equation}
	dX X^{-1} = -\frac{\varphi^2}{\varphi^2 + \sigma}A - \frac{\varphi}{\varphi^2 + \sigma}\star A.
	\label{dXXinv}
\end{equation}
Using \eqref{PVX}, 
we compute
\begin{eqnarray} \label{eqperpq}
	d\mathcal{P}\mathcal{P}^{-1} &=& dVV^{-1} + VdXX^{-1}V^{-1} \nonumber\\ 
	&=& Q+P - \frac{2\varphi^2}{\varphi^2+\sigma} P - \frac{2\varphi}{\varphi^2+\sigma}\star P = Q - \frac{\varphi^2-\sigma}{\varphi^2+\sigma} P - \frac{2\varphi}{\varphi^2+\sigma} \star P \;,
\end{eqnarray}
where we used the relation \eqref{dXXinv} as well as \eqref{pqinvolution} and \eqref{PVAV}.

Similarly, for the BM linear system with input $(A,  \chi)$ and solution $\widetilde X$, we obtain
\begin{eqnarray} \label{eqperpqtil}
d\widetilde{\mathcal{P}} \widetilde{\mathcal{P}}^{-1} &=&
 Q - \frac{\chi^2-\sigma}{\chi^2+\sigma} P - \frac{2\chi}{\chi^2+\sigma} \star P \;.
\end{eqnarray}

Next, using \eqref{calMXMX} and the property \eqref{comnatd},
we compute
\begin{equation}
	d\mathcal{M} = d \left(\widetilde{\mathcal{P}}^{\natural}  \right)
	\mathcal{P}+\widetilde{\mathcal{P}}^{\natural} d\mathcal{P} =\widetilde{\mathcal{P}}^{\natural} \left[\left(\widetilde{\mathcal{P}}^{\natural}\right)^{-1}
	 \left( d \widetilde{\mathcal{P}}\right)^{\natural}
	+ d\mathcal{P} \mathcal{P}^{-1} \right]\mathcal{P} \;,
\end{equation}
which, using \eqref{NNa}, can be written as
\begin{equation}
	d\mathcal{M}= \widetilde{\mathcal{P}}^{\natural} \left[\left(d\widetilde{\mathcal{P}} \widetilde{\mathcal{P}}^{-1}\right)^{\natural}+ d\mathcal{P}\mathcal{P}^{-1} \right]\mathcal{P} \;.
\end{equation}
Substituting the expression \eqref{eqperpq} and  \eqref{eqperpqtil} into this gives 
\begin{equation}
	d\mathcal{M} = \tilde{\mathcal{P}}^{\natural} \left[Q^{\natural} - 
	\left( 
	\frac{\chi^2-\sigma}{\chi^2+\sigma} P\right)^{\natural} - \left( \frac{2\chi }{\chi^2+\sigma} \star P \right)^{\natural}+Q - \frac{\varphi^2-\sigma}{\varphi^2+\sigma} P - \frac{2\varphi}{\varphi^2+\sigma} \star P \right]\mathcal{P}.
\end{equation}
Then, using \eqref{pqinvolution}, and taking into account that $\chi^{\natural}= -\sigma/\varphi$, we obtain 
\begin{equation}
	d\mathcal{M}=  \widetilde{\mathcal{P}}^{\natural} \left[-Q - \frac{-{\varphi}^2+\sigma}{{\varphi}^2+\sigma} P + \frac{2{\varphi}}{{\varphi}^2+\sigma}\star P+Q - \frac{\varphi^2-\sigma}{\varphi^2+\sigma} P - \frac{2\varphi}{\varphi^2+\sigma} \star P \right]\mathcal{P} = 0.
\end{equation}
Thus, we have shown that the matrix $\cal M$
is independent of the Weyl coordinates $(\rho, v)$ 
for any $X$ and $\widetilde{X}$ 
that solve \eqref{Laxpair} with input  $\varphi$ and $\chi$, respectively.

\end{proof}

Thus, the  matrix $\mathcal{M}$ defined in \eqref{calMXMX} is independent of the Weyl coordinates, even though
the individual factors depend on $(\rho, v)$. This result, which generalizes previous results
from  (Section 3 of)
 \cite{Breitenlohner:1986um} and 
 from (Section 5 of)
 \cite{Camara:2017hez},
 suggests that all information on $M(\rho, v)$ may be lost after multiplication by $X$ and $\widetilde X^\natural$ on the right and on the left, respectively; indeed we will later
 give an example in which, for a given $M(\rho, v)$,  it is possible to obtain in this way an arbitrary constant matrix $\cal M$, in particular the identity matrix. 
 In case we have a matrix $\mathcal{M}$, independent of $(\rho, v)$, from which one can obtain, via an appropriate factorization, a solution $M(\rho, v)$
 of the field equations, we will follow the terminology in the literature \cite{Nicolai:1991tt} and 
  say  that $\mathcal{M}$ is a {\it monodromy matrix} for $M(\rho, v)$. This will be addressed in the next section.

%%%%%%%%%%%%%%%%%%%%%%%%%%%%
\section{Canonical factorization gives a solution to the BM linear system \label{sec:mtbm}}
%%%%%%%%%%%%%%%%%%%%%%%%%%%%

In the previous section, we saw that, by \eqref{calMXMX}, one can construct matrices  $\mathcal{M} $ that are independent
of the Weyl coordinates, from given solutions $X$ and $\widetilde{X}$ of the BM
linear system with input $(A,  \varphi)$ and 
$(A,  \chi)$, respectively, where $\varphi, \chi= - \sigma/\varphi^{\natural} \in  \mathcal{T}$.
Therefore, we obtain a factorization of $\mathcal{M} $ in terms of $X, \widetilde{X}$ and $M$.

In this section, we study the reverse question.  Namely, we obtain $A(\rho,v)$ (or equivalently, $M(\rho, v)$) from a factorization of 
a matrix function $\mathcal{M} $ of the form \eqref{calMu} given below.
 Naturally, to do so, we must make certain (though very general)
assumptions, as follows (these assumptions may seem a bit technical, but in fact they are natural conditions in order for the 
class of involutions $\natural$ to be as general as possible).

Let  $\natural$ be an involution in the space of matrix functions 
 $N(\tau, \rho,v)$ of a given order, with domain  contained in $\mathbb{C} \times \mathbb{R}^+ \times \mathbb{R}$, satisfying \eqref{NNa} and 
 \eqref{comnatd} and such that it coincides with  the  `generalized transposition' if  $N(\tau, \rho, v) = N(\rho,v)$ does not depend on $\tau$ 
 and belongs to $G/H$. We assume moreover that $\natural$ commutes with the involution $\iota_\sigma$ and, if $N(\tau, \rho,v)$ and
 $N^{\natural}(\tau, \rho,v)$ are analytic with respect to $\tau$ in a region $ {\cal A} \subset \mathbb{C}$, then 
 $\partial_{\tau} (N^{\natural}) = (\partial_{\tau} N)^{\natural}$ in this region. \\
 
Let $\mathcal{M} (u)$ be a
matrix function of the complex variable $u$, and 
denote by $\mathcal{M}(\tau, \rho, v) $ the matrix 
that is obtained by composition with 
\be
u = v + \sigma \, \frac{\rho}{2} \frac{\sigma - \tau^2}{\tau} \;\;\;,\;\;\; \tau \neq 0 \;,
\label{utau}
\ee
i.e.
\be
\mathcal{M}(\tau, \rho, v) =  \mathcal{M} \left(v + \sigma \, \frac{\rho}{2} \frac{\sigma - \tau^2}{\tau} \right) 
 \label{calMu}
 \ee
 for all  $\tau \in \mathbb{C}\backslash \{0\}, \rho \in \mathbb{R}^+, v \in \mathbb{R}$ such that $\mathcal{M}(\tau, \rho, v) $ is invertible.
 
For later convenience, we note the relation of \eqref{utau}, called the spectral curve, with 
property
\eqref{spectralcurve}. We also recall the definition \eqref{invtt} of the involution
$\iota_\sigma$.\\

\noindent
{\it Assumption 1: 
There exists an open set $S$ such that, for every 
 $(\rho_0, v_0)\in S$, one can find a
simple closed curve $\Gamma$ in the $\tau$-plane, which is 
$\iota_\sigma$-invariant and 
encircles the origin, such that:
\\
For all 
$(\rho,v)$ in a neighbourhood of $(\rho_0, v_0)$, 
the matrix \eqref{calMu}, as well as its inverse, is 
analytic in a region (i.e. in an open, connected set) in the $\tau$-plane containing 
$\Gamma$, 
which we denote by $O$ and which we require to be 
 invariant under  $\iota_\sigma$
(see  Figure \ref{ribbon}), and such that  $\mathcal{M}^{\natural} (\tau, \rho, v) = \mathcal{M}(\tau, \rho, v) $ on $O$.
}\\

We denote by $D^+$ the simply connected interior region of $\Gamma$  (hence $0 \in D^+$) and by 
$D^- = \mathbb{C} \backslash (D^+ \cup \Gamma )$ the exterior region.
Recall that with these assumptions, we have, for $p \in \mathbb{C} \backslash \{0\}$,  that $p \in D^+ \Leftrightarrow  \iota_\sigma (p) \in D^-$,
see Proposition \ref{lemmainsideoutside}, and $\Gamma$ necessarily passes through the fixed points of the involution $\iota_\sigma$, 
which are $\pm i$ when $\sigma = 1$ and $\pm 1$ when $\sigma = -1$, see Corollary \ref{lemmafixedpoints}. \\

\begin{figure}[hbt!]
	\centering
	\includegraphics[scale=0.9]{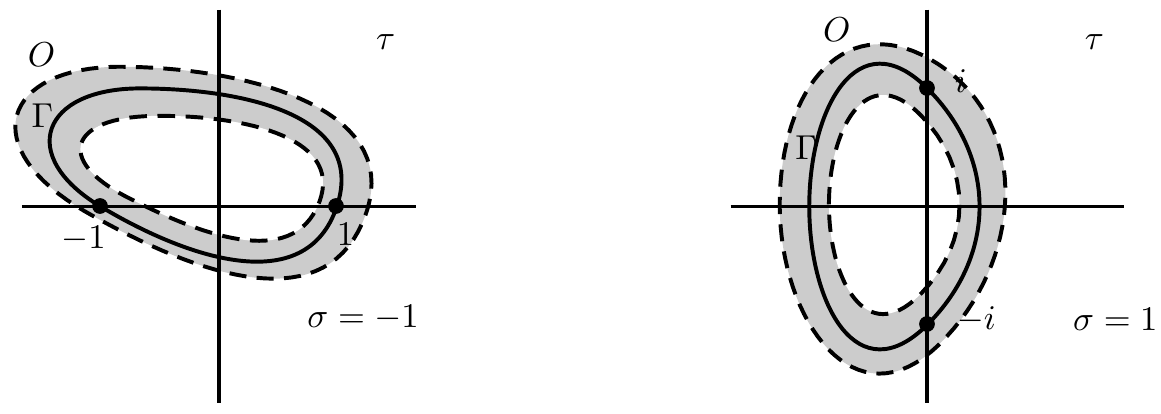}
	\caption{Examples of sets $O$ that are invariant under $\tau\mapsto- \sigma/\tau$. The curve in the middle represents $\Gamma$. 
	\label{ribbon}}
\end{figure}

\noindent
{\it Assumption 2:  If $N(\tau, \rho, v)$ is analytic in $D^+ \cup O$, then $N^{\natural} (\tau, \rho, v)$ is also analytic in $D^+ \cup O$.}
\\

Now consider $(\rho,v)$
as an arbitrary pair of parameters in a neighbourhood of $(\rho_0, v_0)$,
and take $\tau$ as the independent (complex) variable. To emphasize this aspect, we will denote the matrix in \eqref{calMu}
by $ \mathcal M_{(\rho, v)}(\tau) $.
\\

\noindent
{\it Assumption 3: With the same notation as in Assumption 1, for any $(\rho, v)$ in a neighbourhood of $(\rho_0, v_0)$,
$ \mathcal M_{(\rho, v)}(\tau) $ admits a canonical factorization of Wiener-Hopf (or Birkhoff) type with respect to $\Gamma$,
\be
 \mathcal M_{(\rho, v)}(\tau) =  M^-_{(\rho, v)}(\tau) \, M^+_{(\rho, v)}(\tau) \;\;\; {\rm on } \;\; \Gamma \;,
 \label{Mfactor}
 \ee
where $M^+_{(\rho, v)}(\tau)$ (respectively  $M^-_{(\rho, v)}(\tau)$) and its inverse are analytic and bounded in $D^+ \cup O$ (respectively 
$D^- \cup O$),
and we assume 
the normalization condition $M^+_{(\rho, v)}(0) = \mathbb{I}$. 
We set $X(\tau, \rho, v)  = M^+_{(\rho, v)}(\tau)$, so $X(0, \rho, v) = \mathbb{I}$.} 
\\

Such a factorization, if it exists, is unique; necessary and sufficient conditions for its existence were given in \cite{Camara:2017hez} 
(see also the references therein and \cite{CC}). A result given in \cite{Cardoso:2017cgi} states that,
 for any scalar function $ f(u)$ that is continuous, as well as its inverse, on $\Gamma_u$ 
(here  $ \Gamma_u $
denotes the image of $ \Gamma $ under \eqref{utau}), 
the function $f_{(\rho,v)} (\tau) $ always admits a canonical factorization with respect to $\Gamma$. It follows from this result that a canonical
Wiener-Hopf factorization always exists when the monodromy matrix is triangular \cite{CG,LS}.
It is also clear that only boundedness of $ \mathcal M_{(\rho, v)}(\tau) $ in the variable $\tau$, on $\Gamma$,
is assumed. $ \mathcal{M} (u) $ may be unbounded, see for example Section \ref{sec:schw0}.
Note that it follows from {Assumption 1} that  
the representation \eqref{Mfactor} holds for all $\tau \in O$.

Moreover, it was shown in \cite{Camara:2017hez} that under Assumptions 1 and 2, \eqref{Mfactor} can be written as
\be
 \mathcal M_{(\rho, v)}(\tau) =  X^{\natural} (- \frac{\sigma}{\tau}, \rho, v) \, M(\rho, v) \, X(\tau, \rho, v) \;\;\;,\;\;\; \tau \in \Gamma \;,
 \label{MMX}
 \ee
where 
\be 
M(\rho, v) = \lim_{\tau \rightarrow \infty}
M^-_{(\rho, v)}(\tau) = M^-_{(\rho, v)}(\infty) \;.
\ee

Our last assumption is as follows.\\

\noindent
{\it Assumption 4: $M(\rho,v)$ is of class $C^2$, and for each $\tau \in D^+ \cup O$ the matrix function $X$ 
is of class $C^2$ as a function of $(\rho, v)$, 
and $\partial X/ \partial \rho$, $\partial X/ \partial v$, $\partial X^{\natural}/ \partial \rho$ and $\partial X^{\natural}/ \partial v$ 
are analytic as functions of $\tau$ in the domain $D^+ \cup O$.
}
\\

Note that, by Assumption 3 and \eqref{MMX}, we have $X(0, \rho, v) = \mathbb{I}$ for all $(\rho, v)$  in a neighbourhood of $(\rho_0, v_0)$.
Therefore
\be
\left( \frac{ \partial X(\tau, \rho, v) }{\partial \rho} \right)_{\vert_{(0, \rho_0, v_0)}} = \left( \frac{ \partial X(0, \rho, v) }{\partial \rho} \right)_{\vert_{(\rho_0, v_0)}} = 0 \;,
\label{derXr0}
\ee
and analogously 
\be
 \left( \frac{ \partial X(0, \rho, v) }{\partial v} \right)_{\vert_{(\rho_0, v_0)}} = 0 \;.
 \label{derXv0}
 \ee
 
Now, for any $(\rho_0,v_0) \in S$, let
\be
C_{\rho_0,v_0} = \left\{ v_0 + \sigma \, \frac{\rho_0}{2} \frac{\sigma - \tau^2}{\tau} \; \; ,  \;\; \tau \in O\backslash FP_{\sigma} 
\right\} \subset \mathbb{C} \backslash W_{\rho_0, v_0} \;,
\ee
where we recall the definition of $W_{\rho_0, v_0}$ given above \eqref{omvrt}.
For any $\omega \in C_{\rho_0,v_0} $, let $\varphi_{\omega} $ and $\widetilde{\varphi}_{\omega} $ be the two functions defined in \eqref{vfi2},
and let $\mathcal{T}_{\omega} = \{ \varphi_{\omega},
\widetilde{\varphi}_{\omega} \}$. Thus, by definition, for any $\omega \in C_{\rho_0,v_0} $, $\varphi_{\omega} (\rho_0, v_0)$ and 
$\widetilde{\varphi}_{\omega} (\rho_0,v_0) = 
- \sigma / \varphi_{\omega} (\rho_0, v_0)$ lie in $O \backslash FP_{\sigma} $ and therefore, by continuity, also $\varphi_{\omega} (\rho, v)$ and 
$\widetilde{\varphi}_{\omega} (\rho,v)$ lie in $O\backslash FP_{\sigma} $ for all $(\rho, v)$ in a neighbourhood of $(\rho_0, v_0)$.
Therefore, for all $(\rho, v)$ in a neighbourhood of $(\rho_0,v_0)$, 
the matrix functions
$X(\tau, \rho, v)|_{\tau = \varphi (\rho, v)}$ and $X^{\natural} (- \sigma/\tau, \rho, v)|_{\tau = \varphi (\rho, v)}$, with $\varphi \in \mathcal{T}_{\omega} $,
are well defined and are of class $C^2$, and the equality resulting from \eqref{MMX}  by composition holds in a neighbourhood of $(\rho_0, v_0)$.

On the other hand, when replacing $\tau$ by $\varphi (\rho, v)$ of the form \eqref{vfi2}
in the expression for $u$ given in \eqref{utau}, one obtains
$\omega$ and therefore, the matrix on the left-hand side of \eqref{Mfactor} becomes $ \mathcal M (\omega)$, independent of $(\rho, v)$.

%%%%%%%%%%%%%%%%%%%
 
We now state the main result of this section as a theorem.

\begin{theorem} \label{theoremfactorizationsolution} Let Assumptions 1-4 be satisfied.
Then, $M(\rho,v)$ 
 is a solution of the field equations \eqref{eq_motion_A} on $S$.
\end{theorem}

\begin{proof}
Take any point in $S$, which we will denote by $(\rho_0, v_0)$. In the following, $(\rho, v)$ will denote an arbitrary point
in a neighbourhood
of $(\rho_0, v_0)$.
We take the contour $\Gamma$ to have the properties described under {Assumption 1}.
In the factorization \eqref{MMX}, $\tau$ varies along $\Gamma \subset O$, while
$(\rho,v)$ is kept fixed. Substituting $\tau$ by $\varphi$ of the form \eqref{vfi2} with $\omega \in C_{\rho_0,v_0}$,
we 
obtain $\mathcal{M} (\omega)$, which is independent of $(\rho, v)$, and hence differentiating with respect
to $(\rho, v)$ gives zero, i.e. $d ( \mathcal{M} (\omega) ) =0$, where $d$ denotes the differential with respect to $(\rho, v)$.
Then, multiplying by ${(X^{\natural}})^{-1}$ on the left and by $X^{-1}$ on the right, 
we obtain, 
\begin{equation} \label{equacioazero}
\begin{split}
	0 =&\left[ \left. X^{\natural}(-\frac{\sigma}{\tau},\rho,v)\right\vert_{\tau = \varphi} \right]^{-1} d\left[\left. X^{\natural}(-\frac{\sigma}{\tau},\rho,v)\right\vert_{\tau=\varphi} \right] M(\rho,v) \\
	& + dM(\rho,v) + M(\rho,v) \, d\left[\left.X(\tau,\rho,v)\right\vert_{\tau=\varphi}\right] \left.X^{-1}(\tau,\rho,v)\right\vert_{\tau=\varphi}.
\end{split}
\end{equation}
Next, we multiply this system of differential equations by $\varphi + \sigma/\varphi$,
where we recall that $\varphi \neq 0$ and 
$\varphi^2 + \sigma \neq 0$,
\begin{eqnarray}
\label{reltvar}
	0 &=&\left[ \frac{\tau^2 + \sigma}{\tau} 
		 \left.  \left( X^{\natural} \right)^{-1}
	 (-\frac{\sigma}{\tau},\rho,v) \right]\right\vert_{\tau = \varphi} 
	  d\left[\left. X^{\natural} (-\frac{\sigma}{\tau},\rho,v)\right\vert_{\tau=\varphi} \right]M(\rho,v) \\
	&&  \left. + \frac{\tau^2 + \sigma}{\tau} \right\vert_{\tau=\varphi}  \, dM(\rho,v) + \left.  \frac{\tau^2 + \sigma}{\tau}\right\vert_{\tau=\varphi} \, 
	M(\rho,v) \, d\left( \left. X(\tau,\rho,v)\right\vert_{\tau=\varphi}\right) \left. X^{-1}(\tau,\rho,v)\right\vert_{\tau=\varphi}. \nonumber
\end{eqnarray}
We evaluate
\begin{equation} \label{firstdifferentiatedterm}
\begin{split}
	d\left(\left. X^{\natural}(-\frac{\sigma}{\tau},\rho,v)\right\vert_{\tau=\varphi} \right)=& \left.\frac{\partial X^{\natural} (\tau',\rho,v)}{\partial \tau'}\right\vert_{\tau'=-\frac{\sigma}{\varphi}}\frac{\sigma}{\varphi^2}\left( \frac{\partial \varphi}{\partial \rho}d\rho + \frac{\partial\varphi}{\partial v}dv \right)\\
	& + \left. \frac{\partial X^{\natural}(\tau,\rho,v)}{\partial \rho}\right\vert_{\tau = -\frac{\sigma}{\varphi}}d\rho + \left. \frac{\partial X^{\natural}(\tau,\rho,v)}{\partial v}\right\vert_{\tau = -\frac{\sigma}{\varphi}}dv \;.
	\end{split}
\end{equation}
Using
\begin{equation}
d \varphi = \frac{\varphi}{\rho} \left( \frac{\sigma - \varphi^2}{\sigma + \varphi^2} \,  d\rho + \sigma \,  \frac{ 2 \varphi}{\sigma + \varphi^2} \, dv \right) \;,
\end{equation}
\eqref{firstdifferentiatedterm} becomes
\begin{equation}
\begin{split}
	d\left( \left. X^{\natural}(-\frac{\sigma}{\tau},\rho,v)\right\vert_{\tau=\varphi} \right) = & \left[\left.\frac{\partial X^{\natural} (\tau',\rho,v)}{\partial \tau'}\right\vert_{\tau'=-\frac{\sigma}{\tau}}\frac{\sigma}{\tau^2\rho}\left( \tau\frac{\sigma-\tau^2}{\sigma+\tau^2}d\rho + 2\sigma\frac{\tau^2}{\sigma+\tau^2}dv \right)\right.\\
	& \left.\left. + \frac{\partial X^{\natural} (-\frac{\sigma}{\tau},\rho,v)}{\partial \rho}d\rho +  \frac{\partial X^{\natural} (-\frac{\sigma}{\tau},\rho,v)}{\partial v}dv\right]\right\vert_{\tau=\varphi}.
\end{split}
\end{equation}
Next, we evaluate
\begin{eqnarray}
	d\left( \left.X(\tau,\rho,v)\right\vert_{\tau=\varphi} \right) &=& \left.\left[\left.\frac{\partial X(\tau,\rho,v)}{\partial \tau}\right.\frac{1}{\rho}\left( \tau\frac{\sigma-\tau^2}{\sigma+\tau^2}d\rho + 2\sigma\frac{\tau^2}{\sigma+\tau^2}dv \right)  \right. \right. \nonumber\\
	&& \left. \left.  \qquad 
	+ \frac{\partial X({\tau},\rho,v)}{\partial \rho}d\rho +  \frac{\partial X({\tau},\rho,v)}{\partial v}dv\right]\right\vert_{\tau=\varphi}.
\end{eqnarray}
Inserting these expressions into \eqref{reltvar} results in 
\begin{eqnarray} \label{expanded}
	 0 &= &  \left[ \left( X^{\natural} \right)^{-1}   \left(-\frac{\sigma}{\tau},\rho,v \right)
	  \left.\frac{\partial X^{\natural} (\tau',\rho,v)}{\partial \tau'}  \right\vert_{\tau'=-\frac{\sigma}{\tau}}
	 \frac{\sigma}{\tau^2\rho}[(\sigma-\tau^2)d\rho+2\sigma \tau dv] \, M(\rho,v)\right.  \nonumber\\
	 && +\frac{\tau^2+\sigma}{\tau} \left(X^{\natural} \right)^{-1}\left(-\frac{\sigma}{\tau},\rho,v \right)\left\{\frac{\partial X^{\natural}}{\partial \rho}\left(-\frac{\sigma}{\tau},\rho,v \right)d\rho + \frac{\partial X^{\natural}}{\partial v}\left(-\frac{\sigma}{\tau},\rho,v \right) dv \right\} M(\rho,v)  \nonumber\\
	&& +\frac{\tau^2+\sigma}{\tau} dM(\rho,v) + M(\rho,v) \left( \frac{\partial X}{d\tau}(\tau,\rho,v)\frac{1}{\rho}\left[(\sigma-\tau^2)d\rho + 2\sigma \tau dv\right]\right. \nonumber\\
	&&  \left.\left.\left.+ \frac{\tau^2+\sigma}{\tau}\left(\frac{\partial X}{\partial \rho}(\tau,\rho,v)d\rho + \frac{\partial X}{\partial v}(\tau,\rho,v)dv \right) \right)X^{-1}(\tau,\rho,v)\right]\right\vert_{\tau = \varphi} \;.
	\end{eqnarray}
Now evaluate \eqref{expanded} at $(\rho_0, v_0)$.
This holds for any function $\varphi$ satisfying the assumptions.
Recall that any such function is 
labelled by $\omega  \in C_{\rho_0,v_0} $. Keeping $(\rho_0, v_0)$ fixed, and varying over all $\omega  \in C_{\rho_0,v_0}$ results in scanning over all values $\tau \in O \backslash FP_{\sigma}$.
In this way, we can reinterpret \eqref{expanded} as an equality that holds for all $\tau \in \Gamma  \backslash FP_{\sigma}$ at fixed  $(\rho_0, v_0)$.
For ease of notation, we now denote  $(\rho_0, v_0)$ simply by  $(\rho, v)$, taking into account that  $(\rho_0, v_0)$ was arbitrarily fixed. Hence
\begin{eqnarray} \label{expanded2}
	 0 &= &  \left( X^{\natural} \right)^{-1}   \left(-\frac{\sigma}{\tau},\rho,v \right)
	  \left.\frac{\partial X^{\natural} (\tau',\rho,v)}{\partial \tau'}  \right\vert_{\tau'=-\frac{\sigma}{\tau}}
		 \frac{\sigma}{\tau^2\rho}[(\sigma-\tau^2)d\rho+2\sigma \tau dv] \, M(\rho,v)  \nonumber\\
	 && +\frac{\tau^2+\sigma}{\tau} \left(X^{\natural} \right)^{-1}\left(-\frac{\sigma}{\tau},\rho,v \right)\left\{\frac{\partial X^{\natural}}{\partial \rho}\left(-\frac{\sigma}{\tau},\rho,v \right)d\rho + \frac{\partial X^{\natural}}{\partial v}\left(-\frac{\sigma}{\tau},\rho,v \right) dv \right\} M(\rho,v)  \nonumber\\
	&& +\frac{\tau^2+\sigma}{\tau} dM(\rho,v) + M(\rho,v) \left( \frac{\partial X}{d\tau}(\tau,\rho,v)\frac{1}{\rho}\left[(\sigma-\tau^2)d\rho + 2\sigma \tau dv\right]\right. \nonumber\\
	&&  \left.+ \frac{\tau^2+\sigma}{\tau}\left(\frac{\partial X}{\partial \rho}(\tau,\rho,v)d\rho + \frac{\partial X}{\partial v}(\tau,\rho,v)dv \right) \right)X^{-1}(\tau,\rho,v) \;.
	\end{eqnarray}
By continuity, this equality can also be extended to the two fixed points of $\Gamma$.
Hence, \eqref{expanded2} holds $\forall \tau \in \Gamma$.

We now rewrite \eqref{expanded2} as
\begin{eqnarray} \label{poserhp}
	&&  - \left(X^{\natural} \right)^{-1} \left(-\frac{\sigma}{\tau},\rho,v \right)
	 \left.\frac{\partial X^{\natural} (\tau',\rho,v)}{\partial \tau'}  \right\vert_{\tau'=-\frac{\sigma}{\tau}}
		\frac{\sigma}{\tau^2\rho}[(\sigma-\tau^2)d\rho+2\sigma \tau dv] \, M(\rho,v)  \nonumber\\
	 && -\frac{\tau^2+\sigma}{\tau} \left(X^{\natural} \right)^{-1}\left(-\frac{\sigma}{\tau},\rho,v \right)\left[\frac{\partial X^{\natural}}{\partial \rho}\left(-\frac{\sigma}{\tau},\rho,v \right)d\rho + \frac{\partial X^{\natural}}{\partial v}\left(-\frac{\sigma}{\tau},\rho,v \right) dv \right] M(\rho,v)  \nonumber\\
	&&  -\frac{\sigma}{\tau} dM(\rho,v) \nonumber\\
&& =	 \tau \, dM(\rho,v) + M(\rho,v)  \frac{\partial X}{d\tau}(\tau,\rho,v)\frac{1}{\rho}\left[(\sigma-\tau^2)d\rho + 2\sigma \tau dv\right]X^{-1}(\tau,\rho,v) \nonumber\\
	&&  + \frac{\tau^2+\sigma}{\tau}M(\rho,v)\left(\frac{\partial X}{\partial \rho}(\tau,\rho,v)d\rho + \frac{\partial X}{\partial v}(\tau,\rho,v)dv \right)X^{-1}(\tau,\rho,v) \;\;\;, \;\;\; \tau \in \Gamma \;,
\end{eqnarray}
where we used Assumptions 2 and 4 to separate
terms in such a manner that 
the left hand side of the equality is analytic in $D^-$ (by Proposition \ref{lemmainsideoutside}, if $X(\tau,\cdot,\cdot)$ is analytic in $D^+$ 
then $X(-\frac{\sigma}{\tau},\cdot,\cdot)$ is analytic in $D^-$) except for a pole of order one at $\tau = \infty$, while the right hand side 
is analytic in $D^+$ except for a pole of order one at $\tau =0$. Since \eqref{poserhp} holds for $\tau \in \Gamma$, it constitutes
a matricial Riemann-Hilbert problem that is associated to the
canonical Wiener-Hopf factorization \eqref{MMX}. 

This associated 
Riemann-Hilbert problem is solved by means of a 
generalization of Liouville's theorem. Namely, consider a function $f$ that: equals $f_+$ in  $D^+$ and is 
analytic in $D^+$ except for a simple 
pole at $\tau =0$; that equals $f_-$ in  $D^-$ and is 
analytic in $D^-$ except for a simple pole 
at $\tau = \infty$;  and that satisfies $f_+ = f_- $ for all $\tau$ on $\Gamma$. Then, $f$
 can only be of the form
\begin{equation} \label{rationalrh}
	f(\tau) = \frac{A\tau^2 + B\tau + C}{\tau} \;,
\end{equation}
where $A,B,C$ are constants. 

In our case, $A,B,C$ are matrix one-forms that are independent of $\tau$. We then infer two equations from  \eqref{poserhp}. The first one reads
\begin{eqnarray} \label{firstrhp}
	&&  - \left(X^{\natural} \right)^{-1} \left(-\frac{\sigma}{\tau},\rho,v \right)
	 \left.\frac{\partial X^{\natural} (\tau',\rho,v)}{\partial \tau'}  \right\vert_{\tau'=-\frac{\sigma}{\tau}}
		\frac{\sigma}{\tau^2\rho}[(\sigma-\tau^2)d\rho+2\sigma \tau dv] \, M(\rho,v)  \nonumber\\
	 && -\frac{\tau^2+\sigma}{\tau} \left(X^{\natural} \right)^{-1}\left(-\frac{\sigma}{\tau},\rho,v \right)\left[\frac{\partial X^{\natural}}{\partial \rho}\left(-\frac{\sigma}{\tau},\rho,v \right)d\rho + \frac{\partial X^{\natural}}{\partial v}\left(-\frac{\sigma}{\tau},\rho,v \right) dv \right] M(\rho,v)  \nonumber\\
	&&  -\frac{\sigma}{\tau} dM(\rho,v) = \frac{A\tau^2 + B\tau + C}{\tau} \;,
\end{eqnarray}
where the left hand side is analytic in $D^-$ except for a simple pole at $\tau = \infty$.
The second equation reads
\begin{eqnarray} 
\label{secondrhp}
&& \frac{A\tau^2 + B\tau + C}{\tau} = 
 \tau \, dM(\rho,v) + M(\rho,v)  \frac{\partial X}{d\tau}(\tau,\rho,v)\frac{1}{\rho}\left[(\sigma-\tau^2)d\rho + 2\sigma \tau dv\right]X^{-1}(\tau,\rho,v) \nonumber\\
	&&  + \frac{\tau^2+\sigma}{\tau}M(\rho,v)\left(\frac{\partial X}{\partial \rho}(\tau,\rho,v)d\rho + \frac{\partial X}{\partial v}(\tau,\rho,v)dv \right)X^{-1}(\tau,\rho,v) \;,
\end{eqnarray}
where the right hand side is analytic in $D^+$ except for a simple pole at $\tau = 0$.
First we focus on \eqref{secondrhp} and rewrite it as 
\begin{eqnarray} 
 \label{equaciomiganalytic}
&& \tau \, dM(\rho,v) + M(\rho,v)  \frac{\partial X}{d\tau}(\tau,\rho,v)\frac{1}{\rho}\left[(\sigma-\tau^2)d\rho + 2\sigma \tau dv\right]X^{-1}(\tau,\rho,v) 
  \\
	&& =\frac{1}{\tau} \left[   - (\tau^2+\sigma) \,  M(\rho,v)\left(\frac{\partial X}{\partial \rho}(\tau,\rho,v)d\rho + \frac{\partial X}{\partial v}(\tau,\rho,v)dv \right)X^{-1}(\tau,\rho,v)  +  A\tau^2 + B\tau + C \right] \;. \nonumber
\end{eqnarray}
The left hand side of this equation is analytic in  $D^+$, and hence the right hand side has to be analytic at $\tau =0$.
Thus, the expression in the bracket on the right hand side must vanish for $\tau =0$.
Now recall \eqref{derXr0} and \eqref{derXv0}.
Therefore we conclude 
\begin{equation}
	C = 0 \;.
\end{equation}
Applying a similar reasoning to \eqref{firstrhp} and demanding analyticity at $\tau = \infty$, we infer
\begin{equation}
	A=0 \;.
\end{equation}
Hence, we conclude that both sides \eqref{poserhp} are equal to a $\tau$-independent matrix one-form $B$.
Since this is valid for any $(\rho,v) \in S$, we denote $B\equiv G(\rho,v)$.

Now recall that \eqref{equaciomiganalytic} holds $\forall \tau \in \Gamma$, and in particular at the fixed points $\tau^2 = - \sigma$ of the involution
\eqref{invtt}. Let us denote the two fixed points by
$\{p_F,\tilde{p}_F = -p_F\}$. We introduce two projections
\begin{equation}
P_{\pm} = \frac12 \left( 1 \pm \sigma \, p_F \, \star \right) \;,
\end{equation}
which satisfy $P_+ + P_- = {\rm id}, \, P_+ - P_- = \sigma \, p_F \, \star, \,
P^2_{\pm} = P_{\pm}, \; P_+ P_- =0$.  Then, using \eqref{hodgedualcoordinates}, we infer $d\rho +  p_F dv = 2 P_- d \rho$
as well as $\star P_{\pm} = \mp p_F \, P_{\pm}$.

Using these properties, we evaluate
\eqref{equaciomiganalytic} at $\tau=p_F$,
\begin{equation} \label{equationforG}
	p_F \, dM(\rho,v) + \frac{4 \sigma }{\rho} \, P_- d \rho  \left( M(\rho,v)  \frac{\partial X}{d\tau}(p_F,\rho,v) \, X^{-1}(p_F,\rho,v) \right) 
	 = G(\rho,v).
\end{equation}
Applying $P_+$ to this equation, and noting that the projections only act on differential forms, results in
\be
P_+ G(\rho,v) = p_F \, P_+ \, dM(\rho,v) \;.
\label{p+G}
\ee
On the other hand, evaluating \eqref{equaciomiganalytic}  at $\tau= \tilde{p}_F = -p_F$ gives
\begin{equation} \label{equationforG}
	- p_F \, dM(\rho,v) + \frac{4 \sigma }{\rho} \, P_+ d \rho \left( M(\rho,v)  \frac{\partial X}{d\tau}(-p_F,\rho,v) \, X^{-1}(-p_F,\rho,v) \right)
	 = G(\rho,v).
\end{equation}
Applying $P_-$ to this equation results in
\be
P_- G(\rho,v) = - p_F \, P_- \, dM(\rho,v) \;.
\label{p-G}
\ee
Then, adding up \eqref{p+G} and \eqref{p-G} gives
\be
 G(\rho,v) = - \star  dM(\rho,v) \;,
 \ee
while subtracting \eqref{p+G} from \eqref{p-G}  does not yield new information.

Now we return to \eqref{secondrhp} (with $A = C = 0$) and $B =  - \star  dM(\rho,v)$. Let again $(\rho_0, v_0)$ denote a point in $S$,
and let $(\rho, v)$ denote any point in a neighbourhood of $(\rho_0, v_0)$, such that 
for $\varphi \in \mathcal{T}_{\omega}$  with $\omega \in C_{\rho_0, v_0}$ we have
$\varphi (\rho, v) \subset O$.
Then, substituting  $\tau$ by $\varphi$ 
we obtain
\begin{equation}
 -\star dM(\rho,v) = \varphi \,dM(\rho,v) + 
 	\left( \frac{\sigma}{\varphi} + \varphi \right)M(\rho,v) d\left[ X(\tau,\rho,v)\vert_{\tau = \varphi} \right]X^{-1}(\tau,\rho,v)\vert_{\tau=\varphi} \;.
	\end{equation}
Multiplying by $M^{-1}$ on the left and by $\varphi$, this equation is equivalent to the 
BM linear system \eqref{Laxpair} by virtue of {Proposition \ref{propositioniff}}. 
Then, it follows by
Theorem \ref{theoremlaxtoeom} that $M(\rho,v)$ is a solution to the field equations \eqref{eq_motion_A} on $S$.
 \end{proof}

\noindent
{\it Remark:} Note that from the proof of Theorem \ref{theoremfactorizationsolution} we moreover conclude that
for $(\rho,v)$ in a neighbourhood of any $(\rho_0, v_0) \in S$, 
$M(\rho, v)$ and $X(\tau,\rho,v)|_{\tau = \varphi}$, with $\varphi $ of the form \eqref{vfi2} and with $\omega \in C_{\rho_0,v_0}$,
satisfy the
BM linear system \eqref{Laxpair}.

%%%%%%%%%%%%%%%%%%%%%%
\section{Meromorphic factorizations: a case study \label{sec:kasn}}
%%%%%%%%%%%%%%%%%%%%%%%

So far, we discussed canonical Wiener-Hopf factorizations of monodromy matrices. We focussed on canonical factorizations,
since for these we have Theorem \ref{theoremfactorizationsolution}, which guarantees that this type of factorization
yields a solution to the field equations.  We do not have an analogous theorem for other types of factorizations, such as meromorphic
factorizations \cite{CLS}. This does not mean that those other types of factorization may not also yield solutions to the field equations.
In this section, we discuss an example for which this happens. This is the example of a Kasner solution, which is a cosmological
solution to Einstein's field equations in four dimensions. It belongs to the class $\sigma = -1$. We proceed as follows.
We explicitly solve the BM linear system for a particular Kasner solution, and obtain the explicit expression for $X(\rho,v)$, c.f. \eqref{Xc}.
We then use $X(\rho,v)$ to construct
matrices $\mathcal{M}$ associated to this Kasner solution which are independent of the Weyl coordinates $(\rho, v)$.
We then pick a particular matrix $\mathcal{M}$, which has 
a factorization of the form \eqref{decompMkas},
with $M(\rho)$ corresponding to the Kasner solution. However, this factorization is 
not a canonical Wiener-Hopf factorization, since 
$X$ and $X^{\natural}$ are not analytic in the interior and exterior region, respectively, of any contour $\Gamma$ satisfying Assumption 1,
although $X$ is analytic in a
neighbourhood of $\tau =0$.
However, the matrix $\mathcal{M}$ we picked 
also admits a canonical factorization. The resulting factor $M(\rho, v)$ 
yields a different solution of the field equations.  This new solution and the Kasner solution turn out to be related by the transformation
given in Proposition \ref{lemmatangent}. Therefore, 
in certain circumstances, by performing a canonical factorization of $\mathcal{M}$ and applying an appropriate transformation to the 
associated solution, one obtains a new solution to the field equations that results from a meromorphic
factorization of $\mathcal{M}$.

We consider the field equations of General Relativity in vacuum.  They admit a cosmological Kasner solution given by
\be
\label{eq_Kasner}
 ds^2_4 = - dt^2 + \sum_{i=1}^{3} t^{2 p_i} {dx_i}^2\,\,\,,\,\,\, t> 0 \;,
\ee
where the exponents $p_i$ obey the following conditions:
\be \label{eq_cond_Kasner}
\sum_{i=1}^{3} p_i = 1 \,, \qquad \sum_{i=1}^3 {p_i}^2 =1 \,.
\ee
In the following, we take all the $p_i$ to be non-vanishing.
Without loss of generality, we will assume that $0 < p_1 < 1$. 

We now bring \eqref{eq_Kasner} into Weyl-Lewis-Papapetrou form \eqref{4dWLP} with $\sigma = -1$, 
\bea
ds^2_4 = \Delta \, {dx_3}^2 +  \Delta^{-1} \left( e^\psi (  -d \rho^2 + dv^2 ) +  \rho^2 {dx_2}^2 \right) \;,
\eea
where (we take $\rho > 0$)
\be
\begin{cases}
\rho = t^{1-p_1} \,, \\
\Delta =  \rho^{\frac{2 p_3}{1-p_1}} \,, \\
v = (1-p_1)x_1 \,,\\
e^\psi = \frac{\rho ^{\frac{2 (p_1+p_3)}{1-p_1}}}{(1-p_1)^2} \,.
\end{cases}
\ee 
Note that for $t > 0, \rho > 0$, the function $\rho (t) = t^{1-p_1}$ is one-to-one. 

The associated matrix $M \in G/H= SL(2, \mathbb{R})/SO(1,1)$ is diagonal and takes the form
\be
M(\rho)= \begin{bmatrix}
	\Delta(\rho)	& 0 \\
	0 & \Delta^{-1} (\rho)
\end{bmatrix}\,,
\label{Mrhokasn}
\ee
and the resulting matrix one-form $A = M^{-1} d M$ is given by
\be
A=   \begin{bmatrix}
	1& 0 \\
	0 & -1 
\end{bmatrix} \frac{ \partial_{\rho} \Delta }{\Delta}  \, d \rho\,.
\ee
We now specialize to the case
\bea
p_1 = p_3 = \tfrac23 \;\;\;,\;\;\; p_2 = - \tfrac13 \;,
\eea
in which case $e^{\psi} = 9 \rho^8$ and
\be
M( \rho)= \begin{bmatrix}
\rho^4  \; \; & \;\;	0\\
		0 \;  \; & \; \;	{\rho}^{-4} 
\end{bmatrix} \;.
\label{kasM}
\ee

We pick any element $\varphi \in \mathcal{T}$, c.f. \eqref{vfi2}, and we explicitly solve the associated BM linear systen
\eqref{Laxpair} for $X$.  We obtain 
\be
X (\rho, v)
= \begin{bmatrix}
	\left( \frac{\varphi(\rho,v)}{\rho} \right)^2 \, c_1 \; \; & \;\;	\left( \frac{\varphi(\rho,v)}{\rho} \right)^2 \, c_2 \\
		\left( \frac{\varphi(\rho,v)}{\rho} \right)^{- 2} \, c_3  \;  \; & \; \;	\left( \frac{\varphi(\rho,v)}{\rho} \right)^{- 2} \, c_4
\end{bmatrix}\,,
\label{Xc}
\ee
as can be verified by substituting this expression into \eqref{Laxpair}.
Here, $c_1,c_2,c_3,c_4 \in \mathbb{C}$ are arbitrary integration constants.

For the theory at hand, the involution $\natural$ acts as transposition on matrices. We therefore obtain for the solution
$\widetilde{X}$ of the linear system \eqref{Laxpair} based on $\widetilde{\varphi} = 1/\varphi$, 
\be
{\widetilde X}^\natural (\rho, v) 
= \begin{bmatrix}
	\left( \frac{1}{\rho\, \varphi(\rho,v)} \right)^2 \,   {\tilde c}_1 \;\; &\;\;  -\left( \rho\, \varphi(\rho,v) \right)^{2} \,  {\tilde c}_3 \\
-	\left( \frac{1}{\rho\, \varphi(\rho,v)} \right)^2 \,  {\tilde c}_2 \;\; &\;\;  \left( \rho\, \varphi(\rho,v) \right)^{2} \,  {\tilde c}_4
\end{bmatrix}\,,
\ee
where the ${\tilde c}_i  \in \mathbb{C}$ denote an {\it independent} set of integration constants, unrelated to the $c_i$.
Next we compute the combination
\be
\mathcal{M} = {\widetilde X}^{\natural}
\, M \, X
=
 \begin{bmatrix}
{\tilde c}_1 {c_1}  - {\tilde c}_3 {c_3} \;\; & \;\; {\tilde c}_1 c_2 - {\tilde c}_3 c_4 \\
c_3 {\tilde c}_4 - c_1 {\tilde c}_2 \;\; & \;\; {\tilde c}_4 {c_4} - {\tilde c}_2 {c_2}
\end{bmatrix} \,,
\label{monkasner}
\ee
which is independent of the Weyl coordinates $(\rho, v)$. 
Thus, 
we may obtain an arbitrary constant
matrix $\mathcal{M}$, in particular the identity matrix, and the information about $M(\rho, v)$
is lost. However, as mentioned before, by choosing an appropriate 
matrix $\mathcal{M}$ that is independent of $(\rho, v)$,
we can recover a solution $M(\rho, v)$ either by Wiener-Hopf or by meromorphic factorization. To illustrate this, we now describe how
this arises in a specific example.

Let us now choose these constants as follows:  $c_2 = {\tilde c}_2 = c_3 = {\tilde c}_3  =0$ and $c_1 = {\tilde c}_1 =  (2 \omega)^2, \, c_4 = {\tilde c}_4 = (2\omega)^{-2}$, with
$\omega \in \mathbb{C}$. Then, \eqref{monkasner} becomes
\be
\mathcal{M} (\omega) = \begin{bmatrix}
	 (2 \omega)^4
		\; \; & \;\;	0 \\
		0\;  \; & \; \;	 (2 \omega)^{-4} 
\end{bmatrix} = \left[ X^{\natural} (\tfrac{1}{\tau}, \rho, v) \, M(\rho) \, X(\tau, \rho, v) \right]_{\vert \tau = \varphi_{\omega}(\rho,v) } \;,
\label{M2kmon}
\ee
where,  for all $(\rho, v)$ taken in a neighbourhood of  $(\rho_0, v_0)$, 
$\varphi_{\omega} \in \mathcal{T}$ with $\omega \neq v_0 \pm \rho_0$, c.f. \eqref{vfi2}. The expression for $X(\tau, \rho, v)$,
\be
X(\tau, \rho, v)= \begin{bmatrix}
	( 2 v \, \frac{\tau}{\rho} + 1 + \tau^2 )^2  \; \; & \;\;	0  \\
		0   \;  \; & \; \;	 ( 2 v \, \frac{\tau}{\rho} + 1 + \tau^2 )^{-2}
		\end{bmatrix}\,,
\ee
is obtained by making use of the algebraic relation \eqref{spectralcurve} in \eqref{Xc}. Also using it on the left hand side of 
\eqref{M2kmon} yields
\be
\mathcal{M} \left(v + \frac{\rho}{2} \frac{(1 + \tau^2)}{\tau} \right) _{\vert \tau = \varphi_{\omega}(\rho,v)}
= \left[ X^{\natural} (\tfrac{1}{\tau}, \rho, v) \, M(\rho) \, X(\tau, \rho, v) \right]_{\vert \tau = \varphi_{\omega}(\rho,v) } \;.
\label{decompMkas}
\ee
Now, when viewed as a factorization of $\mathcal{M} \left(v + \frac{\rho}{2} \frac{(1 + \tau^2)}{\tau} \right) $ with respect to a contour $\Gamma$,
\eqref{decompMkas} describes a meromorphic factorization: although
$X$ is  normalized at $\tau =0$,  i.e. $X(\tau =0, \rho, v) = \mathbbm{1}$, it has 
a double pole in the 
interior of the curve $\Gamma$, located at one of the zeroes of $2 v \, \frac{\tau}{\rho} + 1 + \tau^2 $. The two zeroes of this quadratic polynomial in $\tau$
are at $\tau = ( -v \pm \sqrt{v^2 - \rho^2}) / \rho$. Let us denote them by $\tau_1 =  ( -v + \sqrt{v^2 - \rho^2}) / \rho$ and
 $\widetilde{\tau}_1 =  ( -v - \sqrt{v^2 - \rho^2}) / \rho$.
They satisfy $\tau_1 \widetilde{\tau}_1 =1$. Since the curve $\Gamma$ is $\iota_\sigma$-invariant, then, by Proposition \ref{lemmainsideoutside},
if $\tau_1$ is located in the interior of the curve $\Gamma$, $\widetilde{\tau}_1$ is located in the exterior region, and vice-versa (recall that 
$\tau_1$ and $\widetilde{\tau}_1$ cannot
lie on $\Gamma$).
Hence, we have obtained the Kasner solution \eqref{Mrhokasn} from a meromorphic factorization
of $\mathcal{M} \left(v + \frac{\rho}{2} \frac{(1 + \tau^2)}{\tau} \right) $.

Now let us discuss the canonical factorization of the same matrix $\mathcal{M} \left(v + \frac{\rho}{2} \frac{(1 + \tau^2)}{\tau} \right) $
with respect to $\Gamma$. If 
$\widetilde{\tau}_1$ is located in the exterior region of $\Gamma$,
the resulting factors in  \eqref{MMX} are
\be
\label{McXc}
M_c (\rho, v) = \begin{bmatrix}
	 \left(\rho\, \widetilde{\tau}_1 (\rho,v) \right)^4  
		\; \; & \;\;	0 \\
		0\;  \; & \; \;	 \left(\rho \,\widetilde{\tau}_1 (\rho,v)  \right)^{-4}  
\end{bmatrix}
\;,\;
X_c(\tau, \rho, v)= \begin{bmatrix}
	\frac{(\tau- \widetilde{\tau}_1)^4}{\widetilde{\tau}_1^4} \; \; & \;\;	0  \\
		0   \;  \; & \; \;	\frac{\widetilde{\tau}_1^4}{(\tau - \widetilde{\tau}_1)^4} 
		\end{bmatrix} \;,
				\ee
and analogously if $\tau_1$ is in the exterior region of $\Gamma$.
The subscript $c$ refers to the canonical factorization.  The matrix $M_c (\rho, v)$ in  \eqref{McXc} describes a solution
to the field equations \eqref{eq_motion_A} that is different from the Kasner solution $M (\rho)$ described by \eqref{Mrhokasn}.
Both matrices are related by
\be
\label{McMK}
M (\rho) = M_c(\rho,v)  \begin{bmatrix}
	  \widetilde{\tau}_1^{-4} (\rho,v) 
		\; \; & \;\;	0 \\
		0\;  \; & \; \;	 \widetilde{\tau}_1^4 (\rho,v) 
\end{bmatrix} \;.
\ee
This is a particular instance of the following result:
we can construct new solutions to the field equations
by multiplying  $M_c (\rho,v)$ by
\be
 \begin{bmatrix}
	K \,  \rho^{\alpha} \varphi^{\beta}  
		\; \; & \;\;	0 \\
		0\;  \; & \; \;	K^{-1} \,   \rho^{-\alpha} \varphi^{-\beta}
\end{bmatrix} \;\;\;,\;\;\; \alpha \in \mathbb{R}  \;,\; \beta \in \mathbb{Z} \;,\; K \in \mathbb{R}\backslash \{0\} \;,
\label{MMcinv}
\ee
with $\varphi \in \mathcal{T}$.
The resulting matrix $M$ as well as \eqref{MMcinv} satisfy the field equations \eqref{eq_motion_A} by virtue of Proposition \ref{lemmatangent}.
The case \eqref{McMK} corresponds to taking $\alpha = 0, \beta = 4, K =1
$ and $\varphi (\rho, v) = \widetilde{\tau}_1 (\rho, v)$. 
Note that neither the canonical factorization nor the solution given by \eqref{MMcinv} are
valid
on the lines $v = \pm \rho$, for which the points $\tau_{1}, \widetilde{\tau}_1$ coincide with the fixed points $\pm 1$. However, upon multiplication by a
matrix as in  \eqref{McMK}, we obtain a solution that is valid for all $(\rho, v)$ on the Weyl upper half-plane.

More generally, given matrix $M_1$ and $M_2$ that solve the field equations \eqref{eq_motion_A}, the product $M_1 M_2$ will also 
solve the field equations \eqref{eq_motion_A} provided that
\be
d \left( \rho \, M_2^{-1} M_1^{-1}  (\star d M_1) M_2 \right) = 0 \;,
\ee
This is the case when both $M_1$ and $M_2$ are diagonal matrices, or for instance, when  $M_2^{-1} M_1^{-1}  (d M_1) M_2 = M_1^{-1}  d M_1$.

%%%%%%%%%%%%%%%%%%
%%%%%%%%%%%%%%%%%%
\section{The Schwarzschild monodromy matrix \label{sec:schw}}
%%%%%%%%%%%%%%%%%%
%%%%%%%%%%%%%%%%%%

In this section, we discuss the family of solutions to 
the field equations of General Relativity in vacuum that results from
the canonical factorization of the monodromy matrix 
\eqref{monschwarzeps} with $\epsilon = 1$, which we call
the Schwarz\-schild monodromy matrix,
\be
\mathcal{M} (u) = \begin{bmatrix}
	\sigma \, \frac{u -m}{u +m} & 0 \\
	0 & \sigma \, \frac{u+m}{u-m}
\end{bmatrix} \;\;\;,\;\;\; m \in \mathbb{R}^+ \;,
\label{monschwarz}
\ee
for both cases $\sigma = \pm 1$. The involution $\natural$ acts as transposition on matrices. 
As we shall see, the case $\sigma = -1$ is more intricate than the case $\sigma = 1$.

Substituting $u$ by the expression on the right hand side of the 
spectral curve relation \eqref{utau}, the resulting matrix  is
\be
\mathcal{M}_{(\rho,v)} (\tau)
= \begin{bmatrix}
	 \sigma \,  \frac{\left(\tau - \tau_1\right) \left(\tau + \sigma/\tau_1 \right) }{\left(\tau - \tau_2\right) \left(\tau + \sigma/\tau_2 \right)} 
 & 0 \\
	0 & 
	 \sigma\,   \frac{\left(\tau - \tau_2\right) \left(\tau + \sigma/\tau_2 \right) }{\left(\tau - \tau_1\right) \left(\tau + \sigma/\tau_1 \right)} 
	 \end{bmatrix} \;\;\;,\;\;\; 
\ee
where
\be
\tau_1(\rho,v) = - \sigma \,  \frac{m -v + \sqrt{\left(m-v\right)^2 + \sigma \,\rho^2}}{\rho} \,, 
\quad \tau_2(\rho,v)=  - \sigma \, \frac{ - (m+v)+ \sqrt{\left(m+v\right)^2 +\sigma  \rho^2} }{\rho}\,,
\label{tt12}
\ee
i.e. $\tau_1 = \varphi_{m} (\rho, v)$  and $\tau_2 =  \varphi_{-m} (\rho, v)$  if we take
\be
\varphi_{\alpha} (\rho, v) = - \sigma \, \frac{(\alpha-v) + \sqrt{(\alpha-v)^2+\sigma\rho^2}}{\rho} \;\;\;,\;\;\; \alpha \in \mathbb{C} \;.
\label{fial}
\ee

Now we perform the canonical factorization of $\mathcal{M}_{(\rho,v)} (\tau)$ with respect to a closed contour $\Gamma$ in the $\tau$-plane,
satisfying the 
following requirements:
\begin{enumerate}
\item $\Gamma$ passes through the fixed points of the transformation $\tau \mapsto  - \frac{\sigma}{\tau}$;
\item $\Gamma$ encircles $\tau=0$;
\item if $\tau \in \Gamma$ then $ - \frac{\sigma}{\tau} \in \Gamma$;
\item $\mathcal{M}_{(\rho,v)} (\tau)$ is analytic in an open set containing $\Gamma$.
\end{enumerate}
For any $(\rho,v)$ in the Weyl upper half-plane such that $ \tau_1(\rho,v), \tau_2(\rho,v) \notin FP_{\sigma}$ (where we recall that $FP_{\sigma}$ denotes the set of fixed points of
the involution \eqref{invtt}), the contour $\Gamma$ can be chosen in such a way as to bypass $\tau_1$ and $\tau_2$. There are four possible classes
of contours from which $\Gamma$ can be chosen:
\begin{itemize}
\item[$(i)$] $\tau_1$ and $\tau_2$ are inside the contour $\Gamma$;
\item[$(ii)$] $\tau_1$ is outside and $\tau_2$ is inside of $\Gamma$;
\item[$(iii)$] both $\tau_1$ and $\tau_2$ are outside of $\Gamma$;
\item[$(iv)$] $\tau_1$ is inside and $\tau_2$ is outside of $\Gamma$.
\end{itemize}
Factorizing with respect to $\Gamma$ we obtain, for each of these cases,
\be
\frac{\left(\tau - \tau_1\right) \left(\tau + \sigma/\tau_1 \right) }{\left(\tau - \tau_2\right) \left(\tau + \sigma/\tau_2 \right)}  = m_- (\tau) \, m_+ (\tau) 
\ee
and
\be
M (\rho,v)
= \begin{bmatrix}
	 \sigma \, m_- (\infty) \; & \;  0 \\
	0 \;  & \;
	\sigma \, m_-^{-1} (\infty) 
	 \end{bmatrix}  =  
	 \begin{bmatrix}
	\Delta \; & \;  0 \\
	0 \;  & \;
	\Delta^{-1}  
	 \end{bmatrix} \;,
\ee
where 
\begin{itemize}
\item [$i)$] for a contour in class $(i)$
\be
m_+(\tau) = \frac{\tau_1}{\tau_2} \, \frac{\tau + \sigma / \tau_1}{\tau + \sigma / \tau_2} \;\;\;,\;\;\; m_-(\tau) = \frac{\tau_2}{\tau_1} \, \frac{\tau -  \tau_1}{\tau -  \tau_2} 
\;\;\;,\;\;\; \Delta = \sigma \, \tau_2/\tau_1 \;\;\;,
\ee

\item[$ii)$] for a contour in class $(ii)$

\be
m_+(\tau) = \frac{1}{\tau_1 \, \tau_2} \, \frac{\tau - \tau_1}{\tau + \sigma / \tau_2} \;\;\;,\;\;\; m_-(\tau) = \tau_1 \tau_2 \, 
\frac{ \tau + \sigma /   \tau_1}{\tau -  \tau_2} 
\;\;\;,\;\;\; \Delta = \sigma \, \tau_1 \tau_2  \;\;\;,
\ee

\item[$iii)$] for a contour in class $(iii)$

\be
m_+(\tau) = \frac{\tau_2}{\tau_1} \, \frac{\tau -  \tau_1}{\tau -  \tau_2} \;\;\;,\;\;\; m_-(\tau) = \frac{\tau_1}{\tau_2} \, \frac{\tau + \sigma/  \tau_1}{\tau + \sigma/  \tau_2} 
\;\;\;,\;\;\; \Delta = \sigma \, \tau_1/\tau_2 \;\;\;,
\ee

\item[$iv)$]  for a contour in class $(iv)$

\be
m_+(\tau) = \tau_1 \tau_2 \, \frac{\tau + \sigma / \tau_1}{\tau -  \tau_2} \;\;\;,\;\;\; m_-(\tau) = \frac{1}{\tau_1 \tau_2}  \, \frac{\tau -  \tau_1}{\tau + \sigma/ \tau_2} 
\;\;\;,\;\;\; \Delta = \sigma \, \frac{ 1 }{\tau_1 \tau_2} \;\;\;.
\ee

\end{itemize}
The resulting four-dimensional space-time metrics are, in Weyl-Lewis-Papapetrou form, given by
\bea
ds_4^2 = - \sigma \Delta \, dt^2 + \Delta^{-1} 
\left(  e^\psi \, ds_2^2
 + \rho^2 d\phi^2 \right) \;,
 \label{solschw}
\eea
with $\Delta > 0$, with $ds_2^2$ given by either \eqref{weylframe} or \eqref{weylframe2}, and with
$\psi$ obtained by solving \eqref{eq_diff_eq_psi}.

We will now first discuss the case $\sigma = 1$, and subsequently the more intricate case $\sigma = -1$.

%%%%%%%%%%%%%%%%%%
%%%%%%%%%%%%%%%%%%
\subsection{$\sigma = 1$ \label{sols1}}
%%%%%%%%%%%%%%%%%%%
When $\sigma =1$, we have
\be
\tau_1(\rho,v) = \frac{v-m - \sqrt{\left(v-m\right)^2 + \rho^2}}{\rho} \,, \qquad \tau_2(\rho,v)= \frac{v+m - \sqrt{\left(v+m\right)^2 + \rho^2}}{\rho}\,,
\ee
which are real for any $(\rho,v)$ in the Weyl upper half-plane. Hence, $\tau_1$ and $\tau_2$ 
can never coincide with the two fixed points $\pm i$ of the involution $\tau \mapsto -1/\tau$. As a consequence, any of the solutions to the  field equations obtained
by canonical factorization of $\mathcal{M}_{(\rho,v)} (\tau)$ is valid in all of the Weyl upper half-plane.

The relative positions of $\tau_1$ and $\tau_2$ on the real axis of the $\tau$-plane depends on $(\rho, v)$ in the Weyl upper half-plane, as follows:
\begin{enumerate}
\item[ $a)$] for $v \leq  -m \implies \tau_1 < \tau_2 \leq -1  $;
\item[ $b)$] for $-m < v < m \implies \tau_1 < -1 < \tau_2  < 0 $;
\item[ $c)$] for $v \geq m \implies -1 \leq  \tau_1 < \tau_2  < 0$.
\end{enumerate}
Possible choices of contours for the four classes of contours $(i)-(iv)$ are depicted in Figure \ref{contout}.

\begin{figure}[hbt!]
	\centering
	\includegraphics[scale=0.4]{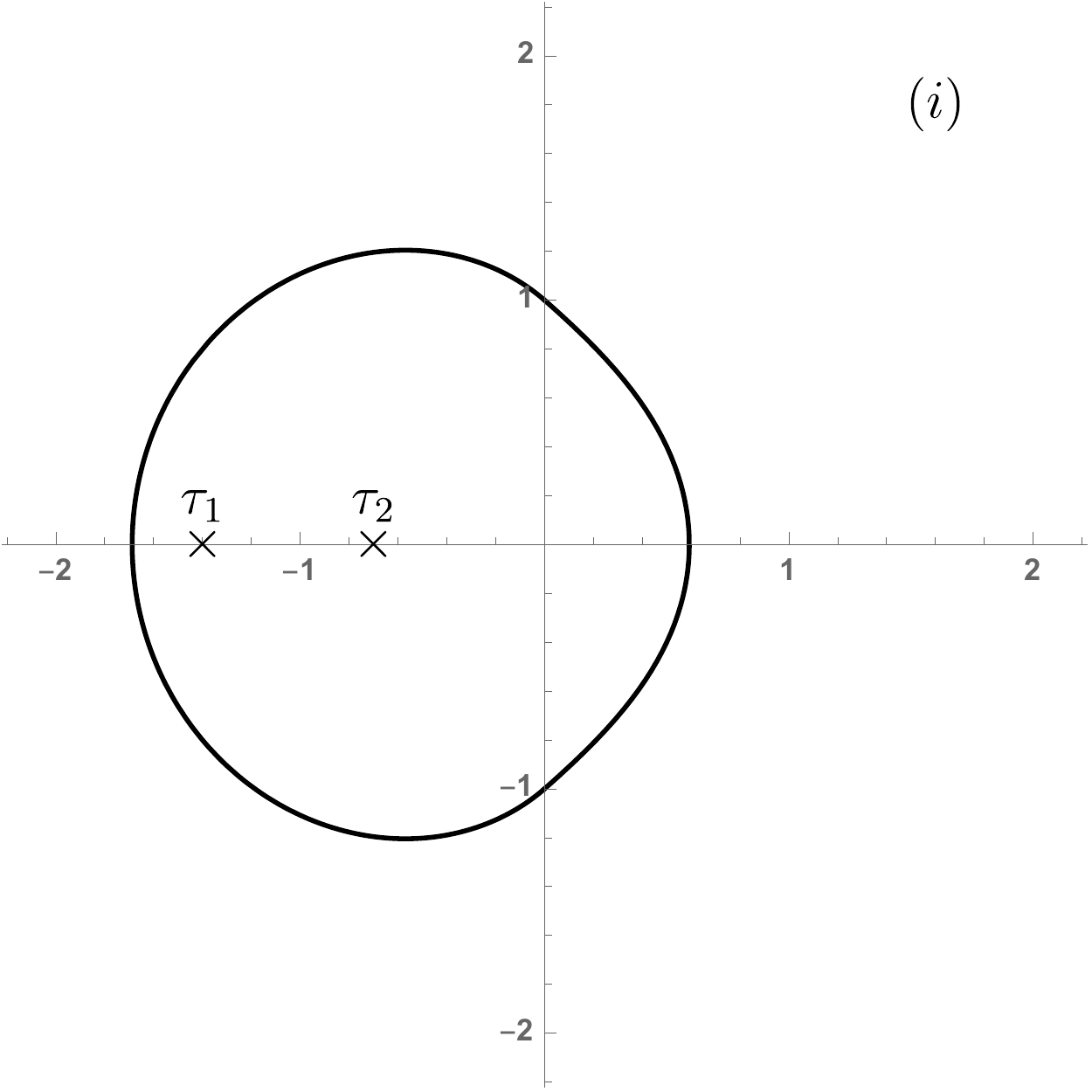}
	\includegraphics[scale=0.4]{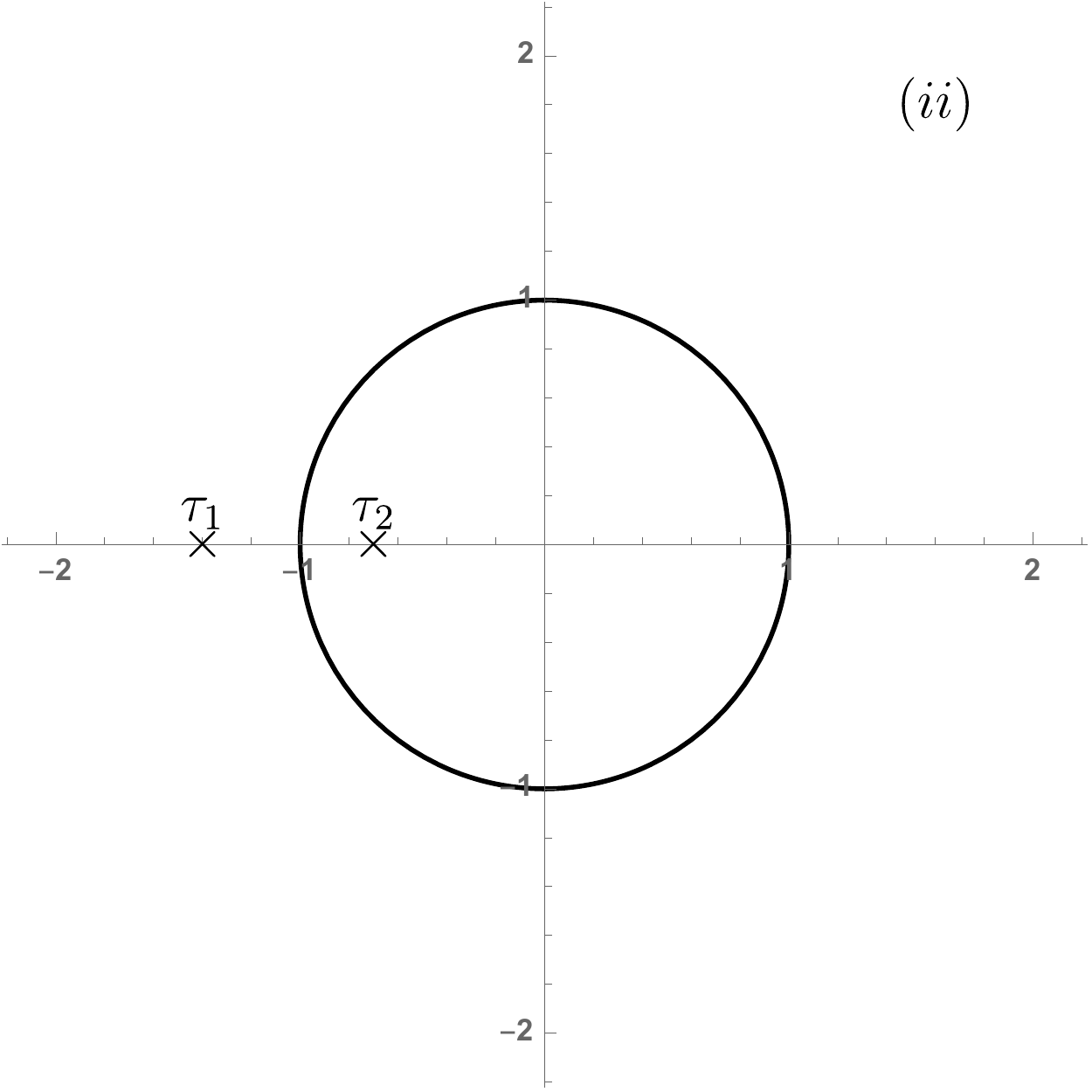}
	\includegraphics[scale=0.4]{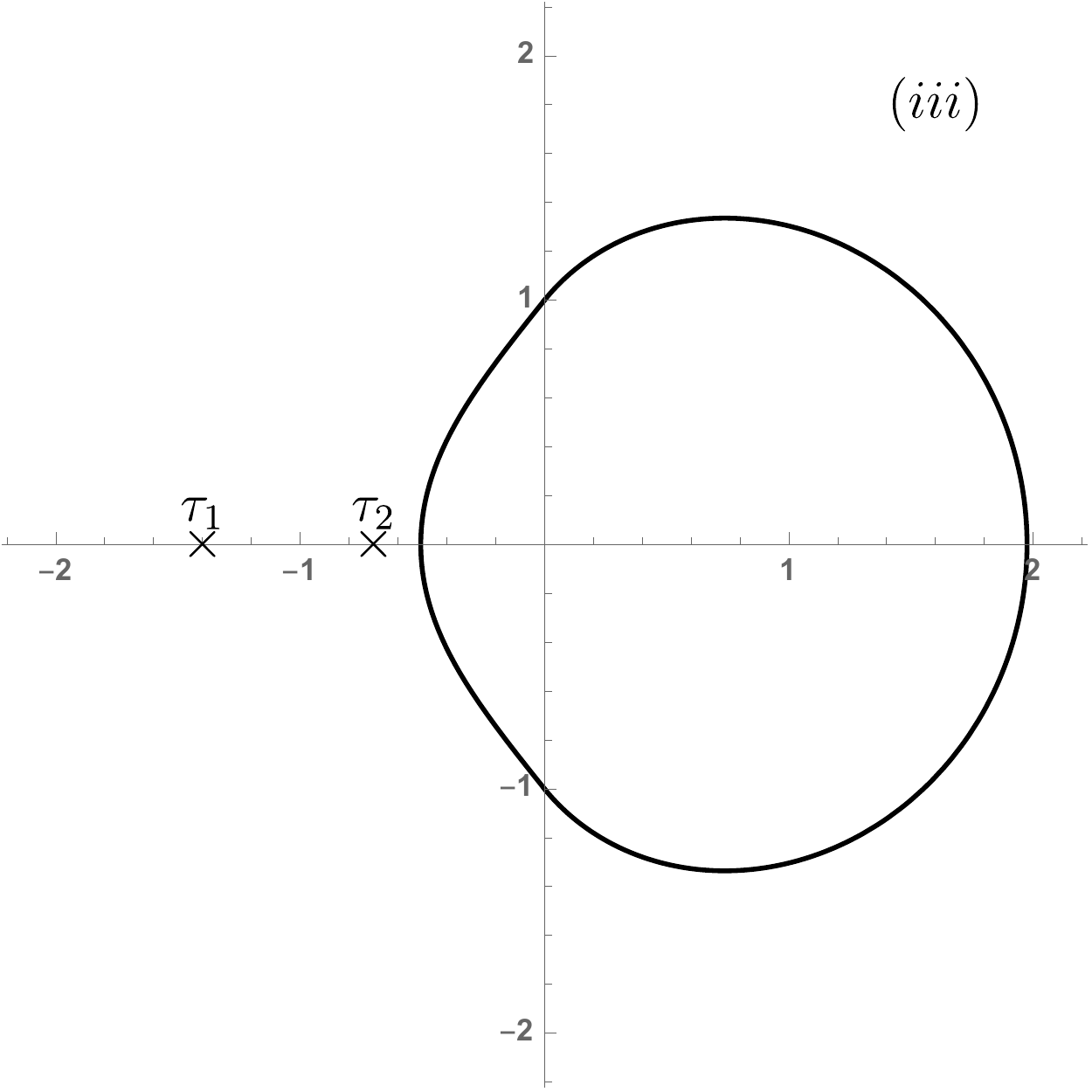}
	\includegraphics[scale=0.4]{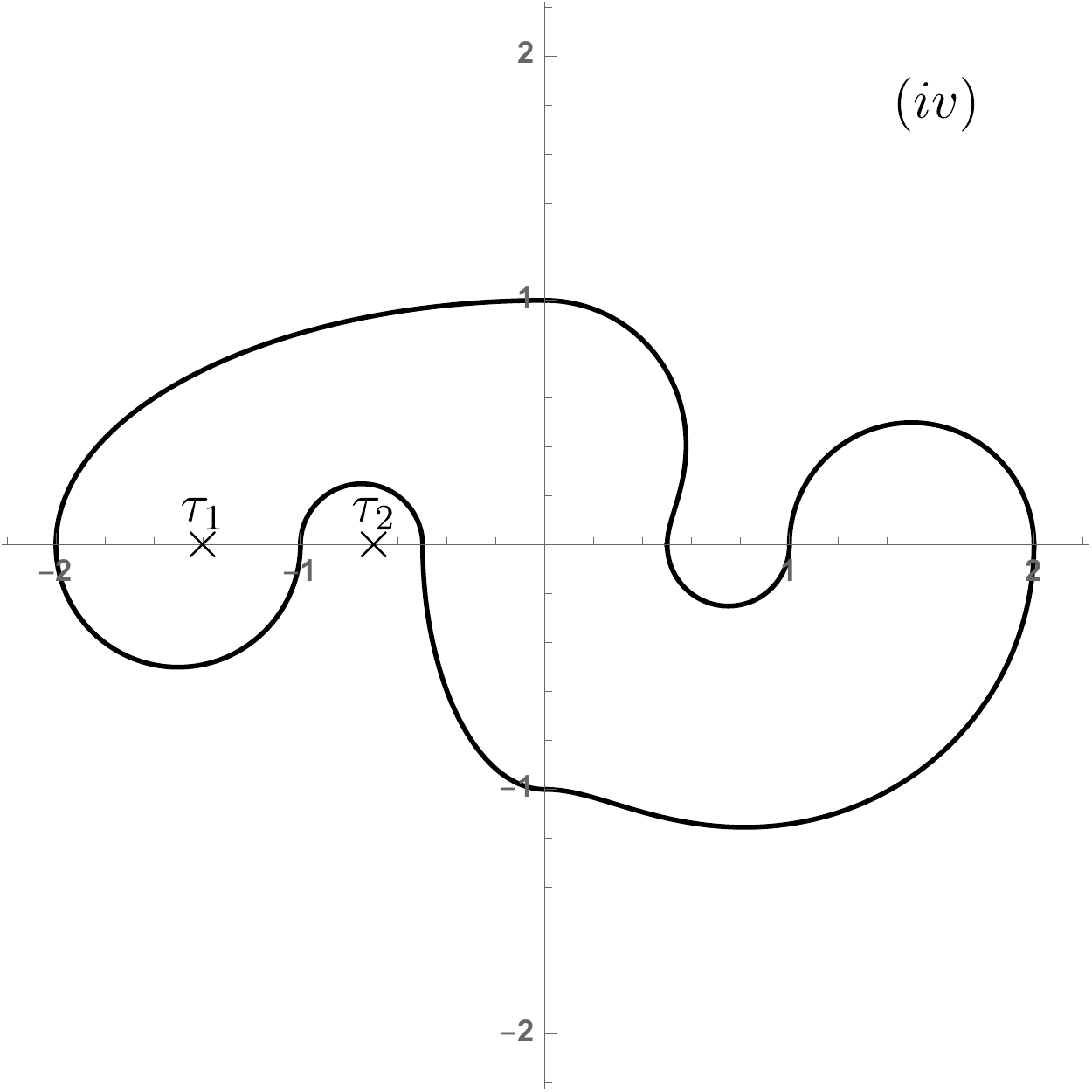}
		\caption{$-m < v < m$: four distinct choices of contours.
		\label{contout}}
\end{figure}

Next, we describe the four solutions
\bea
ds_4^2 = -  \Delta \, dt^2 + \Delta^{-1} 
\left(  e^\psi \, \left(   d \rho^2 + dv^2 \right)
 + \rho^2 d\phi^2 \right) 
 \label{solus1}
\eea
that result from factorizing the Schwarzschild monodromy matrix \eqref{monschwarz} with $\sigma = 1$ with respect
to a contour $\Gamma$ belonging to one of the classes $(i)-(iv)$. To distinguish between these four solutions, we will add the
subscript $i, ii, iii, iv$ to  $\Delta$ and $\psi$.\\

\noindent
{\bf Case $(i)$}: for $\Gamma$ belonging to class $(i)$, we obtain
\begin{eqnarray}
\Delta_i (\rho, v) &=&  \frac{ v + m - \sqrt{\left( v + m \right)^2 +  \rho^2} }
{ v - m -  \sqrt{\left(v - m \right)^2 +  \rho^2} } =   \frac{ v - m + \sqrt{\left( v - m \right)^2 +  \rho^2} }
{ v + m +  \sqrt{\left(v + m \right)^2 +  \rho^2} } \;, 
\nonumber\\
\psi_i (\rho, v) &=& \log \left[ \tfrac12  \frac{v^2 + \rho^2 - m^2 }{ \sqrt{\left(v-m\right)^2 + \rho^2} \,  \sqrt{\left(v+m\right)^2 + \rho^2}  } +\tfrac12   \right] \;, 
\nonumber\\
 K_i(\rho,v) &=& 48 m^2\left(\frac{2}{2m + \sqrt{\left(v-m\right)^2+\rho^2} +\sqrt{\left(v+m\right)^2+\rho^2} }\right)^6 \;,
 \label{expr1}
\end{eqnarray}
where $\psi = \psi_i$ is obtained by solving \eqref{eq_diff_eq_psi}, up to a constant which we set to zero, and where $K = K_i$ is
the Kretschmann scalar $K = R_{\mu \nu \delta \eta} R^{\mu \nu \delta \eta} $. Note that the expressions \eqref{expr1} are invariant under
$v \mapsto -v$.

We note that $\Delta_i$ is bounded for all $(\rho, v)$ in the Weyl upper half-plane, and bounded away from zero except in a neighbourhood of the segment
$\rho =0, \, - m < v < m$, where $\Delta_i \rightarrow 0$ as $\rho \rightarrow 0$ with $v$ fixed.

Using the bijection from the domain $ r > 2m, \, 0 < \theta < \pi$ in spherical coordinates onto the Weyl upper half-plane, given by
\bea
\rho =
\sqrt{r^2 - 2m r }  \, \sin \theta \;\;\;,\;\;\; v = (r -m) \cos \theta \;,
\label{sphext}
\eea
$\Delta_i$ becomes $\Delta_i = 1 - 2m/r$, which results in the line element
\begin{equation}
ds_4^2 = - \left(1 - \frac{2m}{r} \right) dt^2 + \left(1 - \frac{2m}{r} \right)^{-1}  dr^2 + r^2 \left( d \theta^2 +  \sin^2 \theta \, d \phi^2 \right) \:.
\label{schwarzmet}
\end{equation}
Thus, the solution \eqref{solus1} with $(\Delta, \psi) = (\Delta_i, \psi_i)$ describes the exterior region of the Schwarzschild black hole
in Weyl coordinates, and the segment $] -m ,m [$ on the $v$-axis of the Weyl upper half-plane corresponds to the Killing horizon of the Schwarzschild black hole. This solution is an example of an $AI$-metric, c.f. Appendix \ref{sec:classA}.
\\

\noindent
{\bf Case $(ii)$}: for $\Gamma$ belonging to class $(ii)$, we obtain
\begin{eqnarray}
\Delta_{ii} (\rho, v) &=&  \frac{ \left(  v - m -  \sqrt{\left(v - m \right)^2 +  \rho^2} \right) \left(
v + m - \sqrt{\left( v + m \right)^2 +  \rho^2} \right) }{\rho^2}
\nonumber\\
\psi_{ii} (\rho, v) &=& \log \left[ \tfrac12  \frac{m^ 2 - v^2 - \rho^2  }{ \sqrt{\left(v-m\right)^2 + \rho^2} \,  \sqrt{\left(v+m\right)^2 + \rho^2}  } +\tfrac12   \right] \;, 
\nonumber\\
 K_{ii}(\rho,v) &=& 48 m^2\left(\frac{2}{2m + \sqrt{\left(v+m\right)^2+\rho^2} -\sqrt{\left(v-m\right)^2+\rho^2} }\right)^6 \;,
 \label{exp2}
\end{eqnarray}
where $\psi = \psi_{ii}$ is obtained by solving \eqref{eq_diff_eq_psi}, up to a constant which we set to zero, and where $K = K_{ii}$ is
the Kretschmann scalar $K = R_{\mu \nu \delta \eta} R^{\mu \nu \delta \eta} $. 

Let us consider these expressions in the limit $\rho \rightarrow 0$ with $v$ fixed. For fixed $v \in ] -m, m[$,  
$\Delta_{ii}$ and $K_{ii}$ are bounded away from $0$, while $\psi_{ii} \rightarrow 0$. For $v > m$, we have $\Delta_{ii} \rightarrow 0, \, K_{ii} \rightarrow 3/(4m^4)$,
while for $v < -m$ we have $\Delta_{ii} \rightarrow \infty, \, K_{ii} \rightarrow \infty$. Thus, for $\rho =0, v > m$ there is a Killing horizon, and for 
$\rho =0, v <- m$ a curvature singularity.

Using the change of coordinates 
\bea
\rho =
\sqrt{2m \varrho - \varrho^2 }  \, \sinh \vartheta \;\;\;,\;\;\; v = (\varrho -m) \cosh \vartheta \;,
\eea
with $\varrho \in ]0, 2m[, \, \vartheta \in ]0, \infty[$, the space-time metric in four dimensions takes the form 
\begin{equation}
ds_4^2 =  \left(1 - \frac{2m}{\varrho} \right) dt^2-  \left(1 - \frac{2m}{\varrho} \right)^{-1}  d\varrho^2 + \varrho^2 \left( d \vartheta^2 + \sinh^2 \vartheta \, d \phi^2 \right) \:,
\label{schwarzmet2}
\end{equation}
which describes the `interior' region of the $AII$-metric \eqref{hyp}.
Although this metric resembles the Schwarzschild metric \eqref{schwarzmet2}, it has a different isometry group. \\

\noindent
{\bf Case $(iii)$}: 
for $\Gamma$ belonging to class $(iii)$, we obtain
\begin{eqnarray}
\Delta_{iii} (\rho, v) &=& \frac{1}{\Delta_i (\rho,v)} \;\;\;,\;\;\;
\psi_{iii} (\rho, v) = \psi_i (\rho,v) \;,
\nonumber\\
 K_{iii} (\rho,v) &=& 48 m^2\left(\frac{2}{2m - \sqrt{\left(v-m\right)^2+\rho^2} -\sqrt{\left(v+m\right)^2+\rho^2} }\right)^6 \;.
\end{eqnarray}
Using the bijection from the domain $ r > 0, \, 0 < \theta < \pi$ in spherical coordinates onto the Weyl upper half-plane, given by
\bea
\rho =
\sqrt{r^2 + 2m r }  \, \sin \theta \;\;\;,\;\;\; v = (r +m) \cos \theta \;,
\eea
results in the line element
\begin{equation}
ds_4^2 = - \left(1 + \frac{2m}{r} \right) dt^2 + \left(1 + \frac{2m}{r} \right)^{-1}  dr^2 + r^2 \left( d \theta^2 +  \sin^2 \theta \, d \phi^2 \right) \:.
\label{schwarzmetneg}
\end{equation}
This is the `negative mass' Schwarzschild solution \cite{Griffiths:2009dfa}, which has a naked curvature singularity at $r=0$.
This solution is an example of an $AI$-metric, c.f. Appendix \ref{sec:classA}.
\\

\noindent
{\bf Case $(iv)$}: we note that the transformation $\tau_1 \tau_2 \mapsto 1/(\tau_1 \tau_2)$ is effected by mapping $v \mapsto -v$.
Hence, for $\Gamma$ belonging to class $(iv)$, we obtain
\begin{eqnarray}
\Delta_{iv} (\rho, v) &=& \frac{1}{\Delta_{ii} (\rho, v)} \;\;\;,\;\;\;
\psi_{iv} (\rho, v) = \psi_{ii} (\rho, v) \;, 
\nonumber\\
 K_{iv}(\rho,v) &=& 48 m^2\left(\frac{2}{2m + \sqrt{\left(v-m\right)^2+\rho^2} -\sqrt{\left(v+m\right)^2+\rho^2} }\right)^6 \;.
 \label{exp41}
\end{eqnarray}
Since \eqref{exp2} and \eqref{exp41} are related by the
coordinate transformation $v \mapsto -v$, the solution described by \eqref{exp41} is the same as \eqref{schwarzmet2}.

%%%%%%%%%%%%%%%%%%
%%%%%%%%%%%%%%%%%%
\subsection{$\sigma= -1$ \label{ref:subsec-1}}
When $\sigma= -1$ we have, from \eqref{tt12},
\be \label{eq_tau_1_tau_2_interior_region}
\tau_1(\rho,v) = \frac{m-v + \sqrt{\left(m-v\right)^2 - \rho^2}}{\rho} \,, \qquad \tau_2(\rho,v) = \frac{-\left(m+v\right) + \sqrt{\left(m+v\right)^2 - \rho^2}}{\rho}\,.
\ee
Note that, differently from the case $\sigma = 1$, 
$\tau_1$ and $\tau_2$ can now take complex values and may, moreover, coincide with the fixed points $\pm1 \in FP_{-1}$
through which the curve $\Gamma$ has to pass. This occurs for $(\rho,v)$ on four half-lines, namely
\bea
\rho = \pm (v-m) \;\;\;,\;\;\; {\rm where} \quad \tau_1 = \mp 1\;, \nonumber\\
\rho = \pm (v+m) \;\;\;,\;\;\; {\rm where} \quad \tau_2 = \mp 1\;.
\eea

Imposing that $\tau_1$ and $\tau_2$ are real, so as to ensure that the solutions obtained from the canonical 
factorization of $\mathcal{M}_{(\rho,v)} (\tau)$
are real, we find that we have to restrict $(\rho,v)$ to lie in the open set consisting of the three regions in the Weyl upper-half plane denoted by $I, A, B$
in Figure \ref{figure_rho_v_plane_interior}.
\begin{figure}[h]
	\centering
	\includegraphics[width=0.4\textwidth]{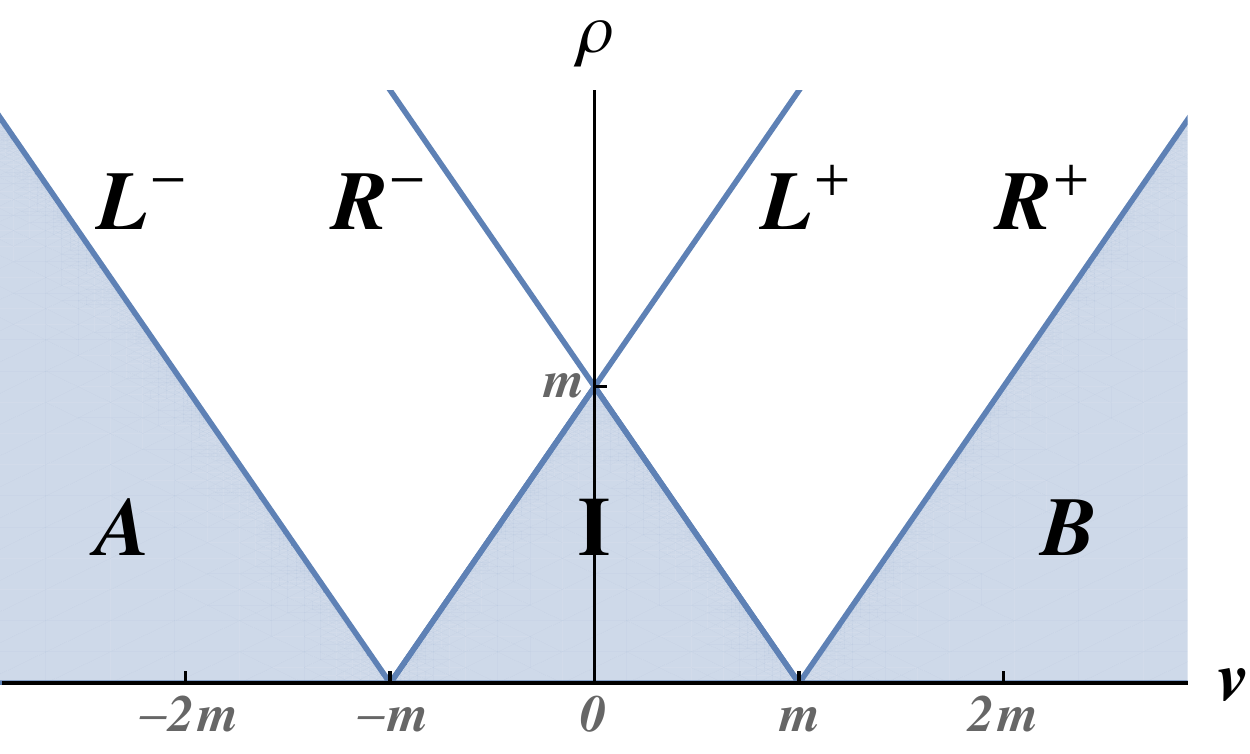}
	\caption[]{Regions in the Weyl upper-half plane for which both $\tau_1$ and $\tau_2$ are real.}
	\label{figure_rho_v_plane_interior}
\end{figure}
In Figure \ref{figure_rho_v_plane_interior}, we denote by $L_{\pm}$ the lines described by $\rho = \pm ( v + m)$ and by $R_{\pm}$ the lines
described by $\rho = \pm (v-m)$.

Considering now the cases $(i)-(iv)$ corresponding to the various possible choices of the contour $\Gamma$, as described
at the beginning of Section \ref{sec:schw}, we start 
by noting that in each of the cases  $(i)-(iv)$ we have $\Delta >0$ if $(\rho,v)$ is taken to lie in region $I$, while $\Delta < 0$ for region $A$ and 
region $B$.

We will first consider the case when $(\rho,v)$ takes values in region $I$, in which case $\tau_1 > 1$ and $-1 < \tau_2 < 0$.
Possible choices of contours for the four classes of contours $(i)-(iv)$ are depicted in Figure \ref{contours}.

\begin{figure}[hbt!]
	\centering
	\includegraphics[scale=0.4]{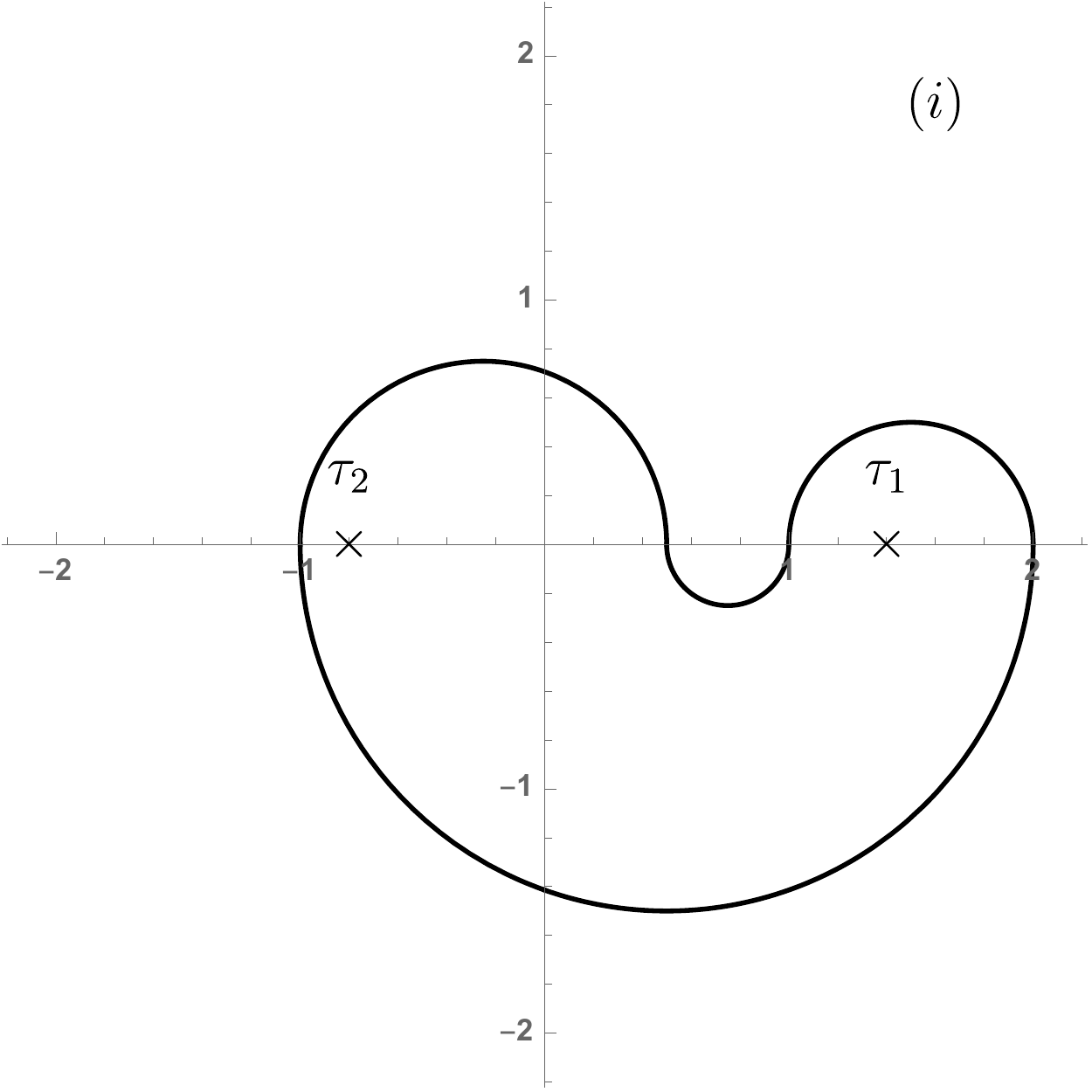}
	\includegraphics[scale=0.4]{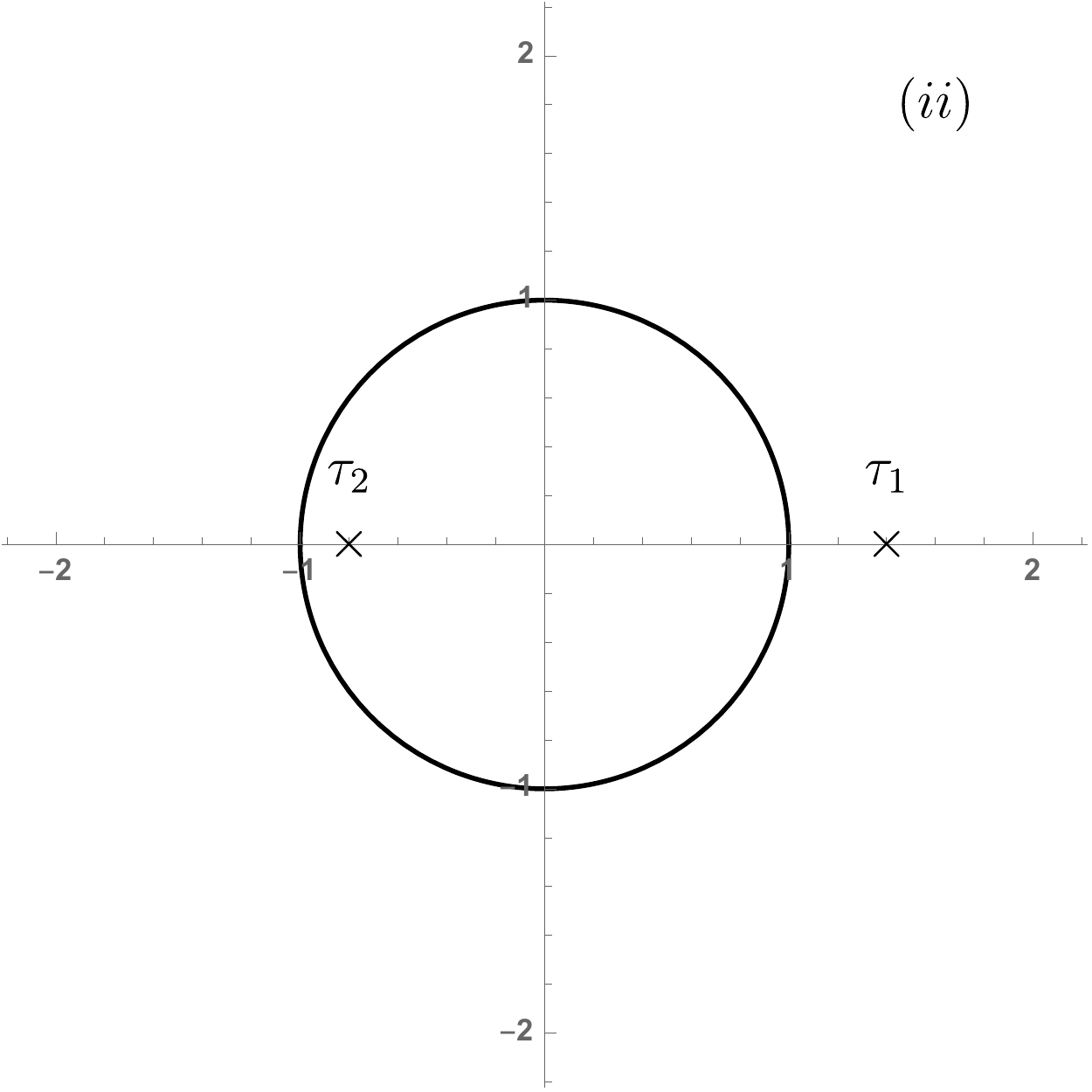}
	\includegraphics[scale=0.4]{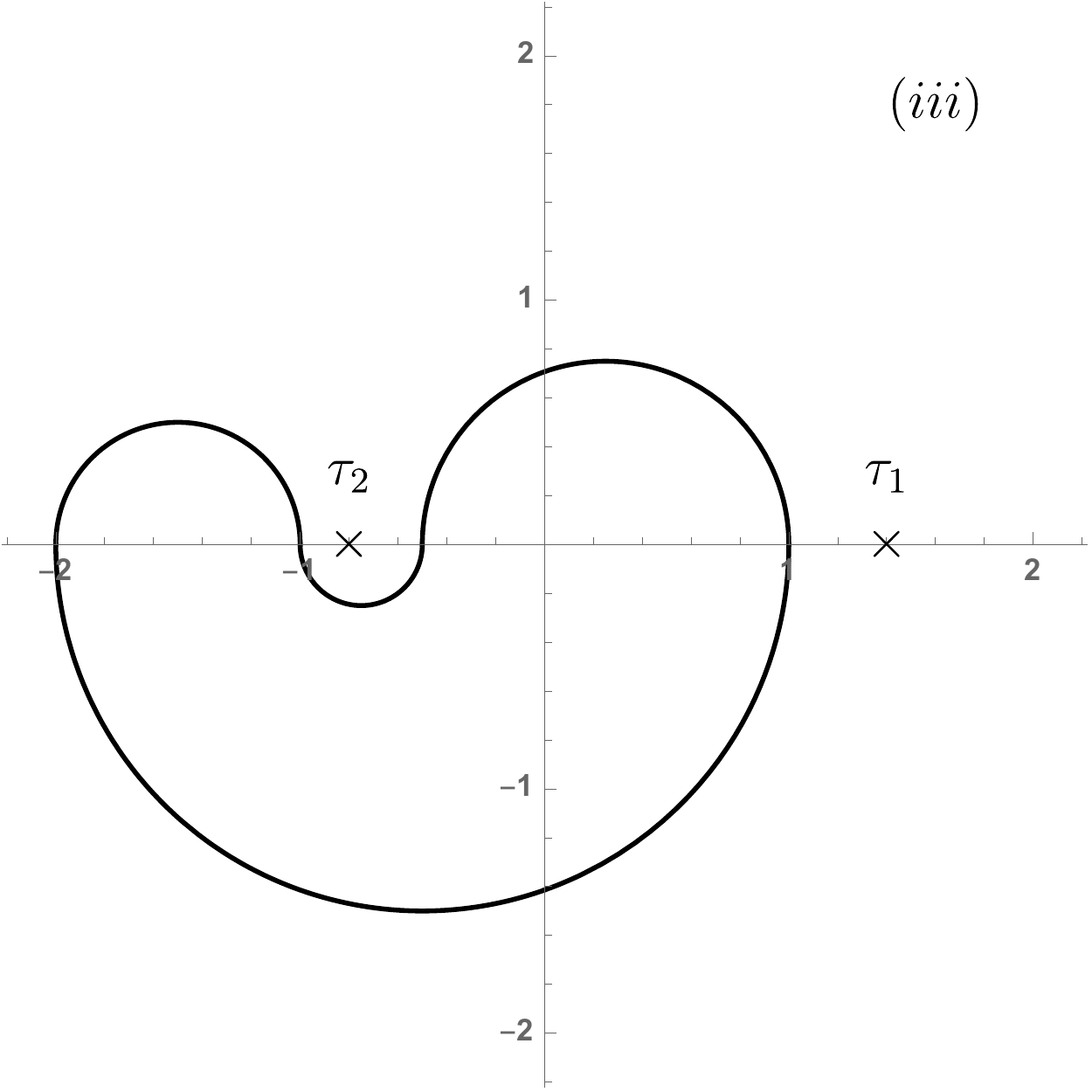}
	\includegraphics[scale=0.4]{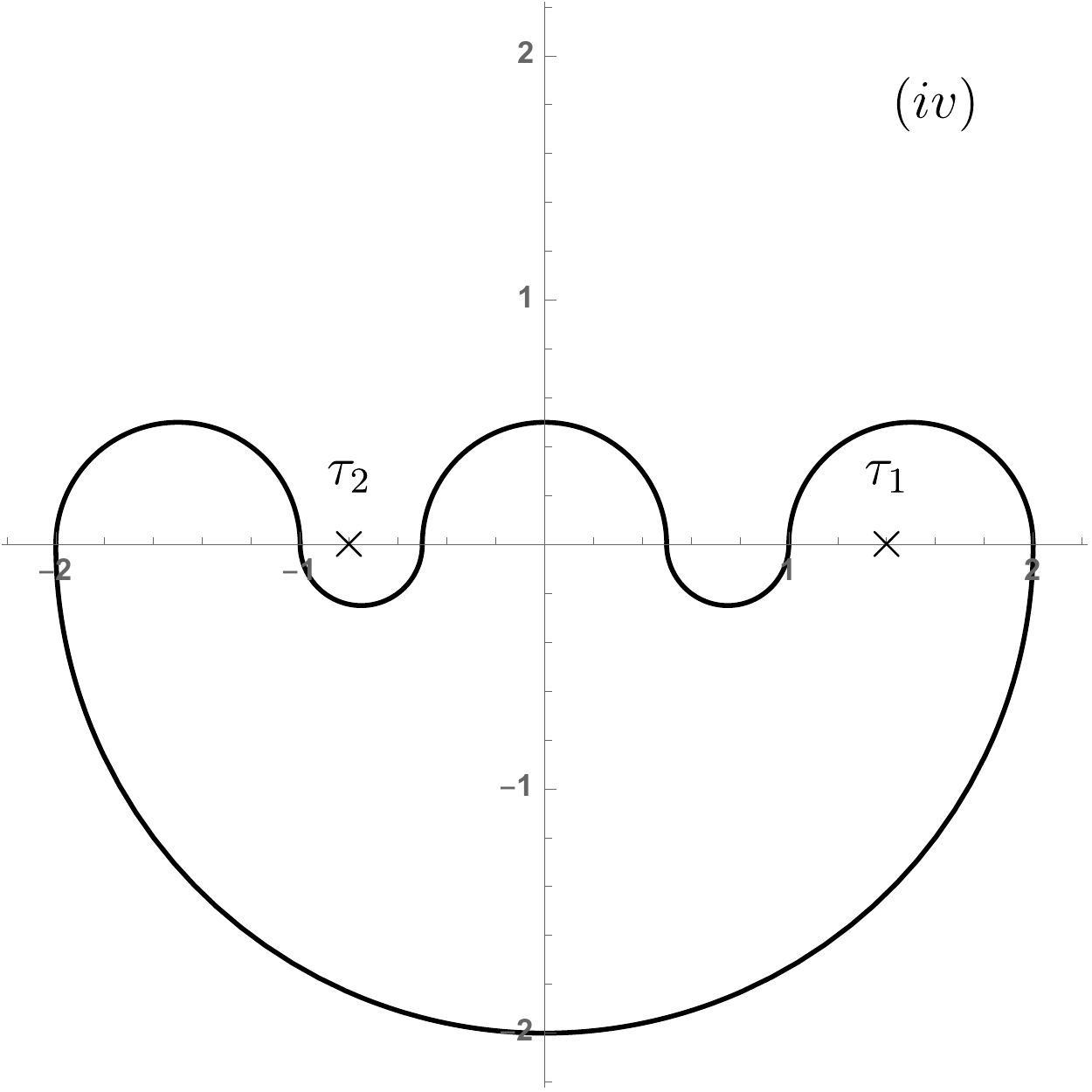}
		\caption[]{Four distinct choices of contours.}
		\label{contours}
\end{figure}

%%%%%%%%%%%
\subsubsection{Region $I$ \label{resultregI}}
%%%%%%%%%%%%%%
Here we take $\rho$ to be a time-like coordinate. Accordingly, we take
the four-dimensional space-time metric to be given by \eqref{solschw}
and 
\eqref{weylframe} (with $\sigma = -1$), i.e.
\bea
ds_4^2 =  \Delta \, dt^2 + \Delta^{-1} 
\left(  e^\psi \, \left(  - \, d \rho^2 + dv^2 \right)
 + \rho^2 d\phi^2 \right) \;.
 \label{metsig-1-1}
\eea
We obtain the following expressions for $\Delta$ and $\psi$ 
and for the four-dimensional Kretschmann scalar $K$:

\noindent
{\bf Case $(i)$}: for $\Gamma$ belonging to class $(i)$, we obtain
\begin{eqnarray}
\Delta_i (\rho, v) &=&  - \frac{\tau_2 (\rho,v)}{\tau_1 (\rho,v)} = 
 \frac{ m +v -\sqrt{\left( m + v \right)^2 -  \rho^2} }
{ m -v +  \sqrt{\left(m - v\right)^2 -  \rho^2} } = 
 \frac{ m -v -\sqrt{\left( m - v \right)^2 -  \rho^2} }
{ m +v +  \sqrt{\left(m + v\right)^2 -  \rho^2} } 
\;, 
\nonumber\\
\psi_i (\rho, v) &=& \log \left[ \tfrac12  \Big \vert \frac{m^2 - v^2 + \rho^2  }{ \sqrt{\left(m-v\right)^2 - \rho^2} \,  \sqrt{\left(m+v\right)^2 - \rho^2}  } - 1
\Big \vert
  \right] \;, 
\nonumber\\
 K_i(\rho,v) &=& 48 m^2\left(\frac{2}{2m + \sqrt{\left(m+v\right)^2-\rho^2} +\sqrt{\left(m-v\right)^2-\rho^2} }\right)^6 \;,
 \label{expr11regionI}
\end{eqnarray}
where $\psi = \psi_i$ is obtained by solving \eqref{eq_diff_eq_psi}, up to a constant which we set to zero.
Note that the expressions \eqref{expr11regionI} are invariant under
$v \mapsto -v$.

For $(\rho,v)$ taking values in region $I$, we obtain the following boundary values: on $R_-$,
\begin{eqnarray}
\Delta_i &=& \frac{2m - \rho - \sqrt{(2m - \rho)^2 - \rho^2}}{\rho} \;\;\;,\;\; (\rm bounded,  non-zero)
\nonumber\\
\psi_i &\rightarrow & \infty \;,
\nonumber\\
K_i &=& 48 m^2 \, \left(\frac{1}{m + \sqrt{(2m - \rho)^2 - \rho^2}} \right)^6 \;,
\label{ir-}
\end{eqnarray}
while on $L_+$
\begin{eqnarray}
\Delta_i &=& \frac{\rho }{2m - \rho + \sqrt{(2m - \rho)^2 - \rho^2}} \;\;\;,\;\; (\rm bounded,  non-zero)
\nonumber\\
\psi_i &\rightarrow & \infty \;,
\nonumber\\
K_i &=& 48 m^2 \, \left(\frac{1}{m + \sqrt{(2m - \rho)^2 - \rho^2}} \right)^6 \;,
\end{eqnarray}
and for $\rho = 0$ and $-m < v < m$,
\begin{eqnarray}
\Delta_i &=& 0\;,
\nonumber\\
\psi_i & \rightarrow & -  \infty \;,
\nonumber\\
K_i &=& \frac{3}{4 m^4} \;.
\end{eqnarray}
This solution has thus a Killing horizon on the segment $\rho = 0$, $-m < v < m$.\\

\noindent
{\bf Case $(ii)$}: for $\Gamma$ belonging to class $(ii)$, we obtain
\begin{eqnarray}
\Delta_{ii} (\rho, v) &=&  - \tau_1 (\rho,v) \, \tau_2 (\rho,v) = 
\frac{
\left( m -v +  \sqrt{\left(m - v\right)^2 -  \rho^2} \right) \left(
 m +v -\sqrt{\left( m + v \right)^2 -  \rho^2} \right)  }
{ \rho^2}
\;, 
\nonumber\\
\psi_{ii} (\rho, v) &=& \log \left[ \tfrac12  \Big \vert \frac{m^2 - v^2 + \rho^2  }{ \sqrt{\left(m-v\right)^2 - \rho^2} \,  \sqrt{\left(m+v\right)^2 - \rho^2}  } + 1
\Big \vert
  \right] \;, 
\nonumber\\
 K_{ii}(\rho,v) &=& 48 m^2\left(\frac{2}{2m + \sqrt{\left(m+v\right)^2-\rho^2} -\sqrt{\left(m-v\right)^2-\rho^2} }\right)^6 \;,
 \label{expr11reg}
\end{eqnarray}
where $\psi = \psi_{ii}$ is obtained by solving \eqref{eq_diff_eq_psi}, up to a constant which we set to zero.

For $(\rho,v)$ taking values in region $I$, we obtain the following boundary values: on $R_-$, 
they coincide with those for the case $(i)$ given in \eqref{ir-},
while on $L_+$
\begin{eqnarray}
\Delta_{ii} &=& \frac{2m - \rho + \sqrt{(2m - \rho)^2 - \rho^2}}{\rho} \;\;\;,\;\; (\rm bounded,  non-zero)
\nonumber\\
\psi_{ii} &\rightarrow & \infty \;,
\nonumber\\
K_{ii}&=& 48 m^2 \, \left(\frac{1}{m - \sqrt{- m v}} \right)^6 \;,
\end{eqnarray}
and for $\rho = 0$ and $-m < v < m$,
\begin{eqnarray}
\Delta_{ii} &=& \frac{m-v}{m+v}\;,
\nonumber\\
\psi_{ii} &=& 0 \;,
\nonumber\\
K_{ii}&=& \frac{48 m^2 }{(m + v)^6} \;.
\end{eqnarray}
This solution has no Killing horizon and no curvature singularity.\\

\noindent
{\bf Case $(iii)$}: for $\Gamma$ belonging to class $(iii)$, we obtain
\begin{eqnarray}
\Delta_{iii} (\rho, v) &=& \frac{1}{\Delta_i (\rho,v)} \;\;\;,\;\;\;
\psi_{iii} (\rho, v) = \psi_i (\rho,v) \;,
\nonumber\\
 K_{iii} (\rho,v) &=& 48 m^2\left(\frac{2}{2m - \sqrt{\left(m-v\right)^2-\rho^2} -\sqrt{\left(m+v\right)^2-\rho^2} }\right)^6 \;.
\end{eqnarray}
For $(\rho,v)$ taking values in region $I$, we obtain the following boundary values: on $R_-$, 
\begin{eqnarray}
\Delta_{iii} &=& \frac{1}{\Delta_{i}} \;,
\nonumber\\
\psi_{iii} &\rightarrow & \infty \;,
\nonumber\\
K_{iii} &=& 48 m^2 \, \left(\frac{2}{2 m - \sqrt{(2m - \rho)^2 - \rho^2}} \right)^6 \;,
\end{eqnarray}
while on $L_+$
\begin{eqnarray}
\Delta_{iii} &=&  \frac{1}{\Delta_{i}} \;,
\nonumber\\
\psi_{iii} &\rightarrow & \infty \;,
\nonumber\\
K_{iii}&=& 48 m^2 \, \left(\frac{2}{2m - \sqrt{(2m - \rho)^2 - \rho^2}} \right)^6 \;,
\end{eqnarray}
and for $\rho = 0$ and $-m < v < m$,
\begin{eqnarray}
\Delta_{iii} &\rightarrow& \infty \;,
\nonumber\\
\psi_{iii} &\rightarrow & - \infty \;,
\nonumber\\
K_{iii}&\rightarrow & \infty \;.
\end{eqnarray}
Thus, this solution has a curvature singularity on the segment $\rho = 0$, $-m < v < m$.\\

\noindent
{\bf Case $(iv)$}: we note that the transformation $\tau_1 \tau_2 \mapsto 1/(\tau_1 \tau_2)$ is effected by mapping $v \mapsto -v$.
Hence, for $\Gamma$ belonging to class $(iv)$, we obtain
\begin{eqnarray}
\Delta_{iv} (\rho, v) &=& \frac{1}{\Delta_{ii} (\rho, v)} \;\;\;,\;\;\;
\psi_{iv} (\rho, v) = \psi_{ii} (\rho, v) \;, 
\nonumber\\
 K_{iv}(\rho,v) &=& 48 m^2\left(\frac{2}{2m + \sqrt{\left(m-v\right)^2-\rho^2} -\sqrt{\left(m+v\right)^2-\rho^2} }\right)^6 \;.
% \label{exp4}
\end{eqnarray}
For $(\rho,v)$ taking values in region $I$, we obtain the following boundary values: on $R_-$, 
\begin{eqnarray}
\Delta_{iv} &=& \frac{1}{\Delta_{ii}} \;,
\nonumber\\
\psi_{iv} &\rightarrow & \infty \;,
\nonumber\\
K_{iv} &=& 48 m^2 \, \left(\frac{2}{2 m - \sqrt{(2m - \rho)^2 - \rho^2}} \right)^6 \;,
\end{eqnarray}
while on $L_+$
\begin{eqnarray}
\Delta_{iv} &=& \frac{1}{\Delta_{ii}} \;,
\nonumber\\
\psi_{iv} &\rightarrow & \infty \;,
\nonumber\\
K_{iv} &=& 48 m^2 \, \left(\frac{2}{2 m + \sqrt{(2m - \rho)^2 - \rho^2}} \right)^6 \;,
\end{eqnarray}
and for $\rho = 0$ and $-m < v < m$,
\begin{eqnarray}
\Delta_{iv} &=& \frac{1}{\Delta_{ii}}\;,
\nonumber\\
\psi_{iv} &=& 0 \;,
\nonumber\\
K_{iv} &=& \frac{48 m^2 }{(m - v)^6} \;.
\end{eqnarray}
This solution has no Killing horizon and no curvature singularity.

%%%%%%%%%%
\subsubsection{Extending solutions: the interior region of the Schwarzschild solution}
%%%%%%%%%%%%%%

In Subsection \ref{sols1} we used the change to spherical coordinates \eqref{sphext} to show that one of the solutions
to the field equations that we obtained in Weyl coordinates describes the exterior region of the Schwarzschild black hole.
If we write  \eqref{sphext} as 
\bea
\rho =
\sqrt{|r^2 - 2m r| }  \, \sin \theta \;\;\;,\;\;\; v = (r -m) \cos \theta \;,
\label{sphint}
\eea
it can be verified that \eqref{sphint}, in the form
\bea
\rho =
\sqrt{2m r - r^2 }  \, \sin \theta \;\;\;,\;\;\; v = (r -m) \cos \theta \;,
\label{sphint2}
\eea
maps bijectively the region $I_S$, contained in the $(r,\theta)$-rectangle with $0 < r < 2m, 0 < \theta < \pi$, defined by
\bea
m-r < m \, \cos \theta < r-m \;,
\eea
onto the $(\rho,v)$-triangle $I$, see Figure \ref{figure_spherical_coords_interior}.
\begin{figure}[h!]
	\centering
	\includegraphics[width=0.4\textwidth]{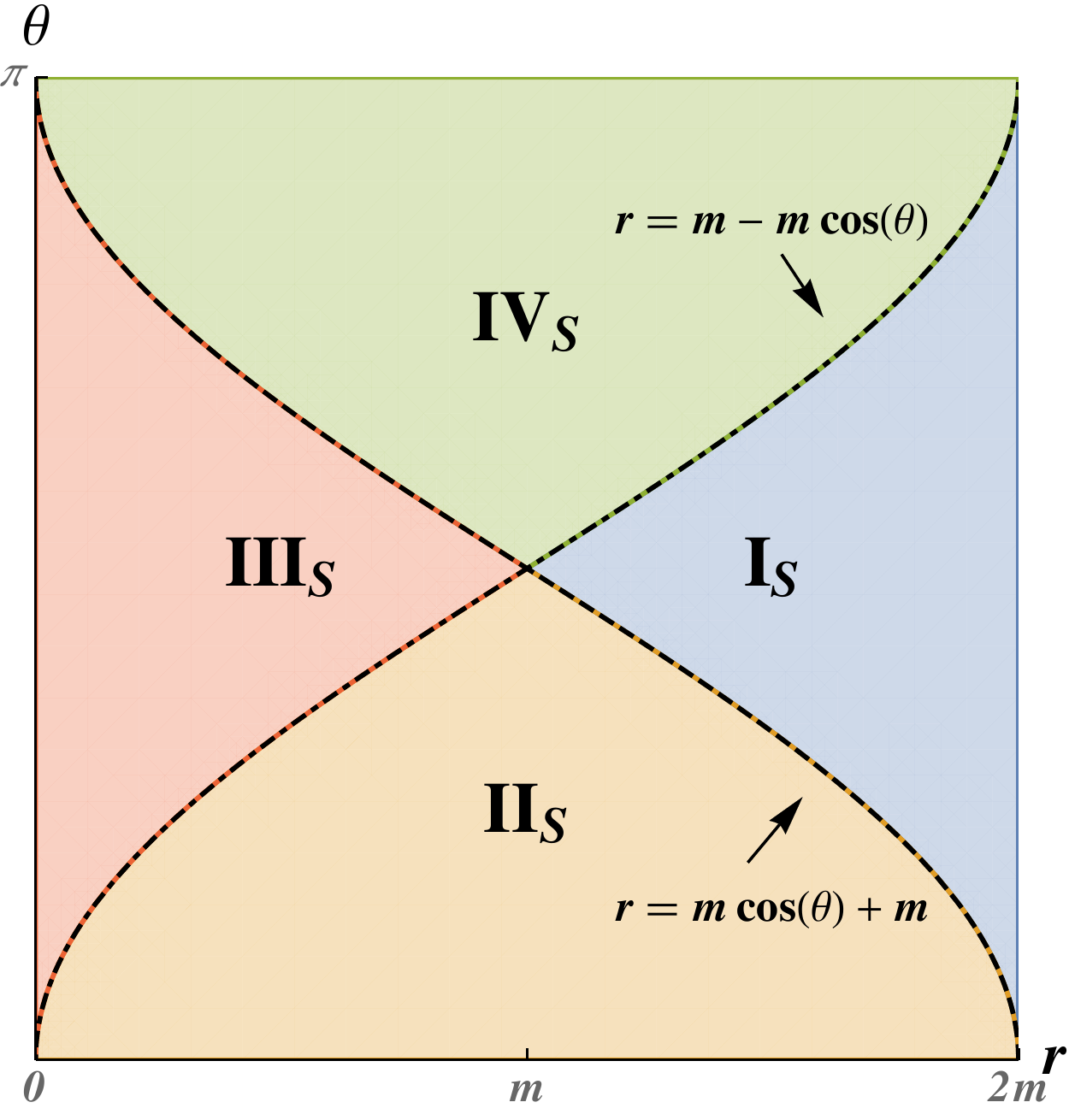}
	\caption[]{Bijection to quadrangle in $(r, \cos \theta)$-plane.}
		\label{figure_spherical_coords_interior}
\end{figure}
With the change of coordinates \eqref{sphint2}, the metric component $\Delta_i$ for case $(i)$ with $\sigma = -1$ becomes
$\Delta_{i, S} = 2m/r -1$, and the associated four-dimensional metric \eqref{metsig-1-1} takes the form
\begin{equation}
ds_4^2 =  \left(\frac{2m}{r} -1 \right) dt^2-  \left(\frac{2m}{r} -1 \right)^{-1}  dr^2 + r^2 \left( d \theta^2 + \sin^2 \theta \, d \phi^2 \right) \:,
\label{intschw}
\end{equation}
i.e. it coincides on $I_S$ with the metric describing the interior region of the Schwarzschild solution.

It is thus natural to ask whether it is possible to extend $\Delta_i$ in region $I$  into neighbouring regions of the Weyl upper-half plane in a continuous manner,
in such a way as to recover the whole interior through appropriate changes of coordinates. This extension can indeed be implemented
by performing affine coordinate changes \eqref{affinetransf} of Weyl coordinates.
The whole interior region of the Schwarzschild solution is obtained by extending the solution in region $I$ to the quadrangle depicted
in Figure \ref{figure_rho_v_generalized_plane_interior}, keeping $\rho$ as the time-like coordinate.  We now proced to explain this extension.

We can define bijections from the triangle $I$ onto each of the four triangles $I, II, III, IV$ represented in Figure \ref{figure_rho_v_generalized_plane_interior}.
\begin{figure}[h!]
	\centering
	\includegraphics[width=0.6\textwidth]{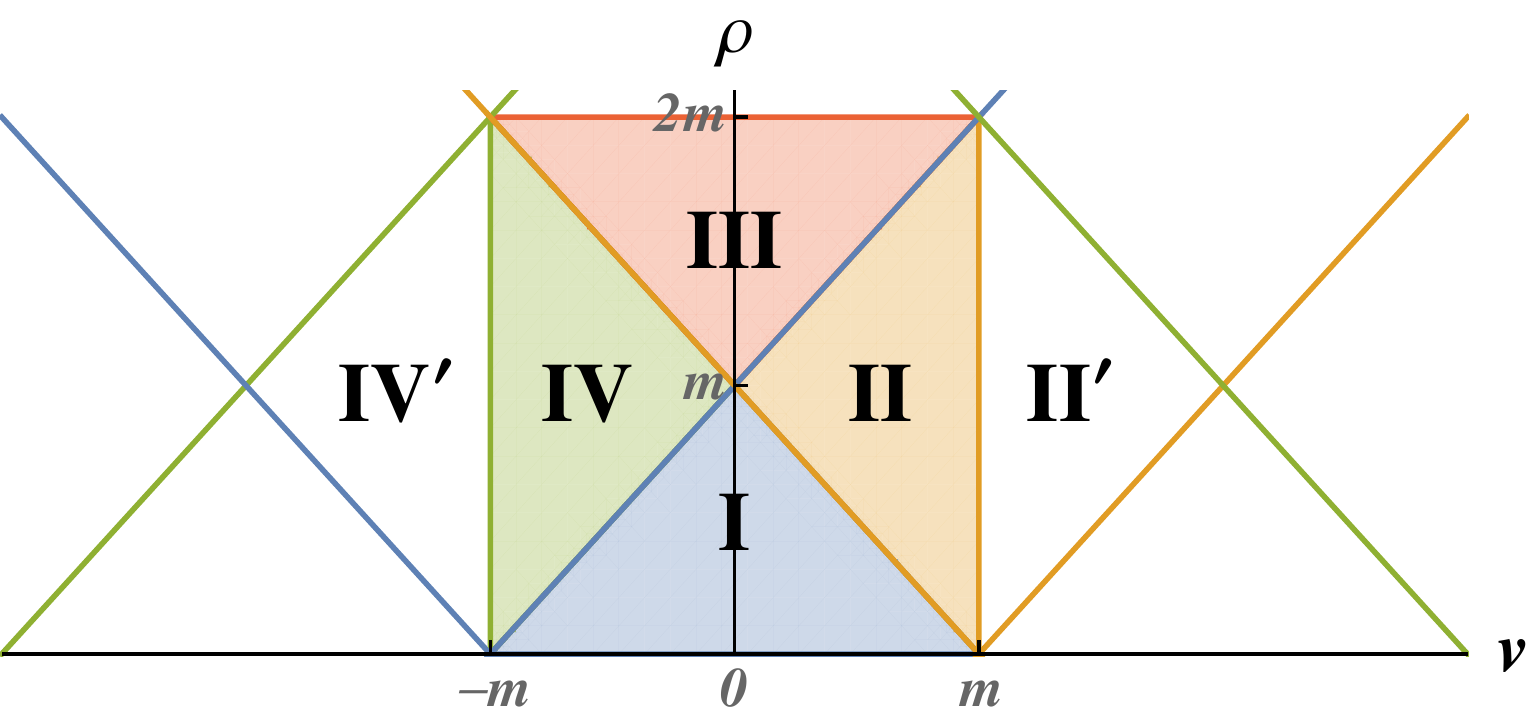}
	\caption[]{Quadrangle $0 < \rho < 2 m, \; -m < v < m$.}
		\label{figure_rho_v_generalized_plane_interior}
\end{figure}
This is implemented using the affine coordinate transformations given in Table \ref{tab:bij}.
Note that the affine transformations are such that they leave the lines $\rho = m + v$ and $\rho = m-v$ invariant.\\

\begin{table}[h!]
  \begin{center}
    \begin{tabular}{|l | l | l | l |} 
          \hline
     $ A$ & $I \rightarrow I$ & $ (\rho, v) \mapsto (\rho,v)$ & $\rho = 0 \mapsto \rho = 0$ \\
      \hline
      $B $& $I \rightarrow II$ & $ (\rho, v) \mapsto (m -v, m-\rho)$ & $\rho = 0 \mapsto v = m$  \\
      \hline
      $C $& $I \rightarrow III$ & $ (\rho, v) \mapsto (2m - \rho,-v)$  & $\rho = 0 \mapsto \rho = 2m$ \\
      \hline
     $  D $& $I \rightarrow IV$ & $ (\rho, v) \mapsto (m+v, -m + \rho)$ & $\rho = 0 \mapsto v = -m$  \\
     \hline
    \end{tabular}
       \caption{Bijections between triangles.}
         \label{tab:bij}
  \end{center}
\end{table}

Taking into account the results of Subsection \ref{resultregI}, we may therefore extend $\Delta_i (\rho,v)$ continuously to the rectangle $-m<v<m, 0< \rho<2m$ by
\bea
\Delta (\rho,v) = 
    \begin{cases}
      \Delta_i (\rho,v) \; \text{ on}  \; I \\
        \Delta_{ii} (B^{-1}(\rho,v)) \; \text{ on}  \; II \\
          \Delta_{iii} (C^{-1}(\rho,v)) \; \text{ on}  \; III \\
            \Delta_{iv} (D^{-1}(\rho,v)) \; \text{ on}  \; IV \\
          \end{cases}
          \label{extintschw}
\eea
where $\Delta_{ii} (B^{-1}(\rho,v))$ and $\Delta_{iv} (D^{-1}(\rho,v))$ can in turn be further extended naturally to the triangles $II'$ and $IV'$, respectively, in Figure
\ref{figure_rho_v_generalized_plane_interior}, since 
\bea
 \Delta_{ii} (B^{-1}(\rho,v)) &=& - \frac{ \left( \rho + \sqrt{\rho^2 - (m-v)^2} \right) \left( \rho - 2m + \sqrt{(\rho - 2m)^2 - (m -v)^2} \right)}{(m -v)^2 } \;,
 \nonumber\\
 \Delta_{iv} (D^{-1}(\rho,v)) &=& - \frac{ \left( \rho + \sqrt{\rho^2 - (m+v)^2} \right) \left( \rho - 2m + \sqrt{(\rho - 2m)^2 - (m +v)^2} \right)}{(m +v)^2 } \;,
 \nonumber\\
 \eea
are well-behaved and positive also in these regions.

In order to ensure that $\rho$ is a time-like coordinate in the quadrangle depicted in Figure \ref{figure_rho_v_generalized_plane_interior},
the affine transformations of Table \ref{tab:bij} from region $I$ into regions $II$ and $IV$ need to be accompanied by 
the transformation $ds_2^2 \rightarrow - ds_2^2$, which is consistent
with the fact that $ds_2^2$ can take the form 
\eqref{weylframe} or \eqref{weylframe2}.

Then, expressing \eqref{extintschw} in spherical coordinates $(r, \theta)$, with $0 < \theta < \pi, 0<r<2m$, according to Table \ref{tab:bijsph}, we obtain 
\eqref{intschw}, which describes the whole interior region of the Schwarzschild black hole. 
This solution is an example of an $AI$-metric, c.f. Appendix \ref{sec:classA}.
\\

\begin{table}[h!]
  \begin{center}
    \begin{tabular}{| l | l  |} 
          \hline
     $I \rightarrow I_S$ & $ \rho = \sqrt{2 m r - r^2} \, \sin \theta \;,\; v = (r-m) \cos \theta $ \\
      \hline
      $II \rightarrow II_S$ & $ \rho = m - (r-m) \cos \theta \;,\; v = m -  \sqrt{2 m r - r^2} \, \sin \theta $   \\
      \hline
       $III \rightarrow III_S$ & $ \rho = 2m -  \sqrt{2 m r - r^2} \, \sin \theta \;,\; v = (m-r) \cos \theta $  \\
      \hline
      $IV \rightarrow IV_S$ & $ \rho = m +  (r-m) \cos \theta  \;,\; v = - m + \sqrt{2 m r - r^2} \, \sin \theta $    \\
      \hline
    \end{tabular}
       \caption{Bijective transformations.}
         \label{tab:bijsph}
  \end{center}
\end{table}

If we now consider $\Delta_{ii}$ or $\Delta_{iii}$ instead of $\Delta_i$ in triangle $I$ of the Weyl upper-half plane, we obtain, by
an appropriate change of coordinates, the Schwarzschild metric \eqref{intschw} in one of the triangles of Figure \ref{figure_spherical_coords_interior},
which we can then extend in an analogous manner. As an example, consider $\Delta_{ii}$ in region $I$ and subject it to the following
sequence of coordinate transformations,
\bea
(\rho,v) \mapsto (m-v, m-\rho) \mapsto (\sqrt{2mr - r^2} \sin \theta, (r-m) \cos \theta) \;. \label{transfII1}
\eea
This yields the Schwarzschild solution \eqref{intschw} in region $II_S$, see  Figure \ref{figure_spherical_coords_interior}, provided
the transformation \eqref{transfII1} is accompanied by the transformation $ds_2^2 \rightarrow - ds_2^2$, to ensure that $\rho$ is a time-like coordinate.

Clearly, we could have defined other continuous extensions of $\Delta_i$ in region $I$,
keeping $\rho$ as the time-like coordinate, by using the same affine transformations.
There are three such possible extensions when taking into account the behaviour of $\Delta_i$ on the boundary lines $\rho = \pm v + m$.
They are given by 
\bea
&&\underbrace{\Delta_i (\rho,v)}_{\text{in region} \, I}
\rightarrow \underbrace{\Delta_{i} (B^{-1}(\rho,v))}_{\text{in region} \, II}
\rightarrow \underbrace{\Delta_{i} (C^{-1}(\rho,v))}_{\text{in region} \, III}
\rightarrow \underbrace{\Delta_{i} (D^{-1}(\rho,v))}_{\text{in region} \, IV}
\nonumber\\
&& \underbrace{\Delta_i (\rho,v)}_{\text{in region} \, I}
\rightarrow \underbrace{\Delta_{i} (B^{-1}(\rho,v))}_{\text{in region} \, II}
\rightarrow \underbrace{\Delta_{iv} (C^{-1}(\rho,v))}_{\text{in region} \, III}
\rightarrow \underbrace{\Delta_{iv} (D^{-1}(\rho,v))}_{\text{in region} \, IV}
\nonumber\\
&& \underbrace{\Delta_i (\rho,v)}_{\text{in region} \, I}
\rightarrow \underbrace{\Delta_{ii} (B^{-1}(\rho,v))}_{\text{in region} \, II}
\rightarrow \underbrace{\Delta_{ii} (C^{-1}(\rho,v))}_{\text{in region} \, III}
\rightarrow \underbrace{\Delta_{i} (D^{-1}(\rho,v))}_{\text{in region} \, IV}
\label{ext_jump}
\eea
However, unlike the extension \eqref{extintschw}, these three extensions
have 
a jump in the transverse extrinsic curvature \cite{Barrabes:1991ng,Poisson:2009pwt}
across the lines $\rho = \pm v + m$. We discuss this in Appendix \ref{sec:glue}.

%%%%%%%%%%%%%%%%
\subsubsection{Regions $A$ and $B$}
%%%%%%%%%%%%%%%%%%
%

Now we take $(\rho,v)$ to lie in regions $A$ and $B$ in Figure \ref{figure_rho_v_plane_interior}. Then,
as mentioned at the beginning of Subsection \ref{ref:subsec-1},
from the factorization of ${\cal M}_{(\rho,v)} (\tau)$ with $\sigma =-1$ we obtain four different solutions $\Delta (\rho,v)$, corresponding
to the cases $(i)-(iv)$, with $\Delta < 0$. Although this violates our initial assumption that $\Delta >0$ (see \eqref{4dWLP}), this is easily
overcome by changing the overall sign of the monodromy matrix \eqref{monschwarz}. This implies that we can take over the expressions
for $\Delta, \psi$ and $K$ for the cases  $(i)-(iv)$ of Subsection \ref{ref:subsec-1}, merely changing the sign of the expression for $\Delta$.

We begin by considering solutions in region $A$.\\

\noindent
{\bf First case}: 
let us first consider the solution $\Delta_{iv}^A = \frac{1}{\tau_1 \tau_2}$, where $\tau_1$ and $\tau_2$ are given by 
\eqref{eq_tau_1_tau_2_interior_region},
\bea
\Delta_{iv}^A (\rho,v)= \frac{\rho^2}{ \left( m - v + \sqrt{(m-v)^2 - \rho^2} \right) \left( -(m+v) +  \sqrt{(m+v)^2 - \rho^2} \right) } \;.
\eea
This solution possesses a Killing horizon for $\rho = 0, v< -m$ (where $\Delta_{iv}^A (\rho,v) =0$). By applying the affine transformation
$(\rho,v) \mapsto (-m-v, -m-\rho)$ to the solution $\Delta_{i}^A = \tau_2/\tau_1$,
\bea
\Delta_{i}^A (\rho,v)= \frac{-(m + v) + \sqrt{(m+v)^2 - \rho^2} }{  m-v +  \sqrt{(m-v)^2 - \rho^2} } \;,
\eea
we obtain a solution in region $A'$ (represented in Figure \ref{figure_rho_v_generalized_A_B_regions}) which continuously extends $\Delta_{iv}^A$
to $A \cup A'$, and is actually also defined and valid in the region $A''$ which comprises the points between the lines $v = -m$ and $\rho = v+m$.

\begin{figure}[h]
	\centering
	\includegraphics[width=0.5\textwidth]{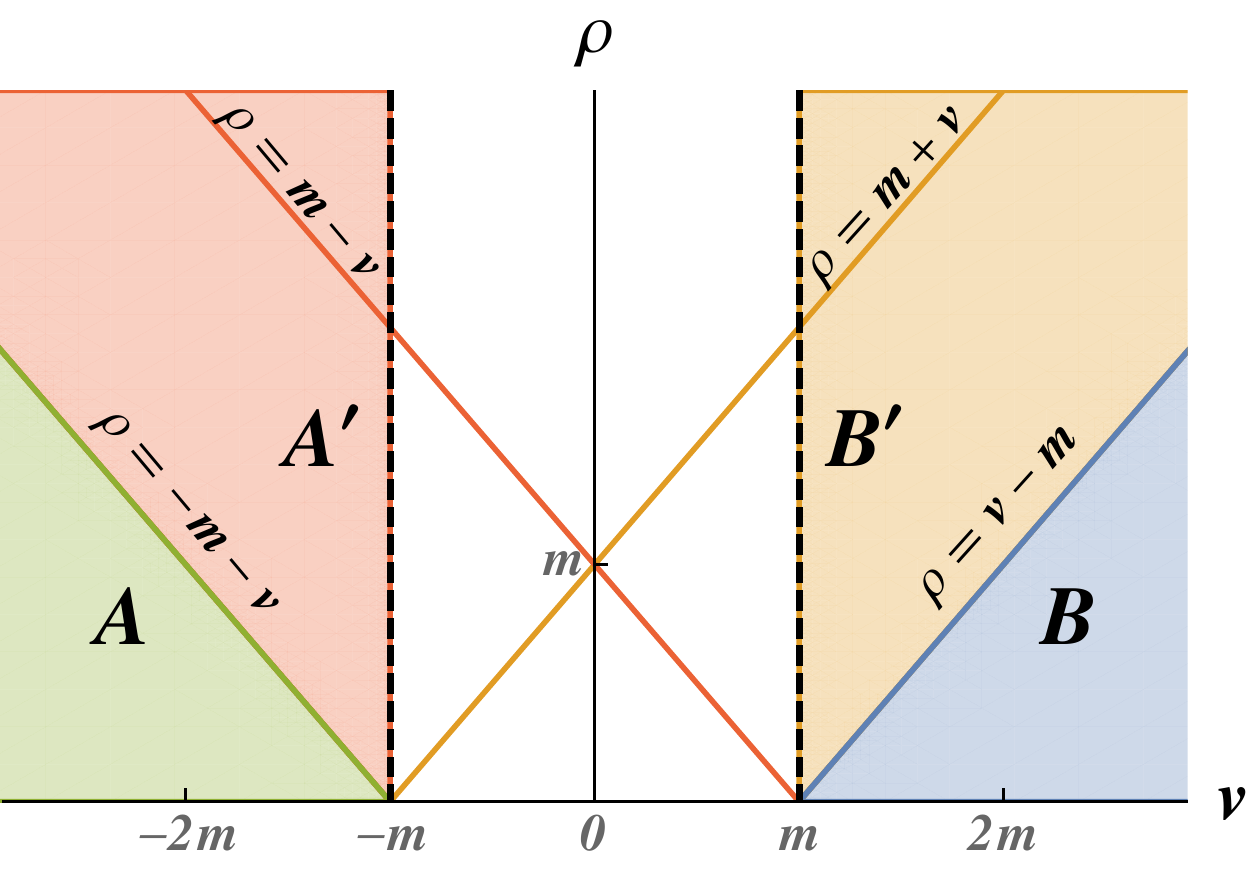}
	\caption[]{ Weyl upper-half plane with solutions extended past regions $A$ and $B$. }
	\label{figure_rho_v_generalized_A_B_regions}
\end{figure}

If we now consider coordinates $(\varrho, \vartheta) \in  \, ]2m, \infty[ \, \times \,  ]0, \infty[$ and divide this domain into regions
\bea
{\tilde A} = \{ (\varrho, \vartheta): \varrho > 2m, \, \vartheta > 0, \, \varrho < m + m \cosh \vartheta \} \;, \nonumber\\
{\tilde A}' = \{ (\varrho, \vartheta): \varrho > 2m, \, \vartheta > 0, \, \varrho > m + m \cosh \vartheta \} \;, 
\eea
we see that ${\tilde A}$ is bijectively mapped onto $A$ by
\bea
\rho = \sqrt{\varrho^2 - 2 m \varrho} \, \sinh \vartheta \;\;\;,\;\,\; v = (m - \varrho) \, \cosh \vartheta \;,
\eea
and 
${\tilde A}'$ is bijectively mapped onto $A'$ by
\bea
\rho = -m - (m - \varrho) \, \cosh \vartheta \;\;\;,\;\,\; v = - m - 
 \sqrt{\varrho^2 - 2 m \varrho} \, \sinh \vartheta \;.
\eea
With this change of coordinates, the previous solution given by
\begin{equation}
\Delta (\rho,v) = \left\{
\begin{aligned}
& \Delta_{iv}^A (\rho,v)  \qquad \text{on $A$,} \\
& \Delta_{i}^A (-m-v, -m - \rho)  \qquad \text{on $A'$,}
\end{aligned}
\right.
\label{solglue1}
\end{equation}
becomes
\bea
\Delta (\rho,v) = 1 - \frac{2m}{\varrho} \;,
\eea
As before, to ensure that $\rho$ is a time-like coordinate in region $A'$,  the affine transformation
$(\rho,v) \mapsto (-m-v, -m-\rho)$ needs to be accompanied by the transformation $ds_2^2 \rightarrow - ds_2^2$. 
The resulting 
four-dimensional metric reads
\bea
ds^2_4 = \left( 1 - \frac{2m}{\varrho} \right) dt^2 - \left( 1 - \frac{2m}{\varrho} \right)^{-1} d \varrho^2 + \varrho^2 \left( d \vartheta^2 + \sinh^2 \vartheta
\, d \phi^2 \right) \;,
\label{a2hyp}
\eea
which describes the `exterior' region of the $AII$-metric \eqref{hyp}.\\

\noindent
{\bf Second case:}
in an analogous manner, we can extend the solution $\Delta_{ii}^A = 1/( \Delta_{iv}^A)$ in region $A$ to region $A \cup A' \cup A''$ by
\begin{equation}
\Delta (\rho,v) = \left\{
\begin{aligned}
& \Delta_{ii}^A (\rho,v)  \qquad \text{on $A$,} \\
& \Delta_{iii}^A (-m-v, -m - \rho)  \qquad \text{on $A' \cup A''$.}
\end{aligned}
\right.
\label{solglue2}
\end{equation}
The associated space-time metric has a curvature singularity at $\rho = 0, v < -m$. By using the change of coordinates
\bea
\rho = \sqrt{\varrho^2 + 2 m \varrho} \, \sinh \vartheta \;\;\;,\;\,\; v = - (m + \varrho) \, \cosh \vartheta \;,
\eea
defined for $\varrho > 0, \, \vartheta > 0$, the regions 
$  \{ (\varrho, \vartheta): \varrho > 0, \, \vartheta > 0, \, \varrho < - m + m \cosh \vartheta \} $ and 
$\{ (\varrho, \vartheta): \varrho > 0, \, \vartheta > 0, \, \varrho > - m + m \cosh \vartheta \} $
are bijectively mapped into $A$ and $A'$, respectively, and in these new coordinates $\Delta (\rho,v)$ is given by
$\Delta (\rho,v) = 1 + 2m/ \varrho$,
which results in the four-dimensional metric
\bea
ds^2_4 = \left( 1 + \frac{2m}{\varrho} \right) dt^2 - \left( 1 + \frac{2m}{\varrho} \right)^{-1} d \varrho^2 + \varrho^2 \left( d \vartheta^2 + \sinh^2 \vartheta
\, d \phi^2 \right) \;.
\eea
This is the `negative mass' version of the $AII$-metric \eqref{a2hyp} \cite{Griffiths:2009dfa},  c.f. Appendix \ref{sec:classA}.

Note that the other two vacuum solutions on $A \cup A'$ given by
\begin{eqnarray}
\Delta (\rho,v) &=& \left\{
\begin{aligned}
& \Delta_{i}^A (\rho,v)  \qquad \text{on $A$,} \\
& \Delta_{iv}^A (-m-v, -m - \rho)  \qquad \text{on $A'$}
\end{aligned}
\right. 
 \nonumber\\[3mm]
\Delta (\rho,v) &=& \left\{
\begin{aligned}
& \Delta_{iii}^A (\rho,v)  \qquad \text{on $A$,} \\
& \Delta_{ii}^A (-m-v, -m - \rho)  \qquad \text{on $A'$}
\end{aligned}
\right.
\end{eqnarray}
are identical to \eqref{solglue1} and \eqref{solglue2}, respectively, via $(\rho,v) \mapsto (-m-v, -m-\rho)$.

Now let us consider solutions defined in region $B$. We note that solutions on $B \cup B'$ obtained in an analogous manner as above are mapped
into the ones on $A \cup A'$ discussed above, so that we do not have to discuss them separately.

Finally, we note that, as in \eqref{ext_jump}, we can define other continuous extensions of $\Delta_i^A, \dots, \Delta_{iv}^A$. These other extensions
will suffer from a jump in the transverse extrinsic curvature across the line $\rho = - m - v$.

%%%%%%%%%%%%%%%%%%
%%%%%%%%%%%%%%%%%%
\section{The monodromy matrix with $\epsilon =0$ \label{sec:schw0}}

%%%%%%%%%%%%%%%%%%
%%%%%%%%%%%%%%%%%%

In this section, we discuss 
the canonical factorization of the monodromy matrix 
\eqref{monschwarzeps} with $\epsilon = 0$, 
\be
\mathcal{M} (u) = \begin{bmatrix}
	\sigma \, \frac{u}{m} & \; 0 \\
	0 & \;  \sigma \, \frac{m}{u}
\end{bmatrix} \;\;\;,\;\;\; m \in \mathbb{R}^+ \;,
\label{monschwarz0}
\ee
for both cases $\sigma = \pm 1$. Note that this monodromy matrix is unbounded in the complex $u$-plane.
The involution $\natural$ acts again as transposition on matrices.

Using the spectral curve \eqref{utauu}, we obtain
\begin{equation}
\sigma \, \frac{u}{m}	=  -  \frac{\rho}{2m} \frac{(\tau-\tau_0)(\tau- \widetilde{\tau}_0)}{\tau} \;,
\label{utau0}
\end{equation}
where
 $\tau_0 = \varphi_{\alpha=0}(\rho,v)$, with  $\varphi_{\alpha}(\rho,v)$ given in \eqref{fial},
\begin{equation}
	\tau_0 = -\sigma \frac{-v +  \sqrt{v^2+\sigma\rho^2}}{\rho} \;,
\end{equation}
and where $\widetilde{\tau}_0 = - \sigma/\tau_0$.

Next, let us factorize \eqref{utau0} with respect to a contour $\Gamma$ that satisfies Assumption $1$, c.f. Section \ref{sec:mtbm}.
There are two possible classes of contours, from which $\Gamma$ can be chosen:
\begin{itemize}
\item[$(i)$] $\tau_0$ is inside the contour $\Gamma$ (and hence $\widetilde{\tau}_0$ is outside of $\Gamma$);
\item[$(ii)$] $\widetilde{\tau}_0$ is inside the contour $\Gamma$ (and hence $\tau_0$ is outside of $\Gamma$).
\end{itemize}
Factorizing with respect to $\Gamma$ we obtain, for each of these cases,
\be
\sigma \, \frac{u}{m} = \sigma \, m_-(\tau) \, m_+ (\tau)
\ee
and
\be
M (\rho,v)
= \begin{bmatrix}
	 \sigma \, m_- (\infty) \; & \;  0 \\
	0 \;  & \;
	\sigma \, m_-^{-1} (\infty) 
	 \end{bmatrix}  =  
	 \begin{bmatrix}
	\Delta \; & \;  0 \\
	0 \;  & \;
	\Delta^{-1}  
	 \end{bmatrix} \;,
	 \label{Meps00}
	\ee
where
\begin{itemize}
\item[$i)$] for a contour in class $(i)$
\be
m_+ (\tau) = \frac{\tau- \widetilde{\tau}_0}{-\widetilde{\tau}_0} \;\;\;,\;\;\; m_- (\tau) =
   \frac{\sigma \, \rho}{2m} \frac{(\tau-\tau_0)}{\tau} \, \widetilde{\tau}_0 \;\;\;,\;\;\;
   \Delta = \frac{\rho \, \widetilde{\tau}_0 }{2m} = \sigma \,  \frac{ v +  \sqrt{v^2+\sigma \rho^2} }{2m}\;,
\label{deltsigvr}
\ee

\item[$ii)$] for a contour in class $(ii)$
\be
m_+ (\tau) = \frac{\tau- {\tau}_0}{- \tau_0} \;\;\;,\;\;\; m_- (\tau) =
   \frac{\sigma \, \rho}{2m} \frac{(\tau-\widetilde{\tau}_0)}{\tau} \, {\tau}_0 \;\;\;,\;\;\;
   \Delta = \frac{\rho \, {\tau}_0 }{2m} = \sigma \,  \frac{ v -  \sqrt{v^2+\sigma \rho^2} }{2m}\;.
    \label{mpmt01}
\ee

\end{itemize}

We take the four-dimensional space-time metric to be given by 
\bea
ds_4^2 = - \sigma \Delta \, dt^2 + \Delta^{-1} 
\left(  e^\psi \, ds_2^2
 + \rho^2 d\phi^2 \right) \;,
 \label{stlineeps0}
\eea
with $ds_2^2$ given by either \eqref{weylframe} or \eqref{weylframe2}.
$\psi$ is obtained by solving \eqref{eq_diff_eq_psi}, and we demand $\Delta > 0$.
Inspection of \eqref{deltsigvr} and
\eqref{mpmt01} however shows that there are regions in the Weyl upper-half plane where $\Delta <0$.
This can be dealt with by changing the overall sign of the monodromy matrix \eqref{monschwarz0}.

We will now first discuss the case $\sigma = 1$, and subsequently the more intricate case $\sigma = -1$.

%%%%%%%%
\subsection{  $\sigma =1$}
%%%%%%%%

When $\sigma =1$, we have 
\be
\tau_0 = \frac{v -  \sqrt{v^2+ \rho^2}}{\rho} \;\;\;,\;\;\;
\widetilde{\tau}_0 =   \frac{v +  \sqrt{v^2+ \rho^2}}{\rho}  \;.
\ee
\\

\noindent
{\bf Case $i)$}: for $\Gamma$ belonging to class $(i)$, we obtain
\begin{eqnarray}
\Delta_i (\rho, v) &=&   \frac{ v +  \sqrt{v^2+ \rho^2} }{2m} > 0 \;, 
\nonumber\\
\psi_i (\rho, v) &=& \log \left[ \frac{v+\sqrt{v^2+  \rho^2}}{2\sqrt{v^2+ \rho^2}} \right] \;,
\end{eqnarray}
up to an integration constant in $\psi_i$. The resulting space-time metric takes the form
\begin{equation}
	ds_4^2 = - \frac{(v+\sqrt{v^2+ \rho^2})}{2m} \, dt^2+\frac{m}{\sqrt{v^2+ \rho^2}}(d\rho^2+dv^2)+\frac{2m \, \rho^2}{v+\sqrt{v^2+\rho^2}} \, d\phi^2.
	\label{rind}
\end{equation}
Now we perform the following change of coordinates, 
\begin{equation}
	\rho = \frac{\tilde{\rho}}{2m} \, \tilde{z} \;\;\;,\;\;\; v = \frac{1}{4m}(\tilde{z}^2- \tilde{\rho}^2) \;\;\;,\;\;\; t = 2m \, \tilde{t} \;.
\end{equation}
Since $\rho > 0$, this requires restricting to either  $\tilde{\rho},\tilde{z}\in ]0,\infty[$ or $\tilde{\rho},\tilde{z}\in ]-\infty,0[$. We choose $\tilde{\rho},\tilde{z}\in ]0,\infty[$.
Then, \eqref{rind} becomes
\begin{equation}
	ds_4^2 = - \tilde{z}^2 d\tilde{t}^2 +d\tilde{z}^2+ d\tilde{\rho}^2 + \tilde{\rho}^2d\phi^2 \;,
\end{equation}
which describes the Rindler metric, i.e.  the uniformly accelerated metric \cite{Griffiths:2009dfa}.\\

Now recall that if $M(\rho,v)$ is a diagonal matrix that 
yields a solution to the field equations, then also $M^{-1} (\rho, v)$ yields a solution to the field equations.  The former solution 
is 
constructed from $(\Delta, \psi)$, whereas the latter is constructed from
$(\Delta^{-1}, \psi)$, and therefore differ from the former in general.
In this latter case, we obtain
for the space-time metric,
\begin{equation}
	ds_4^2 = - \frac{2m}{(v+\sqrt{v^2+\rho^2})}dt^2+\frac{(v+\sqrt{v^2+\rho^2})^2}{4m \, \sqrt{v^2+\rho^2}}( d\rho^2+dv^2)+
	\frac{(v+\sqrt{v^2+\rho^2})  }{2m} \, \rho^2d\phi^2.
	\label{linesig1e0}
\end{equation}
Performing the change of coordinates
\begin{equation} \label{change-coord-sigma-1-omega-inverse}
	\rho = \sqrt{m} \, z\,\sqrt{r} \;\;\;,\;\;\; v = r - \frac{m}{4} \,  z^2 
\end{equation}
with $r>0,\,z>0$, the metric becomes
\begin{equation} 
	ds_4^2 = -\frac{m}{r}dt^2 + \frac{r}{m}dr^2+r^2(dz^2+z^2d\phi^2),
\end{equation}
which is the $AIII$-metric \eqref{aiii-r-space-like} for space-like $r$.
\\

\noindent
{\bf Case $ii)$}: for $\Gamma$ belonging to class $(ii)$, we obtain
\begin{eqnarray}
\Delta_{ii} (\rho, v) =   \frac{ v - \sqrt{v^2+ \rho^2} }{2m} <0 \;,
\end{eqnarray}
which requires changing the overall sign of the monodromy matrix \eqref{monschwarz0}, as mentioned above.
The expression for $\psi_{ii} (\rho, v)$ is, up to an integration constant, given by
\begin{eqnarray}
\psi_{ii} (\rho, v) = \log \left| \frac{v-\sqrt{v^2+  \rho^2}}{2\sqrt{v^2+ \rho^2}} \right| \;.
\label{psi22222}
\end{eqnarray}
Since the coordinate transformation $v \mapsto -v$ converts the expressions for $\Delta_{ii}$ and $\psi_{ii}$ into those for 
 $\Delta_{i}$ and $\psi_{i}$ given above, the resulting space-time metric agrees with the one of Case $i)$.
 
A similar reasoning applies to the solution obtained from $M^{-1} (\rho, v)$.

%%%%%%%
\subsection{$\sigma =-1$}
%%%%%

When $\sigma =-1$, we have 
\be
\tau_0 = \frac{- v +  \sqrt{v^2- \rho^2}}{\rho} \;\;\;,\;\;\;
\widetilde{\tau}_0 =  -  \frac{v +  \sqrt{v^2 - \rho^2}}{\rho}  \;.
\ee
Imposing that $\tau_0$ and $\widetilde{\tau}_0$ are real, so as to ensure that $\Delta$ is real, requires restricting $(\rho,v)$ to lie
in one of the following regions of the  Weyl upper-half plane: either 
$v > \rho$ or $v< - \rho$.

Let us first consider the case when $(\rho,v)$ takes values in the region $v > \rho$.\\

\noindent
{\bf Case $i)$}: for $\Gamma$ belonging to class $(i)$, we obtain
\begin{eqnarray}
\Delta_{i} (\rho, v) =   \frac{ -v - \sqrt{v^2 - \rho^2} }{2m} <0 \;,
\end{eqnarray}
which requires changing the overall sign of the monodromy matrix \eqref{monschwarz0}, as mentioned above.
The expression for $\psi_{i} (\rho, v)$ is, up to an integration constant, given by
\begin{eqnarray}
\psi_{i} (\rho, v) = \log \left( \frac{v+\sqrt{v^2-  \rho^2}}{2\sqrt{v^2- \rho^2}} \right) \;.
\end{eqnarray}
We take the space-time metric to be given by \eqref{stlineeps0} with \eqref{weylframe},
\begin{equation}
	ds_4^2 =  \frac{(v+\sqrt{v^2-  \rho^2})}{2m} dt^2+\frac{m}{\sqrt{v^2 - \rho^2}}( - d\rho^2+dv^2)+\frac{2m \, \rho^2}{v+\sqrt{v^2-\rho^2}}d\phi^2.
	\label{rindddd}
\end{equation}
Now we perform the following change of coordinates, 
\begin{equation}
	\rho = \frac{\tilde{\rho}  \, \tilde{z}}{2m}  \;\;\;,\;\;\; v = \frac{1}{4m}(\tilde{z}^2 + \tilde{\rho}^2) \;\;\;,\;\;\; t = 2m \, \tilde{t} \;.
\end{equation}
Since $\rho > 0$, this requires restricting to either  $\tilde{\rho},\tilde{z}\in ]0,\infty[$ or $\tilde{\rho},\tilde{z}\in ]-\infty,0[$.
We choose  $\tilde{\rho},\tilde{z}\in ]0,\infty[$.
Moreover, the requirement $v > \rho$ implies $( \tilde{z} - \tilde{\rho} )^2 > 0$. 
When  $\tilde{z}>\tilde{\rho}$ we obtain
\begin{equation}
	ds_4^2 = \tilde{z}^2d\tilde{t}^2 -d\tilde{\rho}^2+d\tilde{z}^2  +\tilde{\rho}^2d\phi^2 \;,
\end{equation}
while for $\tilde{z}<\tilde{\rho}$ we obtain
\begin{equation}
	ds_4^2 = \tilde{\rho}^2d\tilde{t}^2+d\tilde{\rho}^2-d\tilde{z}^2+\tilde{z}^2d\phi^2 \;.
\end{equation}
Both these solutions describe a Kasner metric \eqref{eq_Kasner} with exponents $p_1 =1, p_2=p_3=0$,
if $t$ is taken to be an angular coordinate, i.e. $0 \leq \tilde{t} < 2 \pi$.\\

As mentioned above, if 
$M(\rho,v)$ is a diagonal matrix that 
yields a solution to the field equations, also $M^{-1} (\rho, v)$ yields a solution to the field equations.
In this latter case, we take the space-time metric to be given by \eqref{stlineeps0} with \eqref{weylframe2},
\begin{equation} \label{metric-inverse-v-greater-rho}
	ds_4^2 =  \frac{2m}{ v + \sqrt{v^2-\rho^2}}dt^2+\frac{(v + \sqrt{v^2-\rho^2})^2}{4m \sqrt{v^2-\rho^2}}(d\rho^2-dv^2)+\frac{(v + \sqrt{v^2-\rho^2})}{2m} \, \rho^2d\phi^2.
\end{equation}
We perform the change of coordinates 
\begin{equation}
	\rho = \sqrt{m} \, z\,\sqrt{r}   \;\;\;,\;\;\;	v = r + \frac{m}{4} z^2 \;,
\end{equation}
with $r > m z^2/4 $. 
We obtain
\begin{equation} \label{metric-aiii-minusz}
	ds_4^2 = \frac{m}{r}dt\,^2 - \frac{r}{m}dr^2+r^2(dz^2+z^2d\phi^2),
\end{equation}
which is  the $AIII$-metric \eqref{aiii-r-time-like} for time-like $r > m z^2/4 $.\\

\noindent
{\bf Case $ii)$}: for $\Gamma$ belonging to class $(ii)$, we obtain
\begin{eqnarray}
\Delta_{ii} (\rho, v) =   \frac{ -v + \sqrt{v^2 - \rho^2} }{2m} <0 \;,
\end{eqnarray}
which requires changing the overall sign of the monodromy matrix \eqref{monschwarz0}, as mentioned above.
The expression for $\psi_{ii} (\rho, v)$ is, up to an integration constant, given by
\begin{eqnarray}
\psi_{ii} (\rho, v) = \log \left| \frac{- v+\sqrt{v^2-  \rho^2}}{2\sqrt{v^2- \rho^2}} \right| \;.
\label{psi222}
\end{eqnarray}
We take the space-time metric to be given by \eqref{stlineeps0} with \eqref{weylframe},
\begin{equation}
	ds_4^2 =  \frac{(v-\sqrt{v^2-  \rho^2})}{2m} dt^2+\frac{m}{\sqrt{v^2 - \rho^2}}( - d\rho^2+dv^2) + \frac{2m \, \rho^2}{v-\sqrt{v^2-\rho^2}}d\phi^2.
	\label{kasn-}
\end{equation}
We perform the change of coordinates 
\begin{equation}
	\rho = \frac{\tilde{\rho} \, \tilde{z}}{2m}  \;\;\;,\;\;\; v = \frac{1}{2}(\tilde{z}^2 + \tilde{\rho}^2)  \;\;\;,\;\;\; t = 2m \, \tilde{t} \;,
\end{equation}
with the choice $\tilde{\rho},\tilde{z}\in ]0,\infty[$. 
The condition $v > \rho$ becomes $(\tilde{z} - \tilde{\rho})^2 >0$.
When $\tilde{z}>\tilde{\rho}$ we obtain
\begin{equation}
	ds_4^2 = \tilde{\rho}^2d\tilde{t}^2-d\tilde{\rho}^2 +d\tilde{z}^2  +\tilde{z}^2d\phi^2 \;,
\end{equation}
while for $\tilde{z}<\tilde{\rho}$ we obtain
\begin{equation}
	ds_4^2 = \tilde{z}^2d\tilde{t}^2+d\tilde{\rho}^2-d\tilde{z}^2+\tilde{\rho}^2d\phi^2 \;.
\end{equation}
Both these solutions describe a Kasner metric \eqref{eq_Kasner} with exponents $p_1 =1, p_2=p_3=0$.\\

Now let us consider $M^{-1} (\rho, v)$. This again requires changing the overall sign of the monodromy matrix. We take the space-time metric to be given by \eqref{stlineeps0} with \eqref{weylframe},
\begin{equation} \label{metric-inverse-v-greater-rho}
	ds_4^2 =  \frac{2m}{ v - \sqrt{v^2-\rho^2}}dt^2+\frac{(v - \sqrt{v^2-\rho^2})^2}{4m \sqrt{v^2-\rho^2}}(-d\rho^2+dv^2)+\frac{(v- \sqrt{v^2-\rho^2})}{2m} \, \rho^2d\phi^2.
\end{equation}
We perform the change of coordinates 
\begin{equation}
	\rho = \sqrt{m} \, z\,\sqrt{r}   \;\;\;,\;\;\;	v = r + \frac{m}{4} z^2 \;,
\end{equation}
with  $r < m z^2/4 $. 
We obtain
\begin{equation} \label{metric-aiii-minusz2}
	ds_4^2 = \frac{m}{r}dt\,^2 - \frac{r}{m}dr^2+r^2(dz^2+z^2d\phi^2),
\end{equation}
which is  the $AIII$-metric \eqref{aiii-r-time-like} for time-like $r < m z^2/4 $. Thus, combining \eqref{metric-aiii-minusz} with 
\eqref{metric-aiii-minusz2} results in the  $AIII$-metric \eqref{aiii-r-time-like}.\\

Finally, let us consider the case when $(\rho,v)$ takes values in the region $v <-\rho$, or equivalently, $\tilde{v} = -v > \rho$.
Noting that 
\begin{eqnarray}
\Delta_{i} (\rho, v) &=&  \frac{  \tilde{v} - \sqrt{\tilde{v}^2 - \rho^2} }{2m} > 0 \;\;\;,\;\;\;
\Delta_{ii} (\rho, v) =  \frac{ \tilde{v} + \sqrt{\tilde{v}^2 - \rho^2} }{2m}  > 0 \;, \nonumber\\
\psi_{i} (\rho, v) &=& \log \left| \frac{- \tilde{v}+\sqrt{\tilde{v}^2-  \rho^2}}{2\sqrt{\tilde{v}^2- \rho^2}} \right| \;\;\;,\;\;\;
\psi_{ii} (\rho, v) = \log \left(\frac{\tilde{v} +\sqrt{\tilde{v}^2-  \rho^2}}{2\sqrt{\tilde{v}^2- \rho^2}} \right) \;,
\end{eqnarray}
we see that the discussion of the case $v <-\rho$ gets mapped to the previous discussion of the case  $v > \rho$.

%%%%%%%%
\section{Solutions with two Killing horizons \label{sec:nsol}}
%%%%%%%%%%%%%%%%%%%%%%

In this section we consider 
the four-dimensional Einstein-Maxwell-dilaton theory that is obtained by Kaluza-Klein reduction of five-dimensional Einstein gravity.
We take the monodromy ${\cal M}$ that was 
introduced in \cite{Cardoso:2017cgi}, 
\begin{eqnarray}
{\cal M}(u)  =
\left( \frac{H_2}{H_1} \right)^{1/3} 
  \begin{pmatrix}
H_1 H_2   &\quad  \sqrt{2} H_1  &\quad  - 1 \\
- \sqrt{2} H_1  &\quad  - H_1/H_2 
& \quad  0 \\
- 1 &\quad  0 & \quad 0 
\end{pmatrix} \;\;\;,\;\;\; \det {\cal M} = 1 \;,
\label{monomatrixMHH}
\end{eqnarray}
where
\be
H_1 (u)  = h_1 + \frac{ Q}{u} \;\;\;,\;\;\;  H_2 (u) =  h_2 + \frac{ P}{u} \;.
\ee
Here $Q$ denotes an electric charge, while $P$ denotes a magnetic charge. The constants $h_1, h_2 \in \mathbb{R}$ denote two deformation
parameters that are taken to be positive, $h_1, h_2 > 0$.  

We consider the spectral curve relation \eqref{utau} with $\sigma =1$.
When substituting $u$ by the right hand side of \eqref{utau}, the resulting matrix ${\cal M}_{(\rho,v)} (\tau)$ exhibits
6 simple poles located at the following points in the $\tau$-plane,
\begin{eqnarray}
\tau_{0 \, \pm }&=& \frac{1}{\rho}\Big (v \pm \sqrt{\rho^2 + v^2}\Big ) \;, \nonumber\\
\tau_{\tilde Q \,  \pm } &=& \frac{1}{\rho} \left(v + {\tilde Q}
\pm \sqrt{\rho^2 + (v+ {\tilde Q})^2} \right) \;, \nonumber\\
\tau_{\tilde P \, \pm} &=& \frac{1}{\rho} \left(v + {\tilde P} \pm \sqrt{\rho^2 + (v+{\tilde P})^2} \right) \;.
\label{tval}
\end{eqnarray}
Here, ${\tilde Q}$ and ${\tilde P}$ denote the rescaled charges ${\tilde Q} = Q/h_1 \;,\; {\tilde P} = P/h_2 $. 
Note that the values \eqref{tval} are all real, and can thus never coincide with the two fixed points $\pm i$ of the involution
$\tau \mapsto - 1/\tau$. As a consequence, one can always pick a contour that avoids passing through these simple poles. 
Note that $\tau_{0 \, +},  \tau_{\tilde Q \, +}, \tau_{\tilde P \, +} >0$, while  $\tau_{0 \, -},  \tau_{\tilde Q \, -}, \tau_{\tilde P \, -} <0$.
As shown in \cite{Cardoso:2017cgi}, the canonical factorization of ${\cal M}_{(\rho,v)} (\tau)$ with respect to a chosen contour
always exists.

There are several contours that one may pick to perform the canonical factorization of ${\cal M}_{(\rho,v)} (\tau)$.
These various choices of contour will, in general, result in different solutions to the field equations.  
In  \cite{Cardoso:2017cgi} a specific contour was chosen, namely the unit circle in the $\tau$-plane.
Moreover, the discussion presented there was restricted to a certain region in the  Weyl upper-half plane, 
to ensure that the three poles $\tau_{0 \, -},  \tau_{\tilde Q \, -}, \tau_{\tilde P \,-} $ lie inside the unit circle, 
and the poles $\tau_{0 \, +},  \tau_{\tilde Q \, +}, \tau_{\tilde P \, +} $ lie outside. This was achieved by taking ${\tilde Q}, {\tilde P} >0$
and restricting to $v>0$.  The latter restrictions are a consequence of having chosen the unit circle as a factorization contour.
These restrictions are unnecessary: they can be avoided by choosing a contour that encircles the origin and passes through the fixed points $\pm i$ 
such that 
$\tau_{0 \, -},  \tau_{\tilde Q \, -}, \tau_{\tilde P \, -} $ lie inside the contour. Such a contour always exists,
in view of the fact that $\tau_{0 \, -},  \tau_{\tilde Q \, -}, \tau_{\tilde P \, -} $ are real and negative  for any choice of charges
 ${\tilde Q}, {\tilde P}$ and for any $(\rho,v)$ in the Weyl upper-half plane. In the following, we discuss the canonical factorization
 of  ${\cal M}_{(\rho,v)} (\tau)$ with respect to such a contour.  A similar discussion can be carried out for any other choice of
 contour, but we will refrain from doing so in this paper.  We will follow the notation used in  \cite{Cardoso:2017cgi} and use superscripts $\pm$
 to denote the values \eqref{tval} that lie inside the contour by $\tau^+$ and the values that lie outside the contour by $\tau^-$.
 Thus, for the choice of contour described above, we set $\tau_{0 \, -} = \tau_0^+,  \tau_{\tilde Q \, -} = \tau_{\tilde Q}^+, \tau_{\tilde P \, -} = 
 \tau_{\tilde P}^+$.
 
 The resulting four-dimensional space-time
metric takes the form
\begin{equation}
ds^2_4 = - \Delta 
\, \left(dt + B \, d\phi \right)^2 + \Delta^{-1}  \left(e^{\psi} \, (d \rho^2 + dv^2) + 
\rho^2 \, d\phi^2 \right) \;,
\label{linePQg0}
\end{equation}
where $\Delta, B$ and $e^{\psi} $ are functions of $(\rho,v)$.
The solution carries electric-magnetic charges $(Q,P)$, and hence 
is supported by an electric-magnetic field $F$. It is also supported by a scalar field $e^{- 2 \Phi}$.

 We begin by considering the case ${\tilde Q}, {\tilde P} >0$.
  When 
${\tilde Q} = {\tilde P}$, the solution resulting from the canonical factorizaton of ${\cal M}_{(\rho,v)} (\tau)$ describes a four-dimensional
extremal black hole supported by a constant dilaton field. We therefore take ${\tilde Q} \neq {\tilde P}$.
For concreteness, we choose ${\tilde J} \equiv {\tilde P} - {\tilde Q} >0$. The resulting expressions 
for  $\Delta, B, e^{\psi}, F$ and  $e^{- 2 \Phi}$ are given in 
Appendix \ref{nmetQP}.

As $\rho^2 + v^2 \rightarrow \infty$, the solution behaves as follows. It asymptotes to a stationary solution with an effective NUT parameter ${\tilde J}$
which is expressed in terms of the electric-magnetic charges,
\bea
ds^2_4 &=& - \frac{1}{h_1 h_2} \left(1 - \frac{{\tilde P} + {\tilde Q}}
{\sqrt{\rho^2 + v^2}} \right)
\, \left(dt - h_1 h_2 \, {\tilde J} \, \frac{v}{\sqrt{\rho^2 + v^2}} \, d\phi \right)^2 \nonumber\\
&& + h_1 h_2  \left(1 + \frac{{\tilde P} + {\tilde Q}}
{\sqrt{\rho^2 + v^2}} \right)
\left(d \rho^2 + dv^2 + 
\rho^2 \, d\phi^2 \right) \;,
\eea
with $e^{- 2 \Phi} \rightarrow h_1/h_2$.

The metric \eqref{linePQg0} has two Killing horizons, $ || \partial/ \partial t||^2 = \Delta = 0$,
located at $\rho = v = 0$ and at $\rho = 0, - {\tilde P} < v < - {\tilde Q}$, respectively.
The latter could not be detected in the analysis of \cite{Cardoso:2017cgi}, because the choice of the unit circle as a factorization contour only allowed
to consider the region $v>0$.
When approaching the 
Killing horizon $\rho = v = 0$, keeping  $\rho/v$ constant,
the metric takes the form
\begin{equation}
ds^2_4 = - \frac{\rho^2 + v^2}{P Q} 
\, \left(dt +  h_1 h_2 \, {\tilde J} \, f\left(v/\sqrt{\rho^2 + v^2} \right) \,  d\phi \right)^2 + \frac{PQ}{\rho^2 + v^2}  \left( d \rho^2 + dv^2 + 
\rho^2 \, d\phi^2 \right) \;,
\end{equation}
where $f(x) = x (1-x) -1$ denotes a linear combination of the Legendre polynomials $P_0, P_1$ and $P_2$. 
The scalar field approaches the value $e^{- 2 \Phi} \rightarrow P/Q$. Note, however, that $\partial_{\rho,v} e^{- 2 \Phi} $
does not vanish at $\rho=v=0$ when $\tilde J \neq 0$. Thus, only when ${\tilde J} =0$ 
(that is, in the extremal black hole case)
does the solution exhibit an attractor behaviour as
one approaches $\rho=v=0$. 
Both the Ricci and the Kretschmann scalars are well-behaved at the Killing horizon $\rho=v=0$.
However, they both blow up at the Killing horizon $\rho = 0, \, -\tilde{P}< v < -\tilde{Q}$, which points to the existence of a curvature singularity at this horizon.
The scalar field 
$e^{- 2 \Phi}$ also diverges at this Killing horizon.

Thus, this solution describes a space-time that is supported by one electric and one magnetic charge and by a scalar field, possesses two Killing horizons and asymptotes
to a space-time with an effective NUT parameter $\tilde J = {\tilde P} - {\tilde Q}$ that is expressed in terms of the electric-magnetic charges.
To the best of our knowledge, this four-dimensional solution has not been given in the literature before. This solution has similarities
with a solution discovered by Brill \cite{Brill} (see also \cite{Stephani:2003tm}) in a different four-dimensional theory, namely an Einstein+Maxwell theory, in the sense
that it possesses two Killing horizons, one of them being associated with the presence of a NUT parameter. However, 
while Brill's solution describes an electrically charged (or magnetically charged) Reissner-Nordstrom black hole when the NUT parameter is switched off, our solution
describes a dyonic extremal black hole solution (that is supported by a scalar field) when the NUT parameter $\tilde J$ is set to zero. Moreover,
differently from Brill's solution, the NUT parameter
${\tilde J}$ is not an additional parameter, but rather an effective parameter that is expressed in terms of the electric-magnetic charges.

Let us now discuss what happens to the solution when one sends the parameters $h_1$ and $h_2$ to zero, keeping $Q$ and $P$ fixed.
 In this limit the solution becomes the  $AdS_2 \times S^2$ solution, supported by a constant scalar field $e^{- 2 \Phi} = P/Q$,
 that
 describes the near-horizon region of the extremal black hole solution.

Next, let us briefly consider the case ${\tilde Q} <0, {\tilde P} >0$. Performing the canonical factorization of ${\cal M}_{(\rho,v)} (\tau)$ with respect to
the contour described above, results in a solution that,
as $\rho^2 + v^2 \rightarrow \infty$, asymptotes to a stationary solution with an effective NUT parameter ${\tilde J} = {\tilde P} - {\tilde Q}$,
as in the previous case,
\bea
ds^2_4 &=& - \frac{1}{h_1 h_2} \left(1 - \frac{{\tilde P} + {\tilde Q}}
{\sqrt{\rho^2 + v^2}} \right)
\, \left(dt - h_1 h_2 \, {\tilde J} \, \frac{v}{\sqrt{\rho^2 + v^2}} \, d\phi \right)^2 \nonumber\\
&& + h_1 h_2  \left(1 + \frac{{\tilde P} + {\tilde Q}}
{\sqrt{\rho^2 + v^2}} \right)
\left(d \rho^2 + dv^2 + 
\rho^2 \, d\phi^2 \right) \;,
\eea
with $e^{- 2 \Phi} \rightarrow h_1/h_2$. Note that now the combination ${\tilde P} + {\tilde Q}$ may be negative.
Moreover, unlike in the previous case, 
 the function $\Delta (\rho,v)$ in the metric
\eqref{linePQg0}
will generically cease to be real in a region of the Weyl upper-half plane. 
The function $\Delta^{-2}$ is given by $\Delta^{-2} = (g m_1)^2 \, e^{2 \Sigma_2}$, c.f. Appendix \ref{nmetQP}.
The curve in the Weyl upper-half plane where $e^{2 \Sigma_2}$ vanishes determines the boundary where $\Delta$ ceases to be real.
Note that this boundary does not coincide with the axis $\rho =0$, and that
the Ricci and the Kretschmann scalars blow up when approaching this boundary. 
See Figure \ref{PQneg} for examples of the curve in the Weyl upper-half plane where $e^{2 \Sigma_2}$ vanishes.

\begin{figure*}[h!]
\begin{subfigure}{0.22\textwidth}
\includegraphics[scale=0.28]{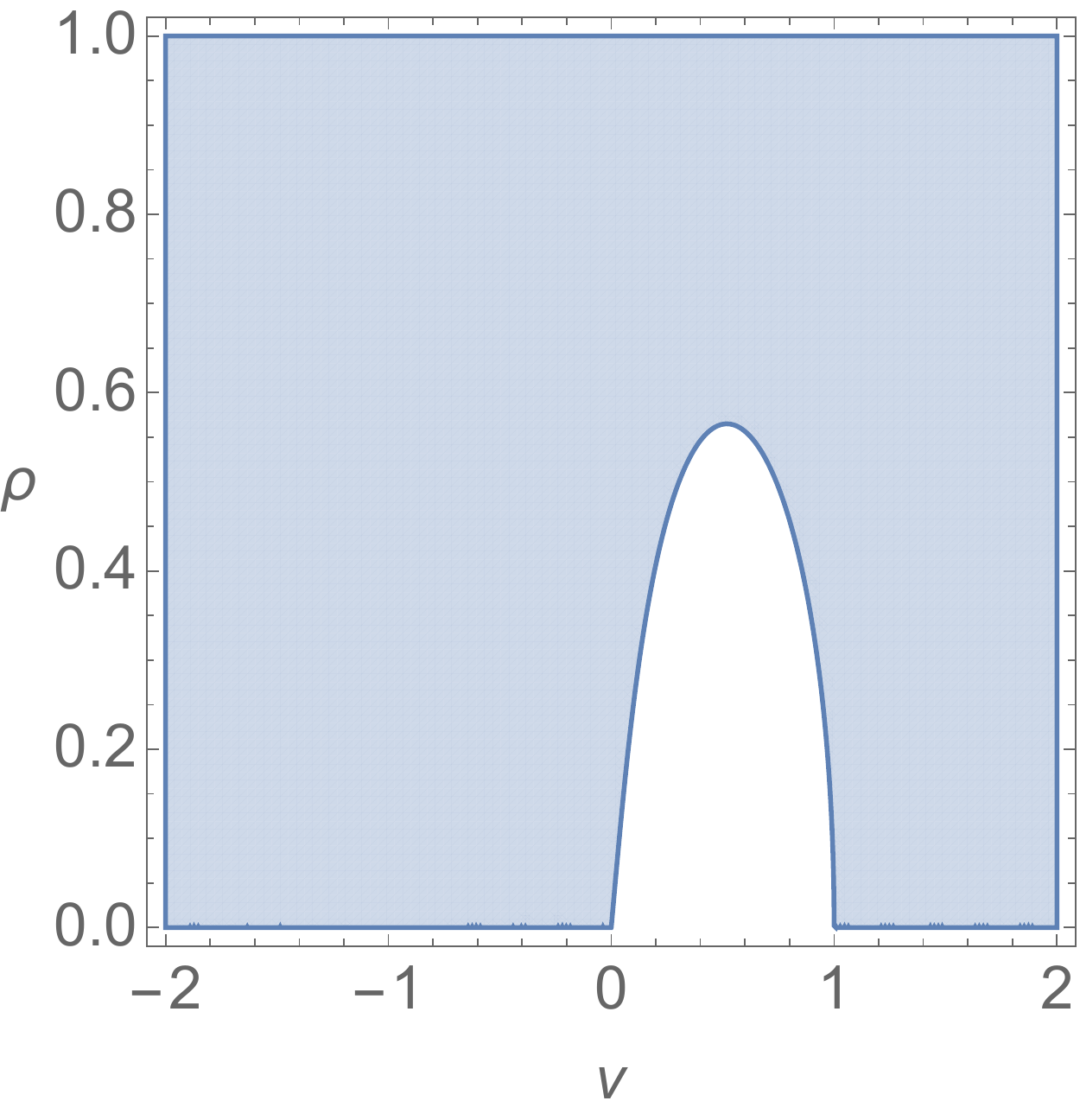}
%plot_P_equals_Q.pdf} 
%\caption{$\tilde{P}=-\tilde{Q}$. } 	
\caption{$P=- Q=1$, $h_1=h_2=1$. } 
%\label{fig_Q_tilde_equals_mP_tilde}
\end{subfigure} \hspace*{\fill}
\begin{subfigure}{0.22\textwidth}
	\includegraphics[scale=0.28]{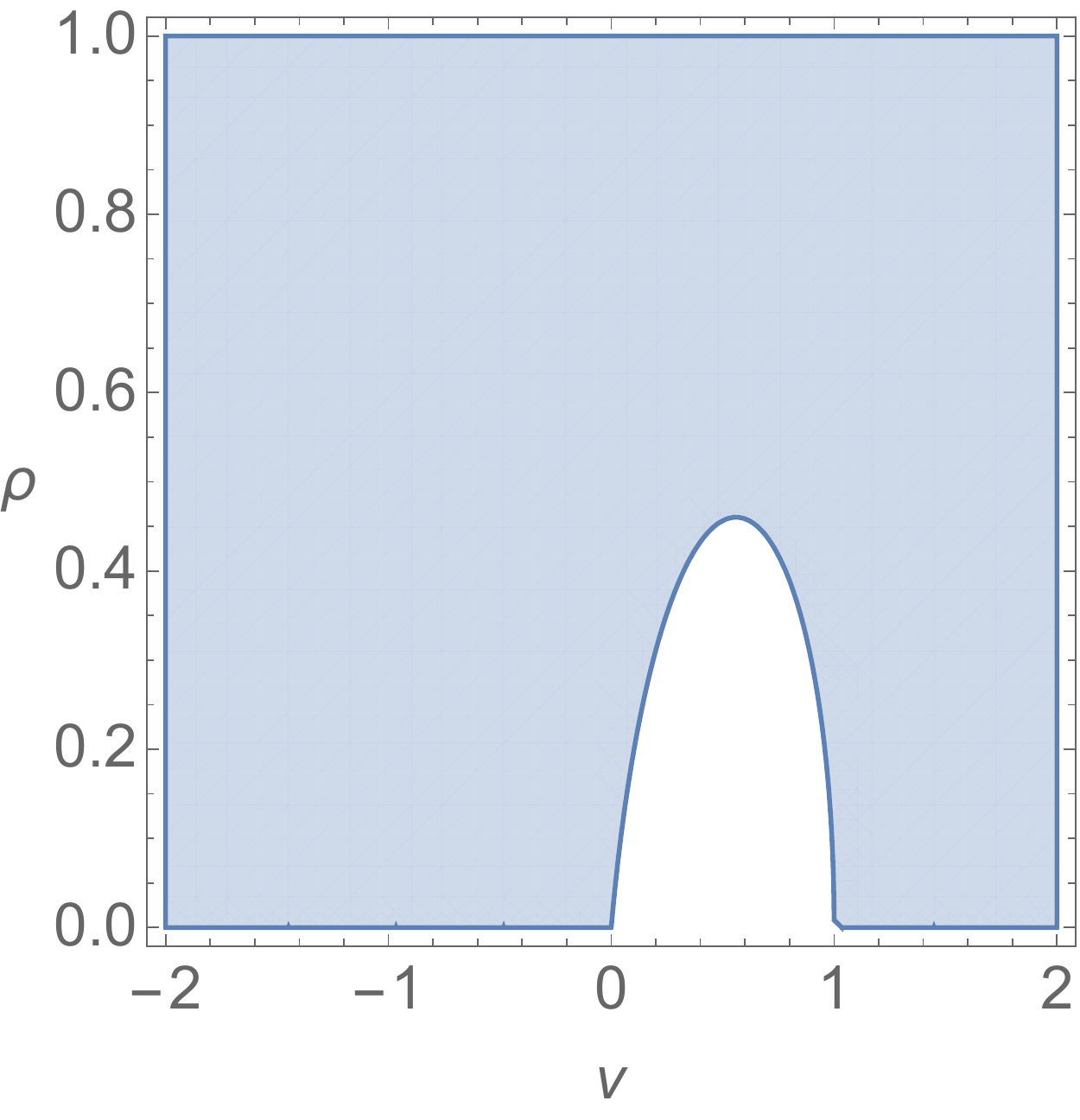}  
	\caption{$P=-Q=1$, $h_1=1$,  $h_2=1/9$. } 
%	\label{fig_Q__equals_mP_h1_1_h2_1_9}
\end{subfigure}\hspace*{\fill}
\begin{subfigure}{0.22\textwidth}
	\includegraphics[scale=0.28]{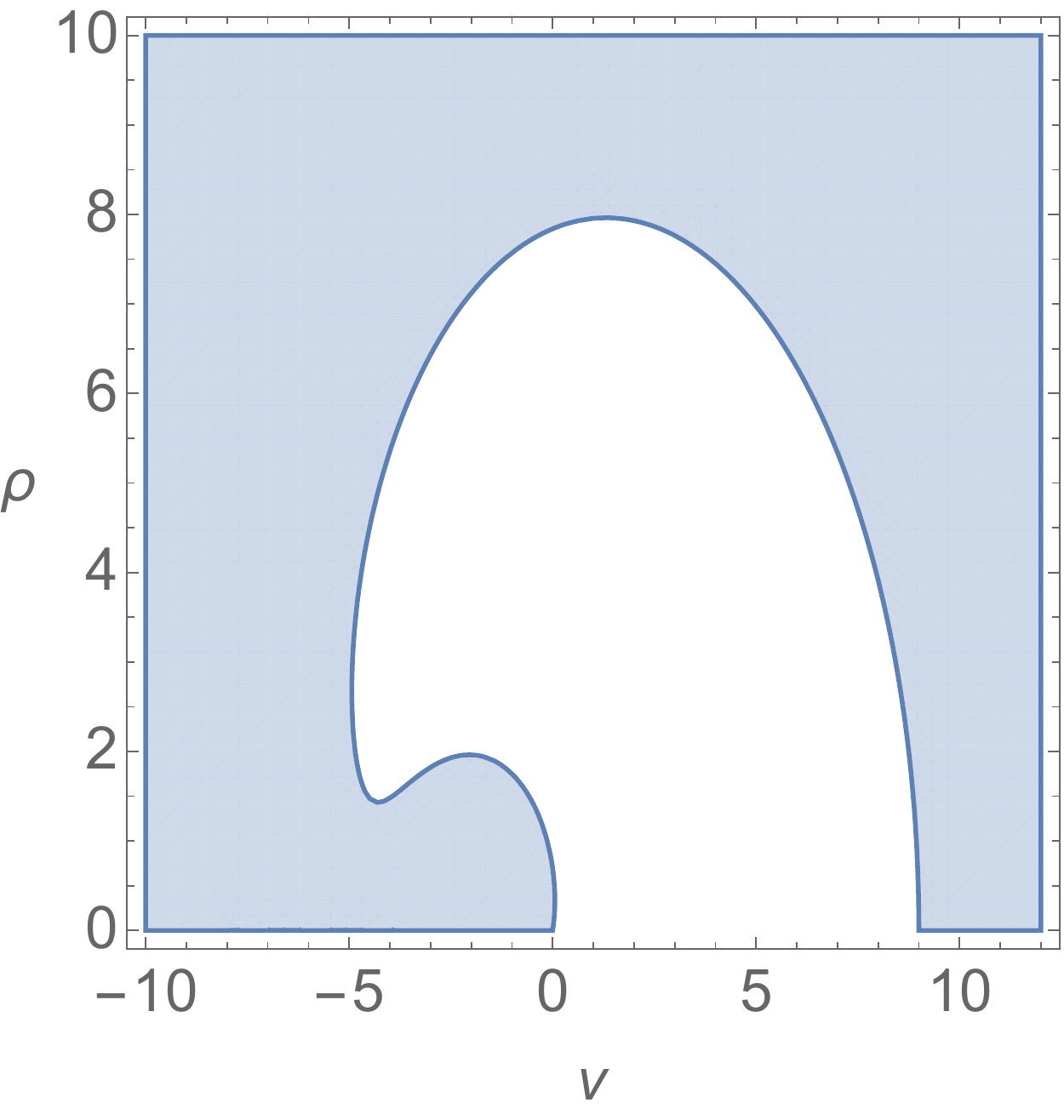} 
	\caption{$P=-Q=1$, $h_1=1/9$,  $h_2=1$. } 	
%	\label{fig_Q__equals_mP_h1_1_9_h2_1}
\end{subfigure}
\hspace*{\fill}
\begin{subfigure}{0.22\textwidth}
	\includegraphics[scale=0.28]{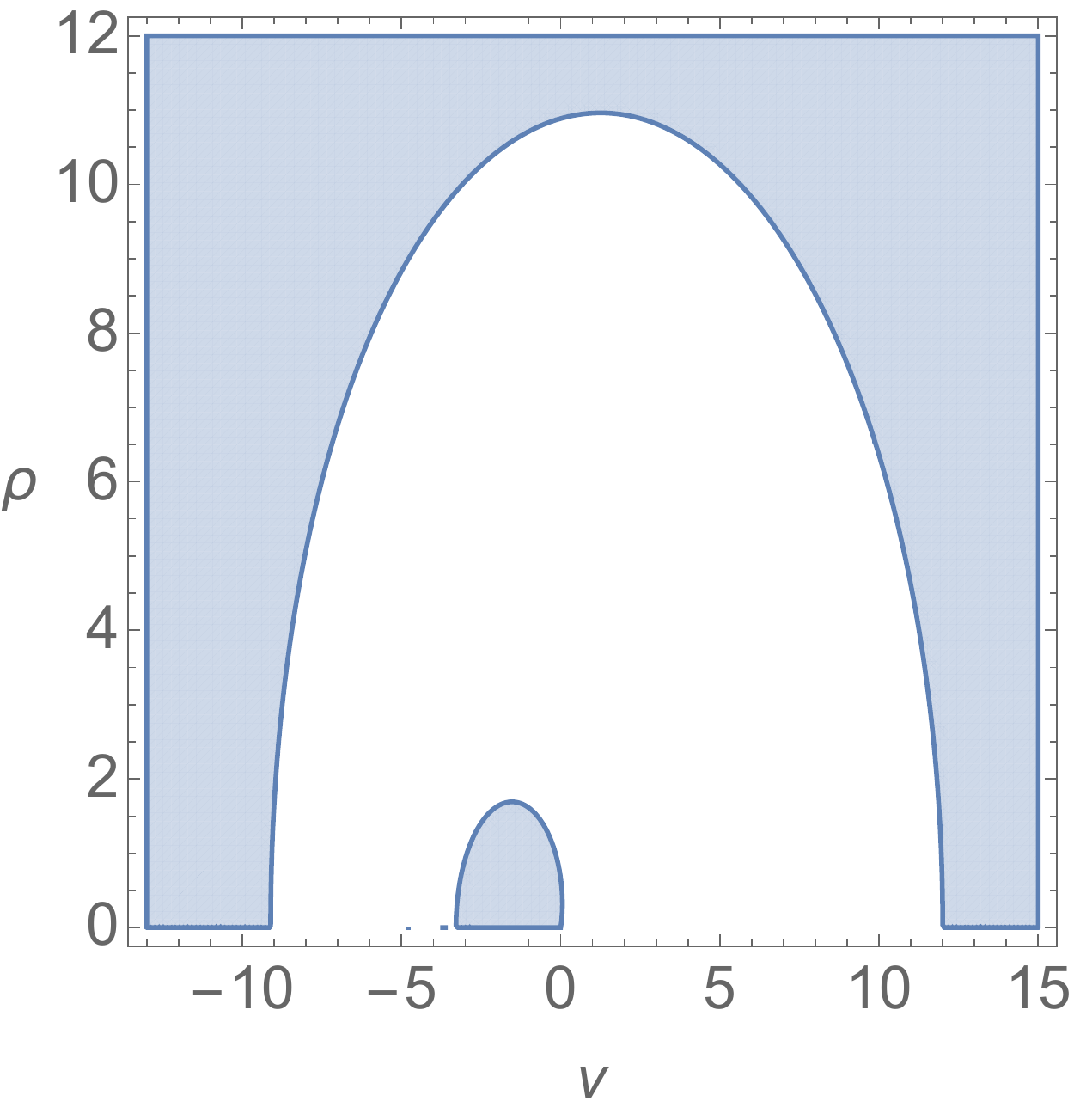}
	\caption{$P=-Q=1$, $h_1=1/12$, $h_2=1$. } 	
	%\label{fig_Q__equals_mP_h1_1_12_h2_1}
\end{subfigure}\hspace*{\fill}
\caption{Examples of the curve $e^{2\Sigma_2} =0$ for different values of $\tilde{P}$ and $\tilde{Q}$ with $Q P <0$. \label{PQneg}}
\end{figure*}

Finally,  let us discuss what happens to the solution when one sends the parameters $h_1$ and $h_2$ to zero, keeping $Q$ and $P$ fixed.
For the various function in the line element \eqref{linePQg0} we obtain
\begin{equation}
e^{\psi} =  \rho^\frac{4}{9} \,, \quad B = 2 P Q \,  \frac{v-\sqrt{\rho^2+v^2}}{\rho^2+v^2}  +f_1\,, \quad \Delta = -  \frac{\rho^2+v^2}{P Q} \frac{\rho}{\sqrt{\rho^2+ 4v\left(v-\sqrt{\rho^2+v^2}\right) }} \;,
\end{equation}
where $f_1$ denotes an integration constant. When solving the differential
equations \eqref{diffpsi} for $\psi$, we took the resulting integration constant to equal
\begin{equation}
c_1 =  \frac{4}{9}\log\left[\frac{2 {\tilde P} {\tilde  Q}}{{\tilde Q}-{\tilde P}} \right] \,,
\end{equation}
to ensure that the differential
equations \eqref{diffpsi} remain well-defined in the limit $h_1, h_2 \rightarrow 0$.
Note that when $v > 0$, $\Delta$ is only real for $\rho \geq 2 \sqrt{2} v$.

Thus, the case ${\tilde Q} <0, {\tilde P} >0$ is markedly different from the case ${\tilde Q} >0, {\tilde P} >0$ discussed earlier.

%%%%%%%%%%%%%%%%%%%%%%%%%%%%%%%%%%
\subsection*{Acknowledgements}
We would like to thank Thomas Mohaupt and Suresh Nampuri
for useful discussions. M.C. C\^amara and G.L. Cardoso thank David Krej\v{c}i\v{r}\'ik for hospitality at the
Department of Mathematics, Czech Technical University Prague,
during the course of this work.
This work was partially
supported by FCT/Portugal through UID/MAT/04459/2019
and through the LisMath PhD fellowships PD/BD/128415/2017 (P. Aniceto) and 
PD/BD/135527/2018 (M. Rossell\'o). 
 %%%%%%%%%%%%%%%%%%%%%%%%%%

%%%%%%%%%%%%%%%%%
\appendix

%%%%%%%%%%%%%%%%%%%%

%%%%%%%%%%
\section{Expressions for $ {\tilde P} > \tilde {Q} >0$  \label{nmetQP}}
%%%%%%%%%%%%%

We give the expressions for the solution with $ {\tilde P} > \tilde {Q} >0$ discussed in Section \ref{sec:nsol}.

We introduce the quantities \cite{Cardoso:2017cgi},
\bea
g &=&    \Big(\frac{h_2}{h_1}\Big)^{1/3} \Big(
\frac{\tau_{\tilde P}^{-}}{\tau_{\tilde Q}^{-}}
\Big )^{1/3}  \;, \nonumber\\
m_1 &=& h_1 h_2  \left( 1 - \frac{2 {\tilde Q}}{\rho (\tau_0^+ - \tau_0^-)}  \right)
  \left( 1 - \frac{2 {\tilde P}}{\rho (\tau_0^+ - \tau_0^-)}  \right) 
   - 2 h_1 h_2  \, \frac{(\tau_{\tilde Q}^+ - \tau_{\tilde P}^+ ) (\tau_0^+ - \tau_{\tilde P}^-)
(\tau_0^+ - \tau_{\tilde Q}^+)}{\tau_{\tilde Q}^+ (\tau_0^+ - \tau_0^- )^2} \;,
    \nonumber\\
m_2 &=& 
\sqrt 2 \, h_1  \left( 1 - \frac{2 {\tilde Q}}{\rho (\tau_0^+ - \tau_0^-)}  \right)
- \sqrt{2} \, h_1  \, \frac{(\tau_{\tilde Q}^+ - \tau_{\tilde P}^+ ) 
(\tau_0^+ - \tau_{\tilde Q}^+)}{\tau_{\tilde Q}^+ (\tau_0^+ - \tau_0^- )} \;,
\nonumber\\
m_3 &=& - \frac{h_1}{h_2} \, \frac{\tau_{\tilde Q}^-}{\tau_{\tilde P}^-} = - \frac{1}{g^{3}} \;, \nonumber\\
\chi_1 &=& - m_2/(m_1 \, m_3 + (m_2)^2) \;\;\;,\;\;\; \chi_2 = m_2/m_1 \;\;\;,\;\;\; \chi_3 = -1/m_1 
 \;, \nonumber\\
e^{2 \Sigma_1} &=& m_1 \, g \;\;\;,\;\;\; e^{2 \Sigma_2}  = (m_3 + (m_2)^2/m_1) \, g \:.
\eea
Then \cite{Cardoso:2017cgi}
\begin{eqnarray}
\Delta^{-2} &=&  g^3 \, m_1 \, (m_1 \, m_3 + (m_2)^2) \;, \nonumber\\
e^{- 2 \Phi} &=& g^{3/2} \, \left(m_3 + \frac{(m_2)^2}{m_1} \right)^{3/2} \;.
\end{eqnarray}
%%%%%%%%%%%%%%%%%%%%%

In \cite{Cardoso:2017cgi}, the 
expressions  for $\psi$ and for $B$ were studied by taking ${\tilde J} = {\tilde P} - {\tilde Q}$ to be small and 
performing a power series expansion in ${\tilde J}$. Only the first few terms in these expansions were given.
Here we give the exact expressions which, when power expanded in ${\tilde J}$, reproduce the results 
given in \cite{Cardoso:2017cgi}. 
$\psi$ 
 is obtained by integrating
\begin{eqnarray}
\partial_{\rho} \psi &=& \tfrac14 \rho \,  g^4 \, \left[ \Big( \partial_{\rho} (g \, m_3) \Big)^2 
-  \Big( \partial_{v} (g \, m_3) \Big)^2 \right]
\;, \nonumber\\
\partial_{v} \psi &=& \tfrac12  \, \rho \, g^4 \, \partial_{\rho} (g \, m_3)
 \, \partial_v (g \, m_3)   \;.
 \label{diffpsi}
\end{eqnarray}
We obtain
\begin{equation}
\psi = \frac{1}{9} \log \left[ \frac{\left(\rho^2 + \left(\tilde{P}+v\right) \left(\tilde{Q}+v\right) +\sqrt{\rho^2 +(\tilde{P}+v)^2} \sqrt{\rho^2 +(\tilde{Q}+v)^2}   \right)^2}{4 \left(  \rho^2 +(\tilde{P}+v)^2\right)   \left(\rho^2 +(\tilde{Q}+v)^2\right)} \right] + c_1 \,,
\end{equation}
with $c_1$ an integration constant that we set to zero to ensure that 
in the limit ${\tilde J} = \tilde{P} - \tilde{Q} = 0$ we obtain ${\psi} = 0$.

$B$ is determined by solving for the two-form
${\cal F} = d \left( B \, d\phi \right)$,
\begin{equation}
- e^{- 2 \phi_2} *{\cal F} = d \chi_3 - \chi_1 \, d \chi_2 \;.
\label{acc}
\end{equation}
Here, the dual $*$ is taken with respect to the three-dimensional metric 
\bea
ds_3^2 = e^{\psi} \left(d \rho^2 + dv^2 \right) + \rho^2 \, d \phi^2 \;.
\eea
Adjusting the integration constant, we obtain the following expression for $B$,
\begin{equation}
B= h_1 h_2 \frac{\left(\tilde{Q}\sqrt{ \rho^2 +(\tilde{P}+v)^2}  -\tilde{P}\sqrt{ \rho^2 +(\tilde{Q}+v)^2} \right) \left( v- \sqrt{\rho^2 +v^2} \right)  }{\rho^2 +v^2} + 
%f_1 \,,
h_1 h_2 \left(\tilde{Q}-\tilde{P}\right) \;.
\end{equation}
In the limit ${\tilde J} = \tilde{P} - \tilde{Q} = 0$ we obtain $B = 0$.

The above expressions completely determine the space-time metric \eqref{linePQg0} and the scalar field $e^{- 2 \Phi}$ that supports the solution.
The solution is further supported by an electric-magnetic field $F = d A^0$, where the one-form $A^0$
is
given by
\begin{equation}
A^0 = \chi_1 \, dt + A_{\phi} \, d \phi \;,
\end{equation}
with $A_{\phi}$ determined by solving
\begin{eqnarray}
\partial_{\rho} A_{\phi} &=& - e^{2(\Sigma_1 - \Sigma_2)} \, \rho \, \partial_{v} \chi_2
+ B\, \partial_{\rho} \chi_1 \;, \nonumber\\
\partial_{v} A_{\phi} &=& e^{2(\Sigma_1 - \Sigma_2)} \, \rho \, \partial_{\rho} \chi_2
+ B \, \partial_{v} \chi_1 \;.
\end{eqnarray}
In the following, we solve these partial differential equations by means of
a power series in ${\tilde J}$,
\begin{equation}
A_\phi= \sum_{n=0}^{\infty} A_{\phi}^{(n)}\frac{ {\tilde J}^n}{n !} ~\,.
\end{equation}
In \cite{Cardoso:2017cgi} $A_{\phi}$  was given up to first order in $\tilde{J}$ only. Here we determine $A_{\phi}$ up 
to second order in ${\tilde J}$,
\begin{align}
A_\phi^{(0)} &=   \frac{\sqrt{2} h_2 \tilde{Q} v }{\sqrt{\rho^2 +v^2}} \,,\nonumber\\
A_\phi^{(1)} &= \frac{h_2}{\sqrt{2}}\left(   \frac{v}{\sqrt{\rho^2 +v^2}} +   \frac{\left(v-\sqrt{\rho^2+v^2}\right) \sqrt{\rho^2 + (\tilde{Q}+v)^2}}{\rho^2+v^2+\tilde{Q} \sqrt{\rho^2+v^2}}  - \frac{2 \tilde{Q}}{\tilde{Q}+ \sqrt{\rho^2+v^2}} \right) \,, \nonumber\\
A_\phi^{(2)} &=  \frac{h_2 }{2 \sqrt{2}} \left(  \frac{2v^2 \left(\tilde{Q}+4v- \sqrt{\rho^2 + \left(\tilde{Q}+v\right)^2}\right)}{\tilde{Q}^2 \left( \tilde{Q} +2v \right) \sqrt{\rho^2+v^2}} - \frac{4 \left[\left(\tilde{Q}+v\right)^2   + \tilde{Q} \sqrt{\rho^2 + \left(\tilde{Q} +v\right)^2}\right]  }{\left(\tilde{Q} + \sqrt{\rho^2 +v^2  }\right)^3} \right. \nonumber\\
& \qquad \qquad  - \frac{4\left(2 \tilde{Q}^2 +v^2\right)+ \left(\tilde{Q}-v\right) \sqrt{\rho^2 + \left(\tilde{Q}+v\right)^2}}{\tilde{Q}^2 \left(\tilde{Q}+\sqrt{\rho^2+v^2}\right)}  \\
& \qquad \qquad +\frac{2 \left[-2v^2 +\left(3\tilde{Q} +v\right)\left(2 \tilde{Q}+\sqrt{\rho^2 + \left(\tilde{Q}+v\right)^2}\right)\right]}{\tilde{Q} \left(\tilde{Q}+\sqrt{\rho^2+v^2}\right)^2} \nonumber\\
&  \left. \qquad \qquad - \frac{\left(\tilde{Q}+v\right) \left[ \left(\tilde{Q}+2v-\sqrt{\rho^2+v^2}\right) \sqrt{\rho^2 + \left(\tilde{Q}+v\right)^2}   -2 \left(\tilde{Q}+v\right)  \sqrt{\rho^2+v^2} \right] }{\tilde{Q} \left(\tilde{Q}+2v\right) \left[\rho^2 + \left(\tilde{Q}+v\right)^2\right]} \right)  \,.\nonumber
\end{align}

%%%%%%%%%
\section{Gluing solutions along the lines $\rho = \pm v + m$  \label{sec:glue}}
%%%%%%%%%%%

The extensions \eqref{ext_jump} 
exhibit a jump in the transverse extrinsic curvature
across the lines  $\rho = \pm v + m$, as we now verify for the first extension in \eqref{ext_jump}. We consider the gluing of 
the solution based on 
$\Delta_i (\rho,v)$
in region $I$  with the solution based on  $\Delta_i (B^{-1} (\rho,v)) $ in region $II$ along the hypersurface given by $ \rho= m-v$.
We recall that demanding $\rho$ to be a time-like coordinate in both regions $I$ and $II$ requires
the transformation $ds_2^2 \rightarrow - ds_2^2$ when passing from region $I$ to region $II$.
Using the bijections given in Table \ref{tab:bijsph}, the associated space-time metrics read in spherical coordinates,
\bea
{g_I} &=&\left( \frac{2m}{r}-1\right) dt^2 - \frac{1}{\frac{2m}{r}-1} dr^2 + r^2 \left(d \theta^2 + \sin^2 \theta\, d \phi^2\right) \,, \nonumber\\
g_{II} &=& \tan^2 \left(\frac{\theta}{2}\right) dt^2 - 4 m^2 \cos^4 \left(\frac{\theta}{2}\right) \left( \frac{dr^2}{r\left(2m-r\right)} - d\theta^2 - \frac{r\left(2m-r\right)}{m^2}  d\phi^2 \right) \,.
\eea
The gluing of these solutions is performed along the 
hypersurface $\Sigma$ given  by $r=m+m\cos \theta$ with $0<\theta<\pi/2$. The metric is continuous at $\Sigma$, i.e.
 $g_{I\vert_{T\Sigma}} = g_{II\vert_{T\Sigma}}$. There is, however, a jump in the transverse extrinsic curvature
 \cite{Barrabes:1991ng,Poisson:2009pwt}, as follows.

To study the jump in the transverse  extrinsic curvature when traversing $\Sigma$,
we begin by considering the `normal' vector field $k$ to this hypersurface, evaluated on $\Sigma$,  given by
\be
k^\flat\vert_{\Sigma}= k_\mu dx^\mu \vert_{\Sigma}  =  \left( d\left(r-m - m \cos \theta\right) \right) \vert_{\Sigma}  
= \left( dr + \sqrt{(2m-r)r}  \, d\theta \right) \vert_{\Sigma} \;.
\ee
Since $||k ||_{{\Sigma}} =0$, the vector field $k$ is also tangent  to $\Sigma$, and so are the vector fields 
 $\partial_A$ with $A=t,\phi$. A
   null transverse vector field $\ell $, i.e. a vector field    
   satisfying $ || \ell ||_{\Sigma} =0, \, (\ell \cdot k)_{\Sigma} = -1, 
   \,    (\ell \cdot \partial_A)_{\Sigma} =0$, is then given by
\be
\ell= \ell^\mu \partial_\mu =
 -\frac{1}{2} \partial_r - \frac{1}{2 \sqrt{ \left(2m-r\right)r}}  \, \partial_\theta \,.
\ee
The jump in  the transverse  extrinsic curvature is computed using $\ell$  \cite{Barrabes:1991ng,Poisson:2009pwt},
and given by
\be
\frac{1}{2}\mathsterling_\ell g_{II\vert_{T\Sigma}} - \frac{1}{2}\mathsterling_\ell g_{I\vert_{T\Sigma}} = -\frac{m}{r} dt^2 + \frac{\left(2m-r\right)r^2}{m^2} d\phi^2\,,
\ee
which is non-vanishing. Here, $\mathsterling_\ell$ denotes the Lie derivative with respect to  the vector field $\ell$. 

 The jump in the transverse extrinsic curvature is 
 associated with the presence 
of a discontinuity of the Weyl tensor
\cite{Barrabes:1991ng,barrab};
across $\Sigma$, the metric is only continuous.

%%%%%%%%%
\section{$A$-metrics  \label{sec:classA}}
%%%%%%%%%%%

Let $m> 0$. 
The class of $A$-metrics comprises the following space-time metrics in four dimensions \cite{Ehlers:1962zz,Griffiths:2009dfa}:

\begin{enumerate}

\item $AI$-metrics: \\
the Schwarzschild solution (with $0 < r < 2m$ or $2m < r < \infty$, $0 \leq \phi < 2 \pi, \, 0 < \theta < \pi$)
\begin{equation}
ds_4^2 = - \left(1 - \frac{2m}{r} \right) dt^2 + \left(1 - \frac{2m}{r} \right)^{-1}  dr^2 + r^2 \left( d \theta^2 +  \sin^2 \theta \, d \phi^2 \right) \:,
\end{equation}
and the `negative mass' Schwarzschild solution (with  $0 < r < \infty $, $0 \leq \phi < 2 \pi, \, 0 < \theta < \pi$)
\begin{equation}
ds_4^2 = - \left(1 + \frac{2m}{r} \right) dt^2 + \left(1 + \frac{2m}{r} \right)^{-1}  dr^2 + r^2 \left( d \theta^2 +  \sin^2 \theta \, d \phi^2 \right) \:.
\end{equation}

\item $AII$-metrics:\\
the solution (with $0 < \varrho < 2m$ or $2m < \varrho < \infty$, $0 \leq \phi < 2 \pi, \, 0 < \vartheta < \infty $)
\begin{equation}
ds_4^2 =  \left(1 - \frac{2m}{\varrho} \right) dt^2-  \left(1 - \frac{2m}{\varrho} \right)^{-1}  d\varrho^2 + \varrho^2 \left( d \vartheta^2 + \sinh^2 \vartheta \, d \phi^2 \right) \:,
\label{hyp}
\end{equation}
and the solution (with $0 < \varrho < \infty$, $0 \leq \phi < 2 \pi, \, 0 < \vartheta < \infty $)
\begin{equation}
ds_4^2 =  \left(1 + \frac{2m}{\varrho} \right) dt^2-  \left(1 + \frac{2m}{\varrho} \right)^{-1}  d\varrho^2 + \varrho^2 \left( d \vartheta^2 + \sinh^2 \vartheta \, d \phi^2 \right) \:.
\label{hypneg}
\end{equation}

\item $AIII$-metrics: \\
the $r$ space-like solution (with   $m>0$, $0 < r < \infty$, $0 \leq \phi < 2 \pi, \, 0 < \rho < \infty$)
\begin{equation} \label{aiii-r-space-like}
ds_4^2 = - \frac{2m}{r} dt^2 +  \frac{r}{2m}   dr^2 + r^2 \left( d \rho^2 +  \rho^2  \, d \phi^2 \right) \:,
\end{equation}
and the $r$ time-like solution
\begin{equation} \label{aiii-r-time-like}
ds_4^2 =  \frac{2m}{r} dt^2 -  \frac{r}{2m}   dr^2 + r^2 \left( d \rho^2 +  \rho^2  \, d \phi^2 \right) \:.
\end{equation}
These solutions can be transformed, respectively, to the solution (with $0 < x_2, \, -\infty < t, x_1, x_3 <  \infty$)
\be
 ds^2_4 = -x_2^{-2/3 }  dt^2 + dx_2^2 + x_2^{4/3} \, \left( dx_1^2 + dx_3^2 \right) \;,
\label{AIIIst2}
\ee
and to the Kasner solution (with $0 < t, \, -\infty < x_1, x_2, x_3 <  \infty$)
\be
 ds^2_4 = - dt^2 + t^{-2/3 } \, dx_2^2 + t^{4/3} \, \left( dx_1^2 + dx_3^2 \right) \;.
 \label{AIIIst1}
\ee

\end{enumerate}

\noindent
These space-time metrics arise as follows from factorization.\\

\noindent
{\bf $AI, AII$-metrics:}
the solution describing the exterior region of the Schwarzschild black hole ($r > 2m$),
the `negative mass' Schwarzschild solution and the solution describing the interior region of \eqref{hyp} ($0 < \varrho < 2m$)
are obtained by canonical factorization of the Schwarzschild monodromy matrix \eqref{monschwarz} with $\sigma = 1$.
The solution describing the interior region of the Schwarzschild black hole ($0 < r <  2m$),
the solution describing the exterior region of \eqref{hyp} ($\varrho > 2m$) and the solution \eqref{hypneg}
are obtained by canonical factorization of the Schwarzschild monodromy matrix \eqref{monschwarz} with $\sigma = -1$.\\

\noindent
{\bf $AIII$-metrics:}
the space-time metrics \eqref{aiii-r-space-like}  and \eqref{aiii-r-time-like} are obtained by canonical factorization of the monodromy matrix 
\eqref{monschwarz0} with $\sigma = 1$ and $\sigma = -1$, respectively.

%\bibliographystyle{JHEP}

%\bibliography{biblio}

\begin{thebibliography}{10}

\bibitem{Alekseev:2010mx}
G.~A. Alekseev, {\it {Thirty years of studies of integrable reductions of
  Einstein's field equations}},  in {\em {On recent developments in theoretical
  and experimental general relativity, astrophysics and relativistic field
  theories. Proceedings, 12th Marcel Grossmann Meeting on General Relativity,
  Paris, France, July 12-18, 2009. Vol. 1-3}}, pp.~645--666, 2010.
\newblock \href{http://xxx.lanl.gov/abs/1011.3846}{{\tt arXiv:1011.3846}}.



\bibitem{Breitenlohner:1986um}
P.~Breitenlohner and D.~Maison, \emph{{On the Geroch Group}}, {\emph{Ann. Inst.
  H. Poincare Phys. Theor.} {\bf 46} (1987) 215}.

\bibitem{Nicolai:1991tt}
H.~Nicolai, \emph{{Two-dimensional gravities and supergravities as integrable
  system}}, \href{http://dx.doi.org/10.1007/3-540-54978-1_12}{\emph{Lect. Notes
  Phys.} {\bf 396} (1991) 231--273}.

\bibitem{Katsimpouri:2012ky}
D.~Katsimpouri, A.~Kleinschmidt and A.~Virmani, \emph{{Inverse Scattering and
  the Geroch Group}},
  \href{http://dx.doi.org/10.1007/JHEP02(2013)011}{\emph{JHEP} {\bf 02} (2013)
  011}, [\href{http://arxiv.org/abs/1211.3044}{{\tt 1211.3044}}].

\bibitem{Camara:2017hez}
M.~C. C\^amara, G.~L. Cardoso, T.~Mohaupt and S.~Nampuri, \emph{{A
  Riemann-Hilbert approach to rotating attractors}},
  \href{http://dx.doi.org/10.1007/JHEP06(2017)123}{\emph{JHEP} {\bf 06} (2017)
  123}, [\href{http://arxiv.org/abs/1703.10366}{{\tt 1703.10366}}].

\bibitem{Cardoso:2017cgi}
G.~L. Cardoso and J.~C. Serra, \emph{{New gravitational solutions via a
  Riemann-Hilbert approach}},
  \href{http://dx.doi.org/10.1007/JHEP03(2018)080}{\emph{JHEP} {\bf 03} (2018)
  080}, [\href{http://arxiv.org/abs/1711.01113}{{\tt 1711.01113}}].
  
  
 \bibitem{Ehlers:1962zz}
J.~Ehlers and W.~Kundt, \emph{{Exact solutions of the gravitational field
  equations}},  in \emph{{Gravitation: an introduction to current research, ed.
  L. Witten, (Wiley), 49-101}}, 1962.
  
\bibitem{Griffiths:2009dfa}
J.~B. Griffiths and J.~Podolsky, \emph{{Exact Space-Times in Einstein's General
  Relativity}}.
\newblock Cambridge Monographs on Mathematical Physics. Cambridge University
  Press, Cambridge, 2009,
  \href{http://dx.doi.org/10.1017/CBO9780511635397}{10.1017/CBO9780511635397}.

  

\bibitem{WH}
N.~Wiener and E.~Hopf, \emph{{\"Uber eine Klasse singul\"arer
  Integralgleichungen}}, {\emph{S.-B. Preuss. Akad. Wiss. Berlin, Phys.-Math.
  Kl. 30/32} (1931) 696--706}.

\bibitem{CG}
K.~Clancey and I.~Gohberg, \emph{{Factorization of Matrix Functions and
  Singular Integral Operators}},  in \emph{{Operator Theory: Advances and
  Applications, Vol. 3, Birkh\"auser Verlag, Basel}}, 1981.

\bibitem{Speck}
F.-O. Speck, \emph{{General Wiener-Hopf Factorization Methods}},
  {\emph{Research Notes in Mathematics} {\bf 119} (1985) 1}.



\bibitem{Barrabes:1991ng}
C.~Barrabes and W.~Israel, \emph{{Thin shells in general relativity and
  cosmology: The Lightlike limit}},
  \href{http://dx.doi.org/10.1103/PhysRevD.43.1129}{\emph{Phys. Rev.} {\bf D43}
  (1991) 1129--1142}.

\bibitem{barrab}
C.~Barrabes and P.~Hogan, \emph{{Singular Null Hypersurfaces in General
  Relativity}}.
\newblock World Scientific, Singapore, 2003.

\bibitem{Jones:2004pz}
G.~Jones and J.~E. Wang, \emph{{Weyl card diagrams and new S-brane solutions of
  gravity}},  \href{http://arxiv.org/abs/hep-th/0409070}{{\tt hep-th/0409070}}.

\bibitem{Jones:2005hj}
G.~C. Jones and J.~E. Wang, \emph{{Weyl card diagrams}},
  \href{http://dx.doi.org/10.1103/PhysRevD.71.124019}{\emph{Phys. Rev.} {\bf
  D71} (2005) 124019}, [\href{http://arxiv.org/abs/hep-th/0506023}{{\tt
  hep-th/0506023}}].

\bibitem{Schwarz:1995af}
J.~H. Schwarz, \emph{{Classical symmetries of some two-dimensional models
  coupled to gravity}},
  \href{http://dx.doi.org/10.1016/0550-3213(95)00455-2}{\emph{Nucl. Phys.} {\bf
  B454} (1995) 427--448}, [\href{http://arxiv.org/abs/hep-th/9506076}{{\tt
  hep-th/9506076}}].

\bibitem{Lu:2007jc}
H.~Lu, M.~J. Perry and C.~N. Pope, \emph{{Infinite-dimensional symmetries of
  two-dimensional coset models coupled to gravity}},
  \href{http://dx.doi.org/10.1016/j.nuclphysb.2008.07.035}{\emph{Nucl. Phys.}
  {\bf B806} (2009) 656--683}, [\href{http://arxiv.org/abs/0712.0615}{{\tt
  0712.0615}}].

\bibitem{Its}
A.~R. Its, \emph{{The Riemann-Hilbert Problem and Integrable Systems}},
  {\emph{Notices of the AMS} {\bf 50} (2003) 1389}.

\bibitem{LS}
G.~Litvinchuk and I.~Spitkovsky, \emph{{Factorization of Measurable Matrix
  Functions}},  in \emph{{Oper. Theory Adv. Appl., vol. 25, Birkh\"auser
  Verlag, Basel, 1987. Translated from Russian by B. Luderer, with a foreword
  by B. Silbermann}}, 1987.

\bibitem{MP}
S.~Mikhlin and S.~Pr\"ossdorf, \emph{{Singular integral operators}},  in
  \emph{{Springer-Verlag, Berlin, 1986. Translated from German by Albrecht
  B\"ottcher and Reinhard Lehmann}}, 1986.

\bibitem{CLS}
M.~C. C\^amara, A.~B. Lebre and F.-O. Speck, \emph{{Meromorphic factorization,
  partial index estimates and elastodynamic difraction problems}},
  \href{http://dx.doi.org/doi:10.1002/mana.19921570124}{\emph{Math. Nachr.}
  {\bf 157} (1992) 291--317}.



\bibitem{CC}
M.~C. C\^amara, \emph{Toeplitz operators and Wiener-Hopf factorisation: an
  introduction},
  \href{http://dx.doi.org/10.1515/conop-2017-0010}{\emph{Concrete Operators}
  {\bf 4} (10, 2017) }.




\bibitem{Brill}
D. R.~Brill, \emph{{Electromagnetic Fields in a Homogeneous, Nonisotropic Universe}}, {\emph{Phys. Rev.} {\bf 133} (1964) B845}.



\bibitem{Stephani:2003tm}
H.~Stephani, D.~Kramer, M.~A.~H. MacCallum, C.~Hoenselaers and E.~Herlt,
  \emph{{Exact solutions of Einstein's field equations}}, Cambridge Monographs
  on Mathematical Physics. Cambridge Univ. Press, Cambridge, 2003,
  \href{https://doi.org/10.1017/CBO9780511535185}{10.1017/CBO9780511535185}.



\bibitem{Poisson:2009pwt}
E.~Poisson, \emph{{A Relativist's Toolkit: The Mathematics of Black-Hole
  Mechanics}}.
\newblock Cambridge University Press, 2009,
  \href{http://dx.doi.org/10.1017/CBO9780511606601}{10.1017/CBO9780511606601}.



\end{thebibliography}

\providecommand{\href}[2]{#2}\begingroup\raggedright\endgroup

\end{document}